\documentclass[11pt, letterpaper]{article}
\pdfoutput=1

\usepackage{amsthm, amsmath, amssymb, natbib, bbm, stmaryrd, algorithm, float}
\usepackage{epsfig}
\usepackage{color, footmisc}
\RequirePackage[colorlinks,citecolor=blue,urlcolor=blue]{hyperref}
\usepackage{bbm}

\makeatletter
\def\singlespace{\def\baselinestretch{1}\@normalsize}

\makeatletter
\def\singlespace{\def\baselinestretch{1}\@normalsize}

\theoremstyle{plain}
\newtheorem{thm}{Theorem}[section]
\newtheorem{cor}{Corollary}[section]
\newtheorem{rem}{Remark}[section]
\newtheorem{lem}{Lemma}[section]
\newtheorem{ass}{Assumption}[section]

\newcommand{\bfm}[1]{\ensuremath{\mathbf{#1}}}
\def\ba{\bfm a}     \def\bA{\bfm A}          
\def\bb{\bfm b}     \def\bB{\bfm B}          
\def\bc{\bfm c}               
     \def\bD{\bfm D}          
\def\be{\bfm e}     \def\bE{\bfm E}     \def\cE{{\cal  E}}     
\def\bff{\bfm f}    \def\bF{\bfm F}

     \def\bI{\bfm I}

     \def\bM{\bfm M}          
          \def\cN{{\cal  N}}     
               
\def\bp{\bfm p}     \def\bP{\bfm P}          
     \def\bQ{\bfm Q}          
     \def\bR{\bfm R}          
     \def\bS{\bfm S}     \def\cS{{\cal  S}}     
               
\def\bu{\bfm u}     \def\bU{\bfm U}          
\def\bv{\bfm v}     \def\bV{\bfm V}          
\def\bw{\bfm w}     \def\bW{\bfm W}          
\def\bx{\bfm x}     \def\bX{\bfm X}          
\def\by{\bfm y}               
\def\bz{\bfm z}               
\def\bzero{\bfm 0}

\newcommand{\bfsym}[1]{\ensuremath{\boldsymbol{#1}}}
\newcommand{\sam}{\mathrm{sam}}

\def \balpha   {\bfsym{\alpha}}       \def \bbeta    {\bfsym{\beta}}
\def \bgamma   {\bfsym{\gamma}}       
\def \bvarepsilon {\bfsym{\varepsilon}}
\def \bvarepsilon {\bfsym{\varepsilon}}

         \def \btheta   {\bfsym{\theta}}
        
      \def \bmu      {\bfsym{\mu}}
          \def \bxi      {\bfsym{\xi}}

\def \bTheta   {\bfsym{\Theta}}       
          
\def \bSigma   {\bfsym{\Sigma}}

\renewcommand{\hat}{\widehat}

\def \heps     {\hat{\heps}}

\DeclareMathOperator*{\argmin}{argmin}

\DeclareMathOperator{\cov}{cov}

\DeclareMathOperator{\diag}{diag}
\DeclareMathOperator{\E}{E}

\DeclareMathOperator{\rank}{rank}

\DeclareMathOperator{\sgn}{sgn}
\DeclareMathOperator{\supp}{supp}

\DeclareMathOperator{\sr}{sr}

\def \sgn   {\mbox{sgn}}

\newcommand{\row}{\mathrm{row}}
\newcommand{\coln}{\mathrm{col}}
\newcommand{\xx}{\text{\boldmath $x$}}
\newcommand{\yy}{\text{\boldmath $y$}}

\newcounter{CondCounter}
\newcommand{\ltwonorm}[1]{\lVert#1\rVert_2}

\newcommand{\opnorm}[1]{\lVert#1\rVert_{2}}
\newcommand{\fnorm}[1]{\lVert#1\rVert_F}

\def \RR	{\mathbb{R}}
\newcommand{\R}{\mathbb{R}}
\def \PP {\mathbb{P}}

\def \bone {{\bf 1}}

\def\R{{\mathbb R}}

\def\E{{\mathbb E}}

\def\P{{\mathbb P}}
\def\diag{{\rm diag}}

\def\cov{{\rm cov}}

\def\supp{{\rm supp}}

\def\argmax{{\rm argmax}}
\def\argmin{{\rm argmin}}

\newcommand{\beq}  {\begin{equation}}
\newcommand{\eeq}  {\end{equation}}

\renewcommand{\tilde}{\widetilde}

\begin{document}

\title{\bf Robust high dimensional factor models with applications to statistical machine learning\footnote{The authors gratefully acknowledge \textit{NSF grants DMS-1712591 and DMS-1662139 and NIH grant R01-GM072611.} }}

\author{Jianqing Fan\footnote{Department of Operations Research and Financial Engineering, Sherrerd Hall, Princeton University, NJ 08544, USA (Email: jqfan@princeton.edu, kaizheng@princeton.edu, yiqiaoz@princeton.edu, zzw9348ustc@gmail.com)} \footnote{Fudan University} , Kaizheng Wang$^{\dag}$, Yiqiao Zhong$^{\dag}$, Ziwei Zhu$^{\dag}$}
\maketitle

\begin{abstract}
Factor models are a class of powerful statistical models that have been widely used to deal with dependent measurements that arise frequently from various applications from genomics and neuroscience to economics and finance.  As data are collected at an ever-growing scale, statistical machine learning faces some new challenges: high dimensionality, strong dependence among observed variables, heavy-tailed variables and heterogeneity. High-dimensional robust factor analysis serves as a powerful toolkit to conquer these challenges.

This paper gives a selective overview on recent advance on high-dimensional factor models and their applications to statistics including Factor-Adjusted Robust Model selection (FarmSelect) and Factor-Adjusted Robust Multiple testing (FarmTest). We show that classical methods, especially principal component analysis (PCA), can be tailored to many new problems and provide powerful tools for statistical estimation and inference. We highlight PCA and its connections to matrix perturbation theory, robust statistics, random projection, false discovery rate, etc., and illustrate through several applications how insights from these fields yield solutions to modern challenges. We also present far-reaching connections between factor models and popular statistical learning problems, including network analysis and low-rank matrix recovery.

\end{abstract}


\sloppy

\section{Introduction}

In modern data analytics, dependence across high-dimensional outcomes or measurements is ubiquitous. For example, stocks within the same industry exhibit significantly correlated returns, housing prices of a country depend on various economic factors, gene expressions can be stimulated by cytokines. Ignoring such dependence structure can produce significant systematic bias and yields inefficient statistical results and misleading insights. The problems are more severe for high-dimensional big data, where dependence, non-Gaussianity and heterogeneity of measurements are common.

Factor models aim to capture such dependence by assuming several variates or ``factors'', usually much fewer than the outcomes, that drive the dependence of the entire outcomes \citep{LMa62,SWa02}. Stemming from the early works on measuring human abilities \citep{Spe27}, factor models have become one of the most popular and powerful tools in multivariate analysis  and have made profound impact in the past century on psychology \citep{Bar38,MJo92}, economics and finance \citep{CRo82, FFr93, SWa02, BNg02}, biology \citep{HHC02, HCO06, LSt08}, etc. Suppose $\bx_1, \ldots, \bx_n$ are $n$ i.i.d. $p$-dimensional random vectors, which may represent financial returns, housing prices, gene expressions, etc. The generic factor model assumes that
\beq
	\label{fm}
	\begin{aligned}
		\bx_i = \bmu + \bB \bff_i + \bu_i, \quad  \textnormal{or in matrix form, } \quad \bX = \bmu \bone^\top_n + \bB\bF^\top + \bU,
    	\end{aligned}
\eeq
where $\bX = (\bx_1, \ldots, \bx_n) \in \RR^{p \times n}$, $\bmu=(\mu_1,\ldots, \mu_p)^\top$ is the mean vector, $\bB = (\bb_1, \ldots, \bb_p)^\top \in \RR^{p \times K}$ is the matrix of factor loadings, $\bF = (\bff_1, \ldots, \bff_n)^\top \in \RR^{n \times K}$ stores $K$-dimensional vectors of common factors with $\E \bff_i = \bzero $, and $\bU = (\bu_1, \ldots, \bu_n) \in \RR^{p \times n}$ represents the error terms (a.k.a.\ idiosyncratic components), which has mean zero and is uncorrelated with or independent of $\bF$. We emphasize that, for most of our discussions in the paper (except Section~\ref{sec-cov-est}), only $\{\bx_i\}_{i=1}^n$ are observable, and the goal is to infer $\bB$ and $\{\bff_i\}_{i=1}^n$ through $\{\bx_i\}_{i=1}^n$. Here we use the name ``factor model'' to refer to a general concept where the idiosyncratic components $\bu_i$ are allowed to be weakly correlated. This is also known as the ``approximate factor model'' in the literature, in contrast to the ``strict factor model'' where the idiosyncratic components are assumed to be uncorrelated.

Note that the model \eqref{fm} has identifiability issues: given any invertible matrix $\bR \in \RR^{K \times K}$, simultaneously replacing $\bB$ with $\bB\bR$ and $\bff_i$ with $\bR^{-1}\bff_i$  does not change the observation $\bx_i$. To resolve this ambiguity issue, the following identifiability assumption is usually imposed:
\begin{ass}[Identifiability]
\label{ass:1}
$\bB^\top\bB$ is diagonal and $\cov(\bff_i) = \bI_p$.
\end{ass}
\noindent Other identifiability assumptions as well as detailed discussions can be found in \cite{BLi12} and \cite{FLM13}.

Factor analysis is closely related to principal component analysis (PCA), which breaks down the covariance matrix into a set of orthogonal components and identifies the subspace that explains the most variation of the data \citep{Pea01,Hot33}. In this selective review, we will mainly leverage PCA, or more generally, spectral methods, to estimate the factors $\{\bff_i\}_{i=1}^n$ and the loading matrix $\bB$ in \eqref{fm}. Other popular estimators, mostly based on the maximum likelihood principle, can be found in \cite{LMa62,AAm88,BLi12}, etc.
The covariance matrix of $\bx_i$ consists of two components: $\cov(\bB \bff_i)$ and $\cov(\bu_i)$.
Intuitively, when the contribution of the covariance from the error term $\bu_i$ is negligible compared with those from the factor term $\bB\bff_i$, the top-$K$ eigenspace (namely, the space spanned by top $K$ eigenvectors) of the sample covariance of $\{\bx_i\}_{i=1}^n$ should be well aligned with the column space of $\bB$.  This can be seen from the assumption that $\cov(\bx_i) = \bB  \bB^\top + \cov(\bu_i) \approx \bB  \bB^\top$, which occurs frequently in high-dimensional statistics \citep{FLM13}.

Here is our main message: applying PCA to well-crafted covariance matrices (including vanilla sample covariance matrices and their robust version) consistently estimates the factors and loadings, as long as the signal-to-noise ratio is large enough. The core theoretical challenge is to characterize how idiosyncratic covariance $\cov(\bu_i) $ perturb the eigenstructure of the factor covariance $\bB  \bB^\top$. In addition, the situation is more complicated with the presence of heavy-tailed data, missing data, computational constraints, heterogeneity, etc.

The rest of the paper is devoted to solutions to these challenges and a wide range of applications to statistical machine learning problems. In Section \ref{sec:2}, we will elucidate the relationship between factor models and PCA and present several useful deterministic perturbation bounds for eigenspaces. We will also discuss robust covariance inputs for the PCA procedure to guard against corruption from heavy-tailed data. Exploiting the factor structure of the data helps solve many statistical and machine learning problems. In Section \ref{sec:3}, we will see how the factor models and PCA can be applied to high-dimensional covariance estimation, regression, multiple testing and model selection. In Section \ref{sec:4}, we demonstrate the connection between PCA and a wide range of machine learning problems including Gaussian mixture models, community detection, matrix completion, etc. We will develop useful tools and establish strong theoretical guarantees for our proposed methods.

Here we collect all the notations for future convenience. We use $[m]$ to refer to $\{ 1,2,\ldots,m \}$. We adopt the convention of using regular letters for scalars and using bold-face letters for vectors or matrices. For $\bx = (x_1,\ldots,x_p)^\top \in \RR^p$, and $1 \leq q < \infty$, we define $\|\bx\|_{q} = \bigl(\sum_{j=1}^{p}|x_{j}|^{q}\bigr)^{1/q}$, $\|\bx\|_{0}=|\supp(\bx)|$, where $\supp(\bx) = \{j : x_{j}\neq 0\}$, and $\|\bx\|_{\infty} = \max_{1\leq j\leq p} |x_{j}|$. For a matrix $\bM$, we use $\opnorm{\bM}, \fnorm{\bM}, \|\bM\|_{\max}$ and $\| \bM \|_1$ to denote its operator norm (spectral norm), Frobenius norm, entry-wise (element-wise) max-norm, and vector $\ell_1$ norm, respectively. To be more specific, the last two norms are defined by $ \| \bM \|_{\max} = \max_{j,k} |M_{jk}|$ and $\|\bM\|_{1}=\sum_{j,k}|M_{jk}|$. Let $\bI_p$ denote the $p\times p$ identity matrix, $\mathbf{1}_{p}$ denote the $p$-dimensional all-one vector, and $\mathbbm{1}_{A}$ denote the indicator of event $A$, i.e., $\mathbbm{1}_A=1$ if $A$ happens, and $0$ otherwise. We use $\mathcal{N}(\bmu,\bSigma)$ to refer to the normal distribution with mean vector $\bmu$ and covariance matrix $\bSigma$. For two nonnegative numbers $a$ and $b$ that possibly depend on $n$ and $p$, we use the notation $a = O(b)$ and $a \lesssim b$ to mean $a \le C_1b$ for some constant $C_1>0$, and the notation $a = \Omega(b)$ and $a \gtrsim b$ to mean $a \ge C_2 b$ for some constant $C_2>0$. We write $a \asymp b$ if both $a = O(b)$ and $a = \Omega(b)$ hold.  For a sequence of random variables $\{X_n\}_{n=1}^{\infty}$ and a sequence of nonnegative deterministic numbers $\{a_n\}_{n=1}^{\infty}$, we write $X_n=O_{\mathbb{P}}(a_n)$ if for any $\varepsilon>0$, there exists $C>0$ and $N>0$ such that $\mathbb{P}(|X_n|\geq Ca_n )\leq \varepsilon$ holds for all $n>N$; and we write $X_n=o_{\mathbb{P}}(a_n)$ if for any $\varepsilon>0$ and $C>0$, there exists $N>0$ such that $\mathbb{P}(|X_n|\geq Ca_n )\leq \varepsilon$ holds for all $n>N$. We omit the subscripts when it does not cause confusion.

\section{Factor models and PCA}
\label{sec:2}

\subsection{Relationship between PCA and factor models in high dimensions}\label{sec:2.1}

Under model \eqref{fm} with the identifiability condition, $\bSigma=\cov(\bx_i)$ is given by
\beq
	\bSigma = \bB  \bB^\top + \bSigma_{u} , \quad \bSigma_{u} = (\sigma_{u, jk})_{1\leq j,k\leq p} = \cov(\bu_i). \label{cov.fm}
\eeq
Intuitively, if the magnitude of $\bB \bB^\top$ dominates $\bSigma_u$, the top-$K$ eigenspace of $\bSigma$ should be approximately aligned with the column space of $\bB$. Naturally we expect a large gap between the eigenvalues of $\bB \bB^\top$ and $\bSigma_u$ to be important for estimating the column space of $\bB$ through PCA (see Figure~\ref{fig:1}). On the other hand, if this gap is small compared with the eigenvalues of $\bSigma_u$, it is known that PCA leads to inconsistent estimation \citep{johnstone2009consistency}. The above discussion motivates a simple vanilla PCA-based method for estimating $\bfm B$ and $\bfm F$ as follows (assuming the Identifiability Assumption).

\vspace{0.05in}

{\it Step 1}. Obtain an estimator $\hat{\bmu}$ and $\hat{\bfm \Sigma}$ of $\bmu$ and $\bfm \Sigma$, e.g.,\ the sample mean and covariance matrix or their robust versions.

\vspace{0.05in}

{\it Step 2}. Compute the eigen-decomposition of $\hat{\bfm \Sigma}=\sum_{j=1}^p \hat{\lambda}_j \hat{\bfm v}_j \hat{\bfm v}_j^\top$. Let $\{ \hat{\lambda}_k \}_{k=1}^K$ be the top $K$ eigenvalues and $\{\hat{\bfm v}_k \}_{k=1}^K$ be their corresponding eigenvectors. Set $\hat{\bfm V} = (\hat{\bfm v}_1,\ldots, \hat{\bfm v}_K) \in \R^{p \times K}$ and $\hat{\bfm \Lambda} = \diag(\hat{\lambda}_1, \ldots, \hat{\lambda}_K) \in \R^{K \times K}$.

\vspace{0.05in}

{\it Step 3}. Obtain PCA estimators $\hat{\bfm B} = \hat{\bfm V} \hat{\bfm \Lambda}^{1/2}$ and $\hat{\bfm F} = (\bfm X - \hat\bmu \bone^\top)^\top \hat{\bfm V} \hat{\bfm \Lambda}^{-1/2}$, namely, $\hat{\bfm B}$ consists of the top-$K$ rescaled eigenvectors of $\hat{\bSigma}$ and $\hat{\bff_i}$ is just the rescaled projection of $\bx_i-\hat{\bmu}$ onto the space spanned by the eigen-space:  $\hat{\bff_i} = \hat{\bfm \Lambda}^{-1/2} \hat{\bfm V}^T (\bx_i-\hat{\bmu})$.
\vspace*{-0.05in}

% subfigure
\begin{figure}[h]

	\begin{tabular}{cc}
     \includegraphics[width=.4\textwidth]{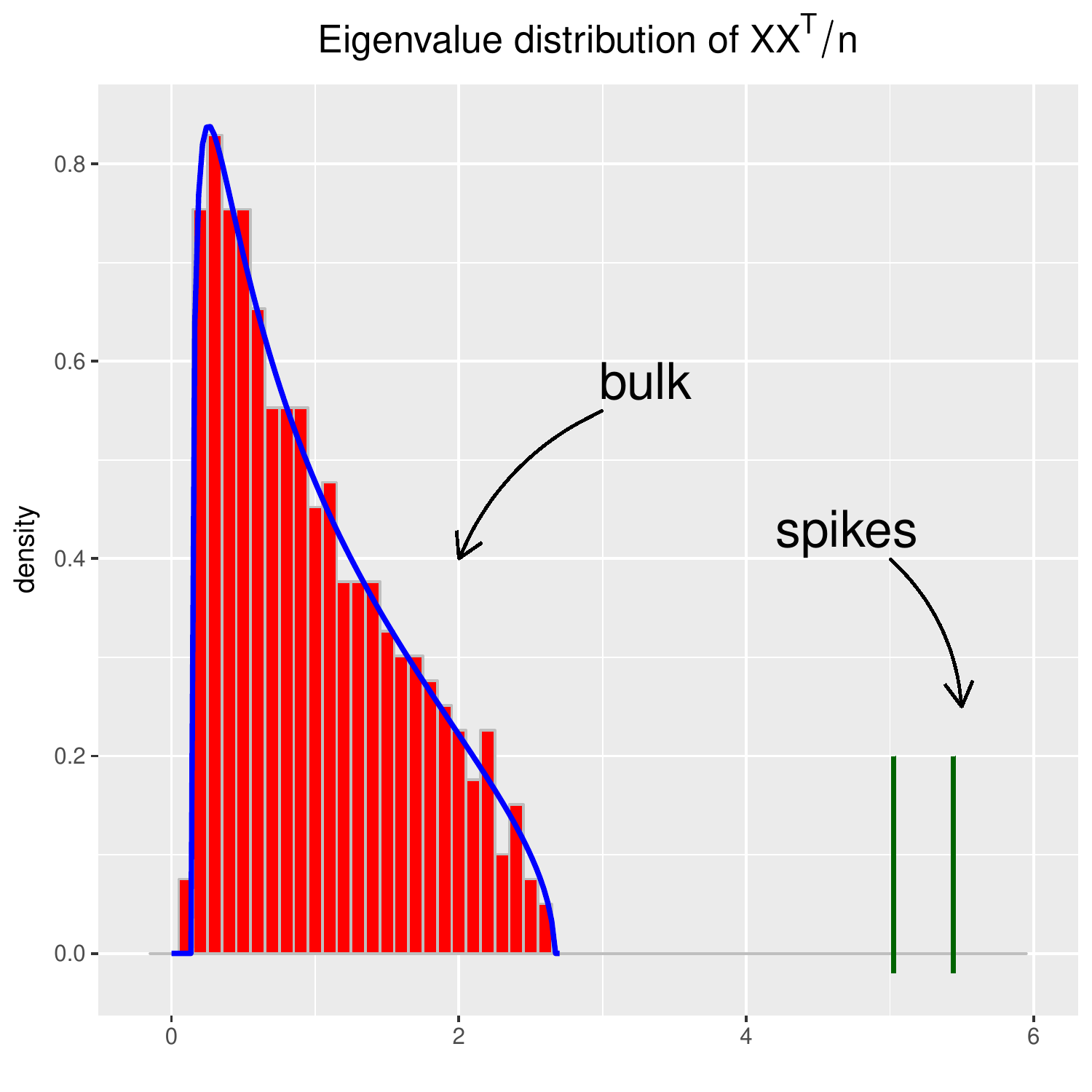} &  \includegraphics[trim = 0 -4cm 0 -4cm, clip, width=.47\textwidth]{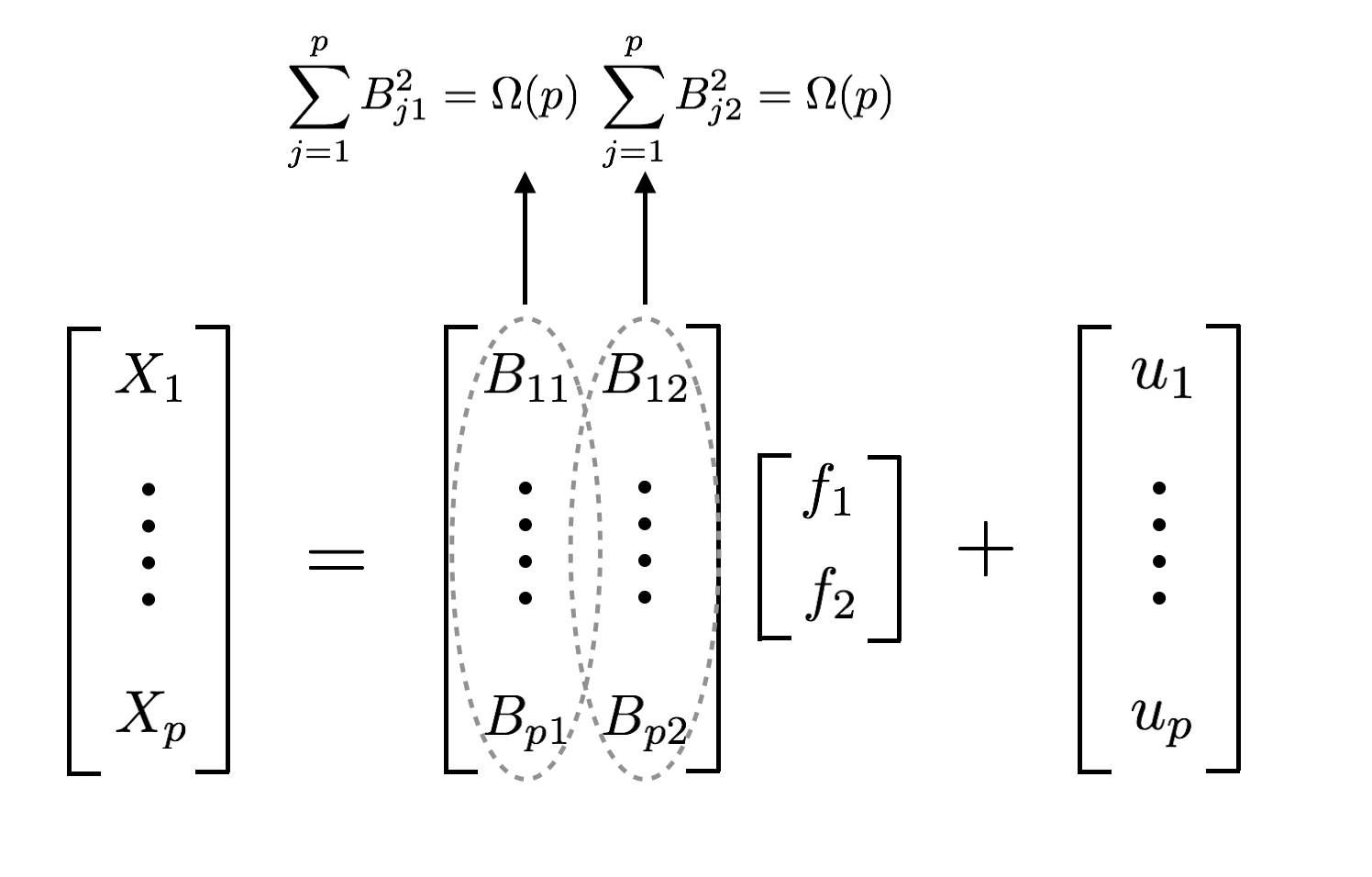}
    \end{tabular}
    \caption{The {\bf left panel} is the histogram of the eigenvalue distribution from a synthetic dataset. Fix $n= 1000$, $p=400$ and $K=2$ and let all the entries of $\bfm B$ be i.i.d. Gaussian $\cN(0,1/4)$. Each entry of $\bfm F$ and $\bfm U$ is generated from i.i.d.\ $\cN(0,1)$ and i.i.d.\ $\cN(0, 5^2)$ respectively. The data matrix $\bX$ is formed according to the factor model \eqref{fm}. The {\bf right diagram} illustrates the Pervasiveness Assumption. }
    	\label{fig:1}
\end{figure}
\vspace*{-0.05in}

Let us provide some intuitions for the estimators in Step 3. Recall that $\bfm b_j$ is the $j$th column of $\bB$.  Then, by model \eqref{fm}, $\bB^\top (\bx_i-\bmu) = \bB^\top\bB \bff_i + \bB^\top \bu_i$.  In the high-dimensional setting, the second term is averaged out when $\bu_i$ is weakly dependent across its component.  This along with the identifiability condition delivers that
\begin{equation}\label{eq2.2}
  \bff_i \approx \diag(\bB^\top \bB)^{-1} \bB^\top (\bx_i - \bmu) = \diag(\|\bfm b_1\|^2, \cdots, \|\bfm b_K\|^2)^{-1} \bB^\top (\bx_i-\bmu).
\end{equation}
Now, we estimate  $\bB \bB^\top$ by $\sum_{j=1}^K \hat{\lambda}_j \hat{\bfm v}_j \hat{\bfm v}_j^\top$ and hence $\bfm b_j$ by $\hat{\lambda}_j^{1/2} \hat{\bfm v}_j$ and
$\| \bfm b_j \|^2$ by $\hat{\lambda}_j$.  Using the substitution method, we obtain the estimators in Step 3.

The above heuristic also reveals that
the PCA-based methods work well if the effect of the factors outweighs the noise. To quantify this, we introduce a form of Pervasiveness Assumption from the factor model literature. While this assumption is strong\footnote{There is a weaker assumption, under which (\ref{fm}) is usually called the weak factor model; see \cite{Onatski12}.}, it simplifies our discussion and captures the above intuition well: it holds when the factor loadings $\{\bb_j\}_{j=1}^p$ are random samples from a nondegenerate population \citep{FLM13}.
%\begin{ass}[Pervasiveness]\label{ass:2}
%	The first $K$ eigenvalues of $\bB \allowbreak \cov(\bff_i ) \bB^\top$ have order $\Omega(p)$, whereas all the eigenvalues of $\bSigma_u$ are bounded by some constant.
%\end{ass}

\begin{ass}[Pervasiveness]\label{ass:2}
	The first $K$ eigenvalues of $\bB \allowbreak \bB^\top$ have order $\Omega(p)$, whereas $\| \bSigma_u \|_2 = O(1)$.
\end{ass}

Note that $\cov(\bff_i ) = \bI_K$ under the Identifiability Assumption \ref{ass:1}.
The first part of this assumption holds when each factor influences a non-vanishing proportion of outcomes. Mathematically speaking, it means that for any $k \in [K] := \{1,2,\ldots,K\}$, the average of squared loadings of the $k$th factor satisfies $p^{-1}\sum_{j = 1}^p B_{jk}^2 = \Omega(1)$ (right panel of Figure~\ref{fig:1}). This holds with high probability if, for example, $\{B_{jk}\}_{j=1}^p$ are i.i.d.\  realizations from a non-degenerate distribution, but we will not make such assumption in this paper. The second part of the assumption is reasonable, as cross-sectional correlation becomes weak after we take out the common factors. Typically, if $ \bfm \Sigma_u$ is a sparse matrix, the norm bound $\| \bSigma_u \|_2 = O(1)$ holds; see Section~\ref{sec-cov-est} for details. Under this Pervasiveness Assumption, the first $K$ eigenvalues of $\bSigma$ will be well separated with the rest of eigenvalues. By the Davis-Kahan theorem \citep{DKa70}, which we present as Theorem 2.1, we can consistently estimate the column space of $\bB$ through the top-$K$ eigenspace of $\bSigma$. This explains why we can apply PCA to factor model analysis \citep{FLM13}.

Though factor models and PCA are not identical (see \citealp{Jol86}), they are approximately the same for high-dimensional problems with the pervasiveness assumption\citep{FLM13}.  Thus, PCA-based ideas are important components of estimation and inference for factor models. In later sections (especially Section 4), we discuss statistical and machine learning problems with factor-model-type structures. There PCA is able to achieve consistent estimation even when the Pervasiveness Assumption is weakened---and somewhat surprisingly---PCA can work well down to the information limit. For perspectives from random matrix theory, see \cite{BBP05, Paul07, johnstone2009consistency, BenFloRaj11, OVuWan16, WFa17}, among others.

\subsection{Estimating the number of factors}\label{sec-est-K}

In high-dimensional factor models, if the factors are unobserved, we need to choose the number of factors $K$ before estimating the loading matrix, factors, etc. The number $K$ can be usually estimated from the eigenvalues of the the sample covariance matrix or its robust version. With certain conditions such as separation of the top $K$ eigenvalues from the others, the estimation is consistent. Classical methods include likelihood ratio tests \citep{bartlett1950}, the scree plot \citep{cattell1966scree}, parallel analysis \citep{horn1965rationale}, etc. Here, we introduce a few recent methods: the first one is based on the eigenvalue ratio, the second on eigenvalue differences, and the third on the eigenvalue magnitude.

For simplicity, let us use the sample covariance and arrange its eigenvalues in descending order: $\lambda_1 \ge \lambda_2 \ge \cdots \ge \lambda_{n \wedge p}$, where $n \wedge p = \min\{n,p\}$ (the remaining eigenvalues, if any, are zero). \cite{LamYao12} and \cite{ahn2013eigenvalue} proposed an estimator $\hat{K}_1$ based on ratios of consecutive eigenvalues. For a pre-determined $k_{\max}$, the eigenvalue ratio estimator is
\begin{equation*}
\hat{K}_1 = \underset{i \le k_{\max}}{\argmax} \, \frac{\lambda_i}{\lambda_{i+1}}.
\end{equation*}
Intuitively, when the signal eigenvalues are well separated from the other eigenvalues, the ratio at $k=K$ should be large. Under some conditions, the consistency of this estimator, which does not involve complicated tuning parameters, is established.

In an earlier work,  \cite{onatski2010determining} proposed to use the differences of consecutive eigenvalues. For a given $\delta > 0$ and pre-determined integer $k_{\max}$, define
\begin{equation*}
\hat{K}_2(\delta) = \max \{ i \le k_{\max}: \lambda_i - \lambda_{i+1} \ge \delta \}.
\end{equation*}
%Assume $n/p \rightarrow c >0$, and $k_{\max}$ is a slowly increasing sequence such that $k_{\max} > K$ for large $p$ and $k_{\max} / p \rightarrow 0$.
Using a result on eigenvalue empirical distribution from random matrix theory, \cite{onatski2010determining} proved consistency of $\hat{K}_2(\delta)$ under the Pervasiveness Assumption. The intuition is that, the Pervasiveness Assumption implies that $\lambda_K - \lambda_{K+1}$ tends on $\infty$ in probability as $n \to \infty$; whereas $\lambda_i - \lambda_{i+1} \to 0$ almost surely for $K <i < k_{\max}$ because these $\lambda_i$-s converge to the same limit, which can be determined using random matrix theory. \cite{onatski2010determining} also proposed a data-driven  way to determine $\delta$ from the empirical eigenvalue distribution of the sample covariance matrix.

%One key structural assumption of this method is that the eigengap between $\lambda_K$ and $\lambda_{K+1}$ is sufficiently large. This is possible in high dimensions ($p \to \infty$), and is, in particular, ensured by the Pervasiveness Assumption. %if the eigenvalues of $\bfm B^\top \bfm B$ diverges and the eigenvalues of $\mathrm{Cov}(\bfm u_i)$ remain bounded.
%By Weyl's inequality, the difference $\lambda_K - \lambda_{K+1}$ also diverges, and therefore consistent estimation of $K$ is possible.

%Instead of the eigenvalue differences, eigenvalue ratios of neighboring eigenvalues are used to construct an estimator.
%With a similar structural assumption, consistency result is also established.
A third possibility is to use an information criterion.
%For simplicity, we assume that $\bmu=\mathbf{0}$ in (\ref{fm}).
Define
\begin{equation*}
V(k) = \frac{1}{np}\min_{\hat{\bfm B} \in \R^{p \times k},  \hat{\bfm F} \in \R^{n \times k} } \| \bfm X - \hat\bmu \mathbf{1}_{n}^{\top} - \hat{\bfm B}  \hat{\bfm F}^\top \|_F^2 = p^{-1} \sum_{j > k} \lambda_j,
\end{equation*}
where $\hat{\bmu}$ is the sample mean, and the equivalence (second equality) is well known. For a given $k$,  $V(k)$ is interpreted as the scaled sum of squared residuals, which measures how well $k$ factors fit the data. A very natural estimator $\hat{K}_3$ is to find the best $k \le k_{\max}$ such that the following penalized version of $V(k)$ is minimized \citep{BNg02}:
\begin{equation*}
PC(k) = V(k) + k\,\hat{\sigma}^2 g(n,p), \qquad \text{where} \quad g(n,p):=\frac{n+p}{np} \log \left( \frac{np}{n+p} \right),
\end{equation*}
and $\hat{\sigma}^2$ is any consistent estimate of $(np)^{-1} \sum_{i = 1}^n \sum_{j = 1}^d \E u_{ji}^2$. The upper limit $k_{\max}$ is assumed to be no smaller than $K$, and is typically chosen as $8$ or $15$ in empirical studies in \cite{BNg02}. Consistency results are established under more general choices of $g(n,p)$.

We conclude this section by remarking that in general, it is impossible to consistently estimate $K$ if the smallest nonzero eigenvalue $\bfm B^\top \bfm B$ is much smaller than $\|\bfm \Sigma_u \|_2$, because the `signals' (eigenvalues of $\bfm B^\top \bfm B$) would not be distinguishable from the the noise (eigenvalues of $\bfm U \bfm U^\top$). As mentioned before, consistency of PCA is well studied in the random matrix theory literature. See \cite{dobriban2017factor} for a recent work that justifies parallel analysis using random matrix theory.

\subsection{Robust covariance inputs}

%As we see in the previous subsection, if we obtain an upper bound on the spectral norm $\opnorm{ \hat{\bSigma} - \bSigma}$, we can use the Davis-Kahan theorem to control  the $\ell_2$ distance between the population eigenvectors and their empirical counterparts obtained from $\hat{\bSigma}$.

To extract latent factors and their factor loadings, we need an initial covariance estimator.  Given independent observations $\bx_1,\ldots,\bx_n$ with mean zero, the sample covariance matrix, namely $\hat{\bSigma}_{{\rm sam}} := n^{-1} \sum_{i=1}^n \bx_i \bx_i^\top$, is a natural choice to estimate $\bSigma \in \RR^{p\times p}$. The finite sample bound on $\| \hat{\bSigma}_{{\rm sam}} - \bSigma \|_2 $ has been well studied in the literature \citep{Ver10, Tro12, KLo17}. Before presenting the result from \cite{Ver10}, let us review the definition of sub-Gaussian variables.

A random variable $\xi$ is called sub-Gaussian if $\|\xi\|_{\psi_2} \equiv \sup_{q \ge 1} q^{-1/2} (\E |\xi|^q)^{1/q}$ is finite, in which case this quantity defines a norm $\| \cdot \|_{\psi_2}$ called the sub-Gaussian norm. Sub-Gaussian variables include as special cases Gaussian variables, bounded variables, and other variables with tails similar to or lighter than Gaussian tails. For a random vector $\bxi$, we define $\| \bxi \|_{\psi_2} := \sup_{\| \bv \|_2=1}\| \bxi^\top \bv \|_{\psi_2}$; we call $\bxi$ sub-Gaussian if $\| \bxi \|_{\psi_2}$ is finite.

\begin{thm}\label{thm2.1}
	Let $\bSigma$ be the covariance matrix of $\bx_i$. Assume that $\{\bSigma^{-\frac{1}{2}}\bx_i\}_{i=1}^n$ are i.i.d.\ sub-Gaussian random vectors, and denote $\kappa = \sup_{\| \bv \|_2=1}\| \bx_i^\top \bv \|_{\psi_2}$. Then for any $t \ge 0$, there exist constants $C$ and $c$ only depending on $\kappa$ such that
	\beq \label{eq2.3}
		\PP\Bigl(\opnorm{\widehat\bSigma_{{\rm sam}} - \bSigma} \ge \max(\delta, \delta^2)\opnorm{\bSigma} \Bigr) \le 2 \exp(-ct^2),
	\eeq
	where $\delta = C\sqrt{p / n} + t / \sqrt{n}$.
\end{thm}
\begin{rem}
	The spectral-norm bound above depends on the ambient dimension $p$, which can be large in high-dimensional scenarios. Interested readers can refer to \cite{KLo17} for a refined result that only depends on the intrinsic dimension (or effective rank) of $\bSigma$.
\end{rem}

An important asepect of the above result is the sub-Gaussian concentration in \eqref{eq2.3}, but this depends heavily on the sub-Gaussian or sub-exponential behaviors of observed random vectors.  This condition can not be validated in high dimensions when tens of thousands of variables are collected.  See \cite{FWZhu16}.
When the distribution is  heavy-tailed\footnote{Here, we mean it has second bounded moment when estimating the mean and has bounded fourth moment when estimating the variance.}, one cannot expect sub-Gaussian or sub-exponential behaviors of the sample covariance in the spectral norm \citep{Cat12}. See also \cite{Ver12} and \cite{SriVer13}. Therefore, to perform PCA for heavy-tailed data, the sample covariance is not a good choice to begin with. Alternative robust estimators have been constructed to achieve better finite sample performance.

\cite{Cat12}, \cite{FLW17} and \cite{FWZhu16} approached the problem by first considering estimation of a univariate mean $\mu$ from a sequence of i.i.d random variables $X_1, \cdots, X_n$ with variance $\sigma^2$.  In this case, the sample mean $\bar{X}$ provides an estimator but without exponential concentration.  Indeed, by Markov inequality, we have  $\PP( |\bar{X} - \mu | \geq t \sigma/\sqrt{n}) \leq t^{-2}$, which is tight in general and has a Cauchy tail (in terms of $t$). On the other hand, if we truncate the data $\widetilde X_i= \mathrm{sign}(X_i)\min(|X_i|, \tau)$ with $\tau \asymp \sigma\sqrt{n}$ and compute the mean of the truncated data, then we have \citep{FWZhu16}
$$
	\PP\Bigl( \bigl |\frac{1}{n}\sum\limits_{i=1}^n  \widetilde X_i - \mu \bigr | \ge t \frac{\sigma}{\sqrt{n}} \Bigr) \le  2\exp(- ct^2),
$$
for a universal constant $c>0$.  In other words, the mean of truncated data with only a finite second moment behaves very much the same as the sample mean from the normal data: both estimators have Gaussian tails (in terms of $t$).  This sub-Gaussian concentration is fundamental in high-dimensional statistics as the sample mean is computed tens of thousands or even millions of times.

As an example, estimating the high-dimensional covariance matrix $\bSigma = (\sigma_{ij})$ involves $O(p^2)$ univariate mean estimation, since the covariance can be expressed as an expectation: as $\sigma_{ij} = \E(X_iX_j) - \E(X_i) \E(X_j)$.  Estimating each component by the truncated mean yields a covariance matrix $\tilde{\bSigma}$.
Assuming the fourth moment is bounded (as the covariance itself are second moments), by using the union bound and the above concentration inequality, we can easily obtain
$$
		\PP\Big ( \|\widetilde\bSigma-\bSigma\|_{\max} \ge \sqrt{\frac{a\log p }{c'n} }  \Big )\lesssim p^{2-a}
$$
for any $a>0$ and a constant $c'>0$. In other words, with truncation, when the data have merely bounded fourth moments, we can achieve the same estimation rate as the sample covariance matrix under the Gaussian data.

\cite{FWZhu16} and \cite{Min16} independently proposed shrinkage variants of the sample covariance with sub-Gaussian behavior under the spectral norm, as long as the fourth moments of $\bX$ are finite. For any $\tau\in \RR^+$,   \cite{FWZhu16} proposed the following shrinkage sample covariance matrix %$\widehat\bSigma_s(\tau)$ to estimate $\bSigma$.
	\beq
		\label{eq:4.1}
		\widehat\bSigma_s(\tau)=\frac{1}{n}\sum\limits_{i=1}^n \widetilde\bx_i\widetilde\bx_i^\top, \qquad \widetilde\bx_i:=(\|\bx_i\|_4 \wedge \tau)\bx_i/\|\bx_i\|_4,
	\eeq
to estimate $\bSigma$,
where $\|\cdot\|_4$ is the $\ell_4$-norm.
	The following theorem establishes the statistical error rate of $\widetilde \bSigma_s(\tau)$ in terms of the spectral norm.
	\begin{thm}
		\label{thm:2.2}
		Suppose $\E(\bv^\top\bx_i)^4\le R$ for any unit vector $\bv\in \cS^{p-1}$. Then it holds that for any $\delta>0$,
		\begin{equation}
			\PP\Big(\opnorm{\widehat\bSigma_s(\tau) - \bSigma} \ge \sqrt{\frac{\delta Rp \log p}{n}} \Big)\le p^{1-C\delta},
		\end{equation}
		where $\tau\asymp \bigl(nR/(\delta\log p) \bigr)^{1/4}$ and $C$ is a universal constant.
	\end{thm}
Applying PCA to the robust covariance estimators as described above leads to more reliable estimation of principal eigenspaces in the presence of heavy-tailed data.

In Theorem~\ref{thm:2.2}, we assume that the mean of $\bx_i$ is zero.  When this does not hold, a natural estimator of $\bSigma = \frac{1}{2}
\E (\bx_1 - \bx_2)(\bx_1 - \bx_2)^\top$ is to use the shrunk $U$-statistic
\citep{FKS17}:
\begin{eqnarray*}
\hat{\bSigma}_{U}(\tau)  &=& \frac{1}{2{n \choose 2}}
\sum_{j \neq k}  \frac{\psi_\tau(\ltwonorm{\bx_j - \bx_k} ^ 2)}{\ltwonorm{\bx_j - \bx_k} ^ 2} (\bx_j - \bx_k)(\bx_j - \bx_k)^\top \\
& = & \frac{1}{2{n \choose 2}} \sum_{j \neq k}   \min\bigl( 1, \tau / \| \bx_j - \bx_k \|_2^2 \bigr) (\bx_j - \bx_k)(\bx_j - \bx_k)^{\top},
\end{eqnarray*}
where $\psi_\tau(u) = (|u| \wedge \tau) {\rm sign}(u)$.  When $\tau = \infty$, it reduces to the usual $U$-statistics. It possesses a similar concentration property to that in Theorem \ref{thm:2.2} with a proper choice of $\tau$.

\subsection{Perturbation bounds}

In this section, we introduce several perturbation results on eigenspaces, which serve as fundamental technical tools in factor models and related learning problems. For example, in relating the factor loading matrix $\bB$ to the principal components of covariance matrix $\bSigma$ in \eqref{cov.fm}, one can regard $\bSigma$ as a perturbation of $\bB \bB^\top$ by an amount of $\bSigma_u$ and take $\bA = \bB \bB^\top$ and $\tilde \bA = \bSigma$ in Theorem~\ref{thm:dk} below.  Similarly, we can also regard a covariance matrix estimator $\hat{\bSigma}$ as a perturbation of $\bSigma$ by an amount of
$\hat{\bSigma} - \bSigma$.

We will begin with a review of the Davis-Kahan theorem \citep{DKa70}, which is usually useful for deriving $\ell_2$-type bounds (which includes spectral norm bounds) for symmetric matrices. Then, based on this classical result, we introduce entry-wise ($\ell_\infty$) bounds, which typically give refined results under structural assumptions.
We also derive bounds for rectangular matrices that are similar to Wedin's theorem \citep{Wed72}. Several recent works on this topic can be found in \cite{YWS14,FWZ16,KolXia16,AFWZ17,Zho17,CTP17,EldBelWan17}.

First, for any two subspaces $\mathcal{S}$ and $\tilde{\mathcal{S}}$ of the same dimension $K$ in $\R^p$, we choose any $\bV, \tilde \bV \in \RR^{p \times K}$ with orthonormal columns that span $\mathcal{S}$ and $\tilde{\mathcal{S}}$, respectively. We can measure the closeness between two subspaces though the difference between their projectors:
\begin{equation*}
d_2(\mathcal{S}, \tilde{\mathcal{S}}) = \| \tilde{\bfm V} \tilde{\bfm V}^\top - \bfm V \bfm V^\top\|_2 \quad \text{or} \quad d_F(\mathcal{S}, \tilde{\mathcal{S}}) = \| \tilde{\bfm V} \tilde{\bfm V}^\top - \bfm V \bfm V^\top\|_F.
\end{equation*}
The above definitions are both proper metrics (or distances) for subspaces $\bS$ and $\widetilde \bS$ and do not depend on the specific choice of $\bV$ and $\tilde \bV$, since $\tilde{\bfm V} \tilde{\bfm V}^\top$ and  $\bfm V \bfm V^\top$ are projection operators.  Importantly, these two metrics are connected to the well-studied notion of \textit{canonical angles} (or principal angles). Formally, let the singular values of $\tilde{\bfm V}^\top \bfm V$ be $\{\sigma_k\}_{k = 1}^K$, and define the canonical angles $\theta_k = \cos ^{-1}\sigma_k$ for $k=1,\ldots,K$. It is often useful to denote the sine of the canonical (principal) angles by $\sin \bTheta(\widehat\bV, \bV) := \diag(\sin \theta_1,\ldots,\sin \theta_K) \in \R^{K \times K}$, which can be interpreted as a generalization of sine of angles between two vectors. The following identities are well known \citep{SSu90}.
\begin{equation*}
\| \sin \bTheta(\tilde{\bfm V}, \bfm V) \|_2 = d_2(\mathcal{S}, \tilde{\mathcal{S}}), \qquad \sqrt{2}\| \sin \bTheta(\tilde{\bfm V}, \bfm V) \|_F = d_F(\mathcal{S}, \tilde{\mathcal{S}}).
\end{equation*}

%the sine of canonical (principal) angles, denoted by $\sin \bTheta(\tilde{\bfm V}, \bfm V)$, which is a generalization of sine of angles between two vectors. For the Frobenius norm and the spectral norm, the following identities are well known \citep{SSu90}.
%\begin{equation*}
%\| \sin \bTheta(\tilde{\bfm V}, \bfm V) \|_2 = \| \tilde{\bfm V} \tilde{\bfm V}^\top - \bfm V \bfm V^\top\|_2, \qquad \sqrt{2}\| \sin \bTheta(\tilde{\bfm V}, \bfm V) \|_F = \| \tilde{\bfm V} \tilde{\bfm V}^\top - \bfm V \bfm V^\top\|_F.
%\end{equation*}
%The definition of canonical angles, as well as the above identity, does not depend on the specific choice of $\bV$ and $\tilde \bV$. The right-hand side can be interpreted as the difference between projectors onto $\tilde{\mathcal{S}}$ and $\mathcal{S}$, as $\tilde{\bfm V} \tilde{\bfm V}^\top$ and  $\bfm V \bfm V^\top$ are projection operators.

In some cases, it is convenient to fix a specific choice of $\tilde{\bfm V}$ and $\bfm V$. It is known that for both Frobenius norm and spectral norm,
\begin{equation*}
\| \sin \bTheta(\tilde{\bfm V}, \bfm V) \| \le \min_{\bfm R \in \mathcal{O}(K)} \| \tilde{\bfm V} \bfm R - \bfm V \| \le \sqrt{2}\, \| \sin \bTheta(\tilde{\bfm V}, \bfm V) \|,
\end{equation*}
where $\mathcal{O}(K)$ is the space of orthogonal matrices of size $K \times K$. The minimizer (best rotation of basis) can be given by the singular value decomposition (SVD) of $\tilde{\bfm V}^\top \bfm V$. For details, see \cite{CTP17} for example.

Now, we present the Davis-Kahan $\sin\theta$ theorem \citep{DKa70}.

\begin{thm}\label{thm:dk}
Suppose $\bfm A ,\widetilde{\bfm A}\in \R^{n \times n}$ are symmetric, and that $ \bfm V,\tilde{\bfm V}\in \R^{n \times K}$ have orthonormal column vectors which are eigenvectors of $\bfm A$ and $\tilde{\bfm A}$ respectively. Let $\mathcal{L}(\bfm V)$ be the set of eigenvalues corresponding to the eigenvectors given in $\bfm V$, and let $\mathcal{L}(\bfm V^{\bot})$ (respectively $\mathcal{L}(\tilde{\bfm V}^{\bot})$) be the set of eigenvalues corresponding to the eigenvectors not given in $\bfm V$ (respectively $\tilde{\bfm V}$). If there exists an interval $[\alpha, \beta]$ and $\delta>0$ such that $\mathcal{L}(\bfm V) \subset [\alpha, \beta]$ and $\mathcal{L}(\tilde{\bfm V}^{\bot}) \subset  (-\infty, \alpha-\delta] \cup [\beta+\delta,+\infty)$, then for any orthogonal-invariant norm\footnote{A norm $\| \cdot \|$ is orthogonal-invariant if $\| \bfm U^\top \bfm B \bfm V \| = \| \bfm B \|$ for any matrix $B$ and any orthogonal matrices $\bfm U$ and $\bfm V$.}
\begin{equation*}
\| \sin \bTheta(\tilde{\bfm V}, \bfm V) \| \le \delta^{-1}\,\| \bfm (\widetilde{\bfm A} - \bfm A) \bfm V \|.
\end{equation*}
\end{thm}

This theorem can be generalized to singular vector perturbation for rectangular matrices; see \cite{Wed72}. A slightly unpleasant feature of this theorem is that $\delta$ depends on the eigenvalues of both $\bfm A$ and $\tilde{\bfm A}$. However, with the help of Weyl's inequality, we can immediately obtain a corollary that does not involve the eigenvalues of $\tilde \bA$. Let $\lambda_j(\cdot)$ denote the $j$th largest eigenvalue of a real symmetric matrix. Recall that Weyl's inequality bounds the differences between the eigenvalues of $\bfm A$ and $\tilde{\bfm A}$:
\begin{equation}\label{eq.weyl}
\max_{1\le j \le n} \left| \lambda_j(\tilde{\bfm A}) - \lambda_j(\bfm A) \right| \le \| \tilde{\bfm A} - \bfm A \|_2.
\end{equation}
This inequality suggests that, if the eigenvalues in $\mathcal{L}(\tilde{\bfm V}^{\bot})$ have the same ranks (in descending order) as those in $\mathcal{L}(\bfm V^{\bot})$, then $\mathcal{L}(\tilde{\bfm V}^{\bot})$ and $\mathcal{L}(\bfm V^{\bot})$ are similar. Below we state our corollary, whose proof is in the appendix.

\begin{cor}\label{cor:dk}
Assume the setup of the above theorem, and suppose the eigenvalues in $\mathcal{L}(\tilde{\bfm V})$ have the same ranks as those in $\mathcal{L}(\bfm V)$. If $\mathcal{L}(\bfm V) \subset [\alpha, \beta]$ and $\mathcal{L}(\bfm V^{\bot}) \subset  (-\infty, \alpha-\delta_0] \cup [\beta+\delta_0,+\infty)$ for some $\delta_0>0$, then
\begin{equation*}
\| \sin \bTheta(\tilde{\bfm V}, \bfm V) \|_2 \le 2 \delta_0^{-1} \| \widetilde{\bfm A} - \bfm A \|_2.
\end{equation*}
\end{cor}
We can then use $\| \sin \bTheta(\tilde{\bfm V}, \bfm V) \|_F \le \sqrt{K}\,\| \sin \bTheta(\tilde{\bfm V}, \bfm V) \|_2$ to obtain a bound under the Frobenius norm.
%Note that in the statement of the theorem, it is equivalent to suppose that the eigenvalues in $\mathcal{L}(\tilde{\bfm V}^{\bot})$ have the same ranks as those in $\mathcal{L}(\bfm V^{\bot})$. We also remark that
In the special case where $\mathcal{L}(\bV) = \{ \lambda \}$ and $\bV = \bv$, $\tilde \bV = \tilde \bv$ reduce to vectors, we can choose $\alpha = \beta = \lambda$, and the above corollary translates into
\begin{equation}
\min_{s \in \{\pm 1\}} \| \widehat{\bv} - s \bv \|_2 \le \sqrt{2}\, \sin\theta(\widehat{\bv} , \bv )  \le 2\sqrt{2}\, \delta_0^{-1} \| \widetilde{\bfm A} - \bfm A \|_2.
\end{equation}

\begin{figure}[h]
\centering
\includegraphics[scale=0.45]{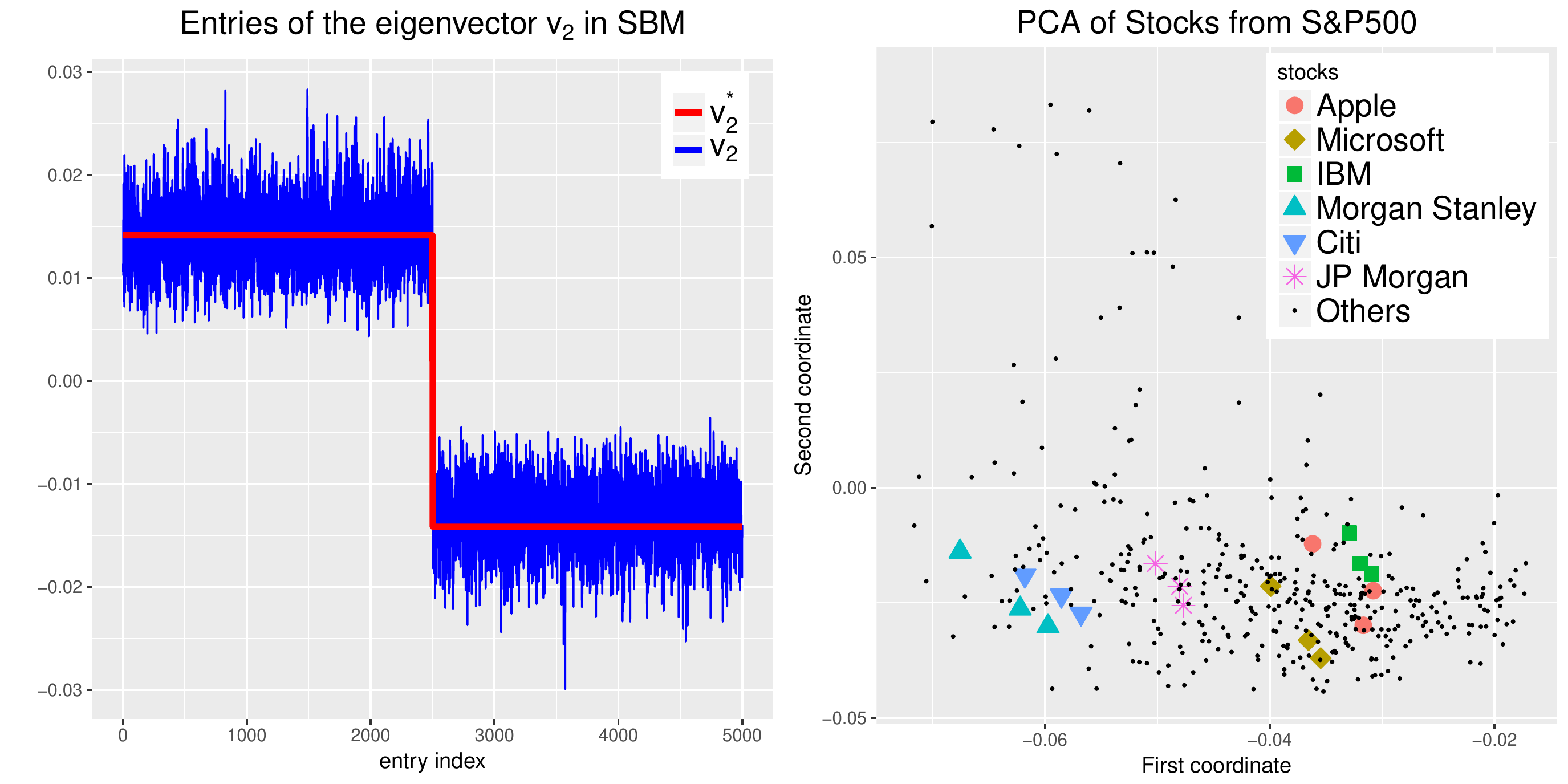}
\caption{The \textbf{left plot} shows the entries (coordinates) of the second eigenvectors $\bfm v_2$ computed from the adjacency matrix from the SBM with two equal-sized blocks ($n=5000, K=2$). The plot also shows the expectation counterpart $\bfm v_2^*$, whose entries have the same magnitude $O(1/\sqrt{n})$. The deviation of $\bfm v_2$ from $\bfm v_2^*$ is quite uniform, which is a phenomenon not captured by the Davis-Kahan's theorem. The \textbf{right plot} shows the coordinates of two leading eigenvectors of the sample covariance matrix calculated from 2012--2017 daily return data of 484 stocks (tiny black dots). We also highlight six stocks during three time windows (2012--2015, 2013--2016, 2014--2017) with big markers, so that the fluctuation/perturbation is shown.  The magnitude of these coordinates is typically small, and the fluctuation is also small.}\label{fig:SBMeig}
\end{figure}

We can now see that the factor model and PCA are approximately the same with sufficiently large eigen-gap. Indeed, under Identifiability Assumption~\ref{ass:1}, we have $\bfm \Sigma = \bfm B \bfm B^\top + \bfm \Sigma_u$. Applying Weyl's inequality and Corollary~\ref{cor:dk} to $\bfm B \bfm B^\top$ (as $\bfm A$) and $\bfm \Sigma$ (as $\tilde{\bfm A}$), we can easily control the eigenvalue/eigenvector differences by $\| \bfm \Sigma_u \|_2$ and the eigengap, which is comparably small under Pervasiveness Assumption~\ref{ass:2}. This difference can be interpreted as the bias incurred by PCA on approximating factor models.
%Once we obtain a good control of $\| \hat{\bfm \Sigma} - \bfm \Sigma \|_2$, we can set $\bfm A = \bfm B \bfm B^\top$ and $\tilde{\bfm A} = \hat{\bfm \Sigma}$. By Wey

%Usually, in factor models or PCA applications, we let $\bfm A$ be the population covariance matrix (typically unknown), and $\widetilde{\bfm A}$ be our covariance estimate (e.g.\ the sample covariance matrix). Then,
Furthermore, given any covariance estimator $\hat{\bfm \Sigma}$, we can similarly apply the above results by setting $\bfm A = \bfm \Sigma$ and $\tilde{\bfm A} = \hat{\bfm \Sigma}$ to bound the difference between the estimated eigenvalues/eigenvectors and the population counterparts. Note that the above corollary gives us an upper bound on the subspace estimation error in terms of the ratio $ \| \hat{\bfm \Sigma} - \bfm \Sigma \|_2/ \delta_0$.
%, whose inverse can be interpreted as the signal-to-noise ratio.

Next, we consider entry-wise bounds on the eigenvectors. For simplicity, here we only consider eigenvectors corresponding to unique eigenvalues rather than the general eigenspace. Often, we want to have a bound on each entry of the eigenvector difference $\widetilde\bv - \bv $,  instead of an $\ell_2$ norm bound, which is an average-type result. In many cases, none of these entries has dominant perturbation, but the Davis-Kahan's theorem falls short of providing a reasonable bound (the na\" ive bound $\| \cdot \|_{\infty} \le \| \cdot \|_2$ gives a suboptimal result).

Some recent papers \citep{AFWZ17} have addressed this problem, and in particular, entry-wise bounds of the following form are established.
\begin{equation*}
| [ \widetilde{\bfm v} - \bfm v]_m | \lesssim \mu\, \frac{\| \widetilde{\bfm A} - \bfm A \|_2}{\delta_0}  + \text{small term}, \qquad \forall\ m \in [n],
\end{equation*}
where $\mu \in [0,1]$ is related to the structure of the statistical problem and typically can be as small as $O(1/\sqrt{n})$, which is very desirable in high-dimensional setting. The small term is often related to independence pattern of the data, which is typically small under mild independence conditions.

We illustrate this idea in Figure~\ref{fig:SBMeig} through a simulated data example (left) and a real data example (right), both of which have factor-type structure. For the left plot, we generated a network data according to the stochastic block model with $K=2$ blocks (communities), each having nodes $n/2 = 2500$: the adjacency matrix that represents the links between nodes is a symmetric matrix, with upper triangular elements generated independently from Bernoulli trials (diagonal elements are taken as 0), with the edge probability $5\log n/n$ for two nodes within blocks and $\log n/(4n)$ otherwise. Our task is to classify (cluster) these two communities based on the adjacency matrix. We used the second eigenvector $\bv_2 \in \R^{5000}$ (that is, corresponding to the second largest eigenvalue) of the adjacency matrix as a classifier. The left panel of Figure~\ref{fig:SBMeig} represents the values of the $5000$ coordinates (or entries) $[\bv_2]_i$ in the y-axis against the indices $i=1,\ldots, 5000$ in the x-axis. For comparison, the second eigenvector $\bv_2^* \in \R^{5000}$ of the expectation of the adjacency matrix---which is of interest but unknown---have entries taking values only in $\{\pm 1/\sqrt{5000}\}$, depending on the unknown nature of which block a vertex belongs to (this statement is not hard to verify). We used the horizontal line to represent these ideal values:  they indicate exactly the membership of each vertex. Clearly, the magnitude of entry-wise perturbation is $O(1/\sqrt{n})$.  Therefore, we can use $\mathrm{sign}(\bv_2)/\sqrt{5000}$ as an estimate of $\bv^*$ and classify all nodes with the same sign as the same community.  See Section~\ref{sec:comdet} for more details.

%The curves represent the values of the $5000$ coordinates (entries), and the horizontal line represent the ideal but unknown values. Clearly, the magnitude of entry-wise perturbation is $O(1/\sqrt{n})$. See Section~\ref{sec:comdet} for more details.

For the right plot, we used daily return data of stocks that are constituents of S\&P 500 index from  2012.1.1--2017.12.31. We considered stocks with exactly $n=1509$ records and excluded stocks with incomplete/missing values, which resulted in $p=484$ stocks. Then, we calculated the sample covariance matrix $\hat{\bfm \Sigma}_{\sam} \in \R^{p \times p}$ using the data in the entire period, and computed two leading eigenvectors (note that they span the column space of $\bB$) and plotted the coordinates (entries) using small dots. Stocks with an coordinate smaller than $5\%$ quantile or larger than $95\%$ quantile are potentially outlying values and are not shown in the plot. In addition, we also highlighted the fluctuation of six stocks during three time windows: 2012.1--2015.12, 2013.1--2016.12 and 2014.1--2017.12, with different big markers. That is, for each of the three time windows, we re-computed the covariance matrices and the two leading eigenvectors, and then highlighted coordinates that correspond to the six major stocks. Clearly, the magnitude for these stocks is small, which is roughly $O(1/\sqrt{p})$, and the fluctuation of coordinates is also very small. Both plots suggest an interesting phenomenon of eigenvectors in high dimensions: \textit{entry-wise behavior of eigenvectors can be benign under factor model structure}.

To state our results rigorously, let us suppose that $\bfm A, \widetilde{\bfm A},  \bfm W \in \mathbb{R}^{n \times n}$ are symmetric matrices, with $\widetilde{\bfm A} = \bfm A + \bfm W$ and $\rank(\bfm A) = K < n$. Let the eigen-decomposition of $\bfm A$ and $\widetilde{\bfm A}$ be
\begin{equation}\label{eq:decomp}
\bfm A = \sum_{k=1}^K \lambda_k \bfm v_k \bfm v_k^\top, \quad \text{and} \quad \widetilde{\bfm A} = \sum_{k=1}^K \widetilde \lambda_k \widetilde{\bfm v}_k \widetilde{\bfm v}_k^\top + \sum_{k=K+1}^n \widetilde \lambda_k \widetilde{\bfm v}_k \widetilde{\bfm v}_k^\top.
\end{equation}
Here the eigenvalues $\{ \lambda_k \}_{k=1}^K$ and $\{\widetilde{\lambda}_k\}_{k=1}^K$ are the $K$ largest ones of $\bA$ and $\widetilde{\bA}$, respectively, in terms of absolute values. Both sequences are sorted in descending order. $\{\widetilde{\lambda}_k\}_{k=K+1}^n$ are eigenvalues of $\widetilde{\bA}$ whose absolute values are smaller than $|\widetilde{\lambda}_K|$. The eigenvectors $\{ \bfm v_k \}_{k=1}^K$ and $\{ \widetilde{\bfm v}_k \}_{k=1}^n$ are normalized to have unit norms.

Here $\{\lambda_k\}_{k=1}^K$ are allowed to take negative values. Thanks to Weyl's inequality, $\{ \widetilde{\lambda}_k \}_{k=1}^K$ and $\{ \widetilde{\lambda}_k \}_{k=K+1}^n$ are well-separated when the size of perturbation $\bfm W$ is not too large. In addition, we have the freedom to choose signs for eigenvectors, since they are not uniquely defined. Later, we will use `up to sign' to signify that our statement is true for at least one choice of sign. With the conventions $\lambda_0 = +\infty$ and $ \lambda_{K+1} = - \infty$, we define the eigen-gap as
\begin{equation}\label{def:delta}
\delta_k = \min \{ \lambda_{k-1} - \lambda_k, \lambda_{k} - \lambda_{k+1}, |\lambda_k|\},  \quad \forall\, k \in [K],
\end{equation}
which is the smallest distance between $\lambda_k$ and other eigenvalues (including $0$). This definition coincides with the (usual) eigen-gap in Corollary~\ref{cor:dk} in the special case $\mathcal{L}(\bv_k) = \{ \lambda_k \}$ where we are interested in a single eigenvalue and its associated eigenvector.

We now present an entry-wise perturbation result.  Let us first look at only one eigenvector.  In this case, when $\| \tilde \bA - \bA\|$ is small, heuristically,
$$
\tilde{\bv}_k  =  \frac{\tilde \bA  \tilde\bv_k}{\tilde \lambda_k} \approx \frac{\tilde \bA  \bv_k }{\lambda_k} = \bv_k + \frac{(\tilde \bA - \bA) \bv_k}{\lambda_k}
$$
holds uniformly for each entry.
When $\bA = \E \tilde \bA$, that is, $\tilde \bA$ is unbiased, this gives the first-order approximation (rather than bounds on the difference $\tilde \bv_k - \bv_k$) of the random vector $\tilde \bv_k$. \cite{AFWZ17} proves rigorously this result and generalizes to eigenspaces. The key technique for the proof is similar to Theorem~\ref{thm:pert-sym} below, which simplifies the one in \cite{AFWZ17}  in various ways but holds under more general conditions. It is stated in a deterministic way, and can be powerful if there is certain structural independence in the perturbation matrix $\bfm W$. A self-contained proof can be found in the appendix.

For each $m \in [n]$, let $\bfm W^{(m)} \in \R^{n \times n}$ be a modification of $\bW$ with the $m$th row and $m$th column zeroed out, i.e.,
\begin{equation*}
W^{(m)}_{ij} = W_{ij} \mathbbm{1}_{\{i\neq m\}} \mathbbm{1}_{\{j\neq m\}}, \quad \forall \, i, j \in [n].
\end{equation*}
We also define $\widetilde{\bfm A}^{(m)} = \bfm A + \bfm W^{(m)}$, and denote its eigenvalues and eigenvectors by $\{ \widetilde \lambda^{(m)}_k \}_{k=1}^n$ and $\{ \widetilde{\bfm v}_k^{(m)} \}_{k=1}^n$, respectively.
%If the eigen-gap $\delta_k$ is large enough compared with $\|\bfm W^{(m)}\|_2$, $\{ \widetilde \lambda^{(m)}_k \}_{k=1}^K$ and $\{ \widetilde \lambda^{(m)}_k \}_{k=K+1}^n$ are well-separated.
This construction is related to the leave-one-out technique in probability and statistics. For recent papers using this technique, see \cite{bean2013optimal, ZhoBou18, AFWZ17} for example.

\begin{thm}\label{thm:pert-sym}
Fix any $\ell \in [K]$. Suppose that $|\lambda_\ell| \asymp \max_{k \in [K]} | \lambda_k|$, and that the eigen-gap $\delta_\ell$ as defined in \eqref{def:delta} satisfies $\delta_\ell \ge 5 \| \bfm W \|_2$. Then, up to sign,
\begin{equation}\label{ineq:pert-sym}
\left| [ \widetilde{\bfm v}_\ell - \bfm v_\ell ]_m \right| \lesssim \frac{ \| \bfm W \|_2}{\delta_\ell} \left( \sum_{k=1}^K [\bfm v_k]_m^2 \right)^{1/2} + \frac{| \langle \bfm w_m,  \widetilde{\bfm v}_\ell^{(m)} \rangle|}{\delta_\ell}, \quad \forall\, m \in [n],
\end{equation}
where $\bfm w_m$ is the $m$th column of $\bfm W$.
\end{thm}

To understand this theorem, let us compare it with the standard $\ell_2$ bound (Theorem~\ref{thm:dk})
, which implies $\| \widetilde{\bfm v}_\ell - \bfm v_\ell  \|_2 \lesssim \| \bfm W \|_2 / \delta_\ell$. The first term of the upper bound in \eqref{ineq:pert-sym} says the perturbation on the $m$th entry can be much smaller, because the factor $(\sum_{k=1}^K [\bfm v_k]_m^2)^{1/2}$, always bounded by $1$, can be usually much smaller. For example, if $\bfm v_k$'s are uniformly distributed on the unit sphere, then this factor is typically of order $O(\sqrt{K\log n/n})$. This factor is related to the notion of incoherence in \cite{CRe09, Can11}, etc.

The second term of the upper bound in \eqref{ineq:pert-sym} is typically much smaller than $\| \bW \|_2 / \delta_{\ell}$, especially under certain independence assumption. For example, if $\bfm w_m$ is independent of other entries, then, by construction, $\widetilde{\bfm v}_\ell^{(m)}$ and $\bfm w_m$ are independent. If, moreover, entries of $\bw_m$ are i.i.d.\ standard Gaussian, $| \langle \bfm w_m,  \widetilde{\bfm v}_\ell^{(m)} \rangle|$ is of order $O_{\mathbb{P}}(1)$, whereas $\| \bfm W \|_2$ typically scales with $\sqrt{n}$. This gives a bound for the $m$th entry, and can be extended to an $\ell_\infty$ bound if we are willing to make independence assumption for all $m \in [n]$ (which is typical for random graphs for example).

We remark that this result can be generalized to perturbation bounds for eigenspaces \citep{AFWZ17}, and the conditions on eigenvalues can be relaxed using certain random matrix assumptions \citep{KolXia16, o2017random, Zho17}.

Now, we extend this perturbation result to singular vectors of rectangular matrices. Suppose $\bfm L, \widetilde{ \bfm L},  \bfm E \in \R^{n \times p}$ satisfy $\widetilde{ \bfm L} = \bfm L + \bfm E$ and $\rank(\bfm L) = K < \min\{n,p\}$. Let the SVD of $\bfm L$ and $\widetilde{ \bfm L}$ be\footnote{Here, we prefer using $\bu_k$ to refer to the singular vectors (not to be confused with the noise term in factor models). The same applies to Section~\ref{sec:4}.}
\begin{equation*}
\bfm L = \sum_{k=1}^K \sigma_k \bfm u_k \bfm v_k^\top \quad \text{and} \quad \widetilde{ \bfm L} = \sum_{k=1}^K \widetilde \sigma_k \widetilde{\bfm u}_k \widetilde{\bfm v}_k^\top + \sum_{k=K+1}^{\min\{n,p\}} \widetilde \sigma_k \widetilde{\bfm u}_k \widetilde{\bfm v}_k^\top,
\end{equation*}
where $\sigma_k$ and $\widetilde \sigma_{k}$ are respectively non-increasing in $k$, and $\bfm u_k$ and $\bfm v_k$ are all normalized to have unit $\ell_2$ norm. As before, let $\{ \tilde \sigma_k\}_{k=1}^K$ have $K$ largest absolute values. Similar to \eqref{def:delta}, we adopt the conventions $\sigma_0= +\infty$, $\sigma_{K+1}= 0$ and define the eigen-gap as
\begin{equation}\label{def:gamma}
\gamma_k = \min \{ \sigma_{k-1} - \sigma_k, \sigma_{k} - \sigma_{k+1}\},  \quad \forall\, k \in [K].
\end{equation}

For $j\in [p]$ and $i\in[n]$, we define unit vectors $\{ \tilde{\bfm u}_k^{(j)} \}_{k=1}^{\min\{n,p\}} \subseteq \R^n$ and $\{ \tilde{\bfm v}_k^{(i)} \}_{k=1}^{\min\{n,p\}} \subseteq \R^p$ by replacing certain row or column of $\bfm E$ with zeros. To be specific, in our expression $\widetilde{\bfm L} = \bfm L + \bfm E$, if we replace the $i$th row of $\bfm E$ by zeros, then the normalized right singular vectors of the resulting perturbed matrix are denoted by $\{ \tilde{\bfm v}_k^{(i)} \}_{k=1}^{\min\{n,p\}}$; and if we replace the $j$th column of $\bfm E$ by zeros, then the normalized left singular vectors of the resulting perturbed matrix are denoted by $\{ \tilde{\bfm u}_k^{(j)} \}_{k=1}^{\min\{n,p\}}$.

\begin{cor}\label{cor:pert-assym}
Fix any $\ell \in [K]$. Suppose that $ \sigma_\ell \asymp \max_{k \in [K]}  \sigma_k$, and that $\gamma_\ell \ge 5 \| \bfm E \|_2$. Then, up to sign,
\begin{align*}
\left| [ \widetilde{\bfm u}_\ell - \bfm u_\ell ]_i \right| &\lesssim \frac{ \| \bfm E \|_2}{\gamma_\ell} \left( \sum_{k=1}^K [\bfm u_k]_i^2 \right)^{1/2} + \frac{| \langle (\bfm e_i^\row)^\top,  \widetilde{\bfm v}_\ell^{(i)} \rangle|}{\gamma_\ell}, \quad \forall \, i \in [n], \quad \text{and}\\
\left| [ \widetilde{\bfm v}_\ell - \bfm v_\ell ]_j \right| &\lesssim \frac{ \| \bfm E \|_2}{\gamma_\ell} \left( \sum_{k=1}^K [\bfm v_k]_j^2 \right)^{1/2} + \frac{| \langle \bfm e_j^\coln,  \widetilde{\bfm u}_\ell^{(j)} \rangle|}{\gamma_\ell}, \quad \forall \, j\in [p], \quad
\end{align*}
where $\bfm e_i^{\row} \in \R^{p}$ is the $i$th row vector of $\bfm E$, and $\bfm e_j^{\coln} \in \R^{n}$ is the $j$th column vector of $\bfm E$.
\end{cor}

If we view $\widetilde{\bfm L}$ as the data matrix (or observation) $\bfm X$, then, the low rank matrix $\bfm L$ can be interpreted as $\bfm B \bfm F^\top$. The above result provides a tool of studying estimation errors of the singular subspace of this low rank matrix. Note that $ \widetilde{\bfm v}_\ell^{(i)}$ can be interpreted as the result of removing the idiosyncratic error of the $i$th observation, and $ \widetilde{\bfm u}_\ell^{(j)}$ as the result of removing the $j$th covariate of the idiosyncratic error.

To better understand this result, let us consider a very simple case: $K=1$ and each row of $\bfm E$ is i.i.d.\ $\cN(\mathbf{0},\bfm I_p)$. We are interested in bounding the singular vector difference between the rank-$1$ matrix $\bfm L = \sigma_1 \bfm u \bfm v^\top$ and its noisy observation $\tilde{\bfm L} = \bfm L + \bfm E$. This is a spiked matrix model with a single spike. By independence between $\bfm e_i^\row$ and $\widetilde{\bfm v}_\ell^{(i)}$ as well as elementary properties of Gaussian variables, Corollary~\ref{cor:pert-assym} implies that with probability $1-o(1)$, up to sign,
\begin{equation}\label{ineq-eigen-inf}
\| \tilde{\bfm u}_1 - \bfm u_1 \|_{\infty} \le \frac{\| \bfm E \|_2}{\sigma_1} \| \bfm u_1 \|_{\infty} + \frac{O(\sqrt{\log n})}{\sigma_1}.
\end{equation}

Random matrix theory gives $\| \bfm E \|_2 \asymp \sqrt{n} + \sqrt{p}$ with high probability. Our $\ell_2$ perturbation inequality (Corollary \ref{cor:dk}) implies that $\| \tilde{\bfm u}_1 - \bfm u_1 \|_{2} \leq \| \bfm E \|_2 / \sigma_1$. This upper bound is much larger than the two terms in (\ref{ineq-eigen-inf}), as $\| \bfm u_1 \|_{\infty}$ is typically much smaller than $1$ in high dimensions. Thus, (\ref{ineq-eigen-inf}) gives a better entry-wise control over the $\ell_2$ counterpart.

Beyond this simple case, there are many desirable features of Corollary~\ref{cor:pert-assym}. First of all, we allow $K$ to be moderately large, in which case, as mentioned before, the factor $( \sum_{k=1}^K [\bfm u_k]_i^2 )^{1/2}$ is related to the incoherence structure in the matrix completion and robust PCA literature. Secondly, the result holds deterministically, so random matrices are also applicable.  Finally, the result holds for each $i \in [n]$ and $j \in [p]$, and thus it is useful even if the entries of $\bE$ are not independent, e.g. when a subset of covariates are dependent.

To sum up, our results Theorem~\ref{thm:pert-sym} and Corollary~\ref{cor:pert-assym} provide flexible tools of studying entry-wise perturbation of eigenvectors and singular vectors. It is also easy to adapt to other problems since their proofs are not complicated (see the appendix).

\section{Applications to High-dimensional Statistics}
\label{sec:3}

\subsection{Covariance estimation}\label{sec-cov-est}

Estimation of high-dimensional covariance matrices has wide applications in modern data analysis.
When the dimensionality $p$ exceeds the sample size $n$, the sample covariance matrix becomes singular. Structural assumptions are necessary in order to obtain a consistent estimator in this challenging scenario. One typical assumption in the literature is that the population covariance matrix is sparse, with a large fraction of entries being (close to) zero, see \cite{BLe08} and \cite{CLi11}. In this setting, most variables are nearly uncorrelated. In financial and genetic data, however, the presence of common factors leads to strong dependencies among variables \citep{FFL08}. The approximate factor model (\ref{fm}) better characterizes this structure and helps construct valid estimates. Under this model, the covariance matrix $\bSigma$ has decomposition (\ref{cov.fm}),
where $\bSigma_u=\cov(\bu_i) = ( \sigma_{u,jk} )_{1\leq j,k\leq p}$ is assumed to be sparse \citep{FLM13}. Intuitively, we may assume that $\bSigma_u$ only has a small number of nonzero entries. Formally, we require the sparsity parameter $$
m_0 :=\max_{j\in[p]}\sum_{k=1}^{p} \mathbbm{1}\left\{ \sigma_{u,jk} \neq 0 \right\}$$ to be small. This definition can be generalized to a weaker sense of sparsity, which is characterized by $m_q = \max_{j\in[p]} \sum_{k=1}^{p} | \sigma_{u,jk} |^q$, where $q\in (0,1)$ is a parameter. Note that small $m_q$ forces $\bSigma_u$ to have few large entries. However, for simplicity, we choose not to use this more general definition when presenting theoretical results below.

The approximate factor model has the following two important special cases, under which the parameter estimation has been well studied.
\begin{itemize}
	\item The sparse covariance model is (\ref{cov.fm}) without factor structure, i.e. $\bSigma = \bSigma_u$; typically, entry-wise thresholding is employed for estimation.
	\item The strict factor model corresponds to (\ref{cov.fm}) with $\bSigma_u$ being diagonal; usually, PCA-based methods are used.
\end{itemize}
The approximate factor model is a combination of the above two models, as it comprises both a low-rank component and a sparse component. A natural idea is to fuse methodologies for the two models into one, by estimating the two components using their corresponding methods. This motivated our high-level idea for estimation under the approximate factor model: (1) estimating the low-rank component (factors and loadings) using regression (when factors are observable) or PCA (when factors are latent); (2) after eliminating it from $\bSigma$, employing standard techniques such as thresholding in the sparse covariance matrix literature to estimate $\bSigma_u$; (3) adding the two estimated components together.

First, let us consider the scenario where the factors $\{ \bff_i \}_{i=1}^n$ are observable. In this setting, we do not need the Identifiability Assumption \ref{ass:1}.
\cite{FFL08} focused on the strict factor model where the $\bSigma_u$ in (\ref{cov.fm}) is diagonal. It is then extended to the approximate factor model (\ref{fm}) by \cite{FLM11}. Later, \cite{FWZ16} relaxed the sub-Gaussian assumption on the data to moment condition, and proposed a robust estimator. We are going to present the main idea of these methods using the one in \cite{FLM11}.
%To facilitate presentation, we assume $\bmu=\mathbf{0}$ in (\ref{fm}).

\vspace{0.1in}
{\it Step 1}. Estimate $\bB$ using the ordinary least-squares: $\hat{\bB}=(\hat{\bb}_1,\ldots,\hat{\bb}_p)^{\top}$ where
\begin{equation*}
	(\hat{a}_j, \hat{\bb}_j) = \argmin_{a,\hat{\bb}} \frac{1}{n} \sum_{i=1}^{n} ( x_{ij} -a - \bb^{\top} \bff_i )^2.
\end{equation*}

\vspace{0.1in}
{\it Step 2}. Let $\hat{\ba} = (\hat{a}_1, \cdots, \hat{a}_p)^\top$ be the vector of intercepts, $\hat{\bu}_i = \bx_i -\hat{\ba} - \hat{\bB} \bff_i$ be the vector of residual for $i\in[n]$, and $\bS_u = \frac{1}{n} \sum_{i=1}^{n} \hat{\bu}_i \hat{\bu}^{\top}$ be the sample covariance. Apply thresholding to $\bS_u$ and obtain a regularized estimator $\hat{\bSigma}_u$.
\vspace{0.1in}

{\it Step 3}. Estimate $\cov(\bff_i)$ by
$\widehat{\cov}(\bff_i) = \frac{1}{n} \sum_{i=1}^{n} (\bff_i - \bar{\bff})(\bff_i - \bar{\bff})^{\top}$.
\vspace{0.1in}

{\it Step 4}. The final estimator is $\hat{\bSigma} = \hat{\bB} \widehat{\cov}(\bff_i) \hat{\bB}^{\top} + \hat{\bSigma}_u$.
\vspace{0.1in}

We remark that in Step 2, there are many thresholding rules for estimating sparse covariance matrices. Two popular choices are the
$t$-statistic-based adaptive thresholding \citep{CLi11} and
correlation-based adaptive thresholding \citep{FLM13}, with the entry-wise thresholding level chosen to be $\omega \asymp K\sqrt{\frac{\log p}{n}} $. As the sparsity pattern of correlation and covariance are the same and the correlation matrix is scale-invariant, one typically applies the thresholding on the correlation and then scales it back to the covariance. Except for the number of factors $K$, this coincides with the commonly-used threshold for estimating sparse covariance matrices.

While it is not possible to achieve better convergence of $\bSigma$ in terms of the operator norm or the Frobenius norm, \cite{FLM11} considered two other important norms. Under regularity conditions, it is shown that
\begin{equation}\label{ineq-factor-cov}
	\begin{split}
		& \| \hat{\bSigma} - \bSigma \|_{\bSigma} = O_{\P} \Bigl ( m_0 K \sqrt{\frac{\log p}{n}} + \frac{K\sqrt{p} \log p}{n} \Bigr ),\\
		&\| \hat{\bSigma} - \bSigma \|_{\max} = O_{\P} \Bigl ( K \sqrt{\frac{\log p}{n}} + K^2 \sqrt{\frac{\log n}{n}} \Bigr ).
	\end{split}
\end{equation}
Here for $\bA\in\R^{p\times p}$, $\|\bA\|_{\bSigma}$ and $\| \bA\|_{\max}$ refer to its entropy-loss norm $p^{-1/2}\|\bSigma^{-1/2}\bA\bSigma^{-1/2}\|_F$ and entry-wise max-norm $\max_{i,j}|A_{ij}|$. As is pointed out by \cite{FLM11} and \cite{WFa17}, they are relevant to portfolio selection and risk management. In addition, convergence rates for $\| \hat{\bSigma}^{-1} - \bSigma^{-1} \|_2$, $\| \hat{\bSigma}_u - \bSigma_u \|_2$ and $\| \hat{\bSigma}_u^{-1} - \bSigma_u^{-1} \|_2$ are also established.

Now we come to covariance estimation with latent factors. As is mentioned in Section \ref{sec:2.1}, the Pervasiveness Assumption \ref{ass:2} helps separate the low-rank part $\bB \bB^{\top}$ from the sparse part $\bSigma_u$ in (\ref{cov.fm}). %We assume the number of factors $K$ to be bounded, and impose the Identifiability Assumption \ref{ass:1}. Then (\ref{cov.fm}) is simplified to $\bSigma = \bB\bB^{\top} + \bSigma_u$.
\cite{FLM13} proposed a Principal Orthogonal complEment Thresholding (POET) estimator, motivated by the relationship between PCA and factor model, and the estimation of sparse covariance matrix $\bSigma_u$ in \cite{FLM11}. The procedure is described as follows.

\vspace{0.1in}
{\it Step 1}. Let $\bS = \frac{1}{n} \sum_{i=1}^{n} \bx_i \bx_i^{\top}$ be the sample covariance matrix, $\{ \hat{\lambda}_j \}_{j=1}^p$ be the eigenvalues of $\bS$ in non-ascending order, $\{ \hat{\bxi}_j \}_{j=1}^p$ be their corresponding eigenvectors.
\vspace{0.1in}

{\it Step 2}. Apply thresholding to $\bS_u = \bS - \sum_{j=1}^{K} \hat{\lambda}_j \hat{\bxi}_j \hat{\bxi}_j^{\top}$ and obtain a regularized estimator $\hat{\bSigma}_u$.
\vspace{0.1in}

{\it Step 3}. The final estimator is $\hat{\bSigma} = \sum_{j=1}^{K} \hat{\lambda}_j \hat{\bxi}_j \hat{\bxi}_j^{\top} + \hat{\bSigma}_u$.
\vspace{0.1in}

Here $K$ is assumed to be known and bounded to simplify presentation and emphasize the main ideas. The methodology and theory in \cite{FLM13} also allow using a data-driven estimate $\hat{K}$ of $K$. In Step 2 above we can choose from a large class of thresholding rules, and it is recommended to use the correlation-based adaptive thresholding. However, the thresholding level should be set to $\widetilde{\omega} \asymp \sqrt{\frac{\log p}{n}} + \frac{1}{\sqrt{p}} $. Compared to the level $\sqrt{\frac{\log p}{n}}$ we use in covariance estimation with observed factors, the extra term $1/\sqrt{p}$ here is the price we pay for not knowing the latent factors. It can be negligible when $p$ grows much faster than $n$. Intuitively, thanks to the Pervasiveness Assumption, the latent factors can be estimated accurately in high dimensions. \cite{FLM13} obtained theoretical guarantees for the POET that are similar to (\ref{ineq-factor-cov}). The analysis allows for general sparsity patterns of $\bSigma_u$ by considering $m_q$
%, rather than $m_0$ in \cite{FLM11},
as the measure of sparsity for $q\in [0,1)$.

Robust procedures handling heavy-tailed data are proposed and analyzed by \cite{FLW18,FWZ16}. In another line of research, \cite{LCF17} considered estimation of the covariance matrix of a set of targeted variables, when additional data beyond the variables of interest are available. By assuming a factor model structure, they constructed an estimator taking advantage of all the data and justified the information gain theoretically.

The Pervasiveness Assumption rules out the case where factors are weak and the leading eigenvalues of $\bSigma$ are not as large as $O(p)$. Shrinkage of eigenvalues is a powerful technique in this scenario. \cite{DGJ13} systematically studied the optimal shrinkage in spiked covariance model where all the eigenvalues except several largest ones are assumed to be the same. \cite{WFa17} considered the approximate factor model, which is more general, and proposed a new version of POET with shrinkage for covariance estimation.

\subsection{Principal component regression with random sketch}

Principal component regression (PCR), first proposed by \cite{Hot33} and \cite{Ken65}, is one of the most popular methods of dimension reduction in linear regression. It employs the principal components of the predictors $\bfm x_i$ to explain or predict the response $y_i$. Why do principal components, not other components, have more prediction power?   Here we offer an insight from the perspective of high-dimensional factor models.

The basic assumption is that the unobserved latent factors $\bff_i \in \R^K$  drive simultaneously the covariates via \eqref{fm} and responses, as shown in  Figure~\ref{Fig3}. As a specific example, we assume
\[
	y_i = {\btheta^*}^\top \bff_i + \varepsilon_i, \quad i = 1, \ldots, n, \quad \textnormal{or in matrix form,} \quad \by = \bF\btheta^* + \bvarepsilon,
\]
%Consider the following linear regression model:
%\beq
%	Y = \bX^\top \bbeta^* + \varepsilon,
%\eeq
where $\by = (y_1, \ldots, y_n) ^ \top$ and the noise $\bvarepsilon = (\varepsilon_1, \ldots, \varepsilon_n)^\top$ has zero means.  Since $\bff_i$ is latent and the covariate vector is high dimensional, we naturally infer the latent factors from the observed covariates via PCA.  This yields the PCR.

\begin{figure}[h]
\centering
	\includegraphics[clip, trim=0 0 0cm 0, width=.75\textwidth]{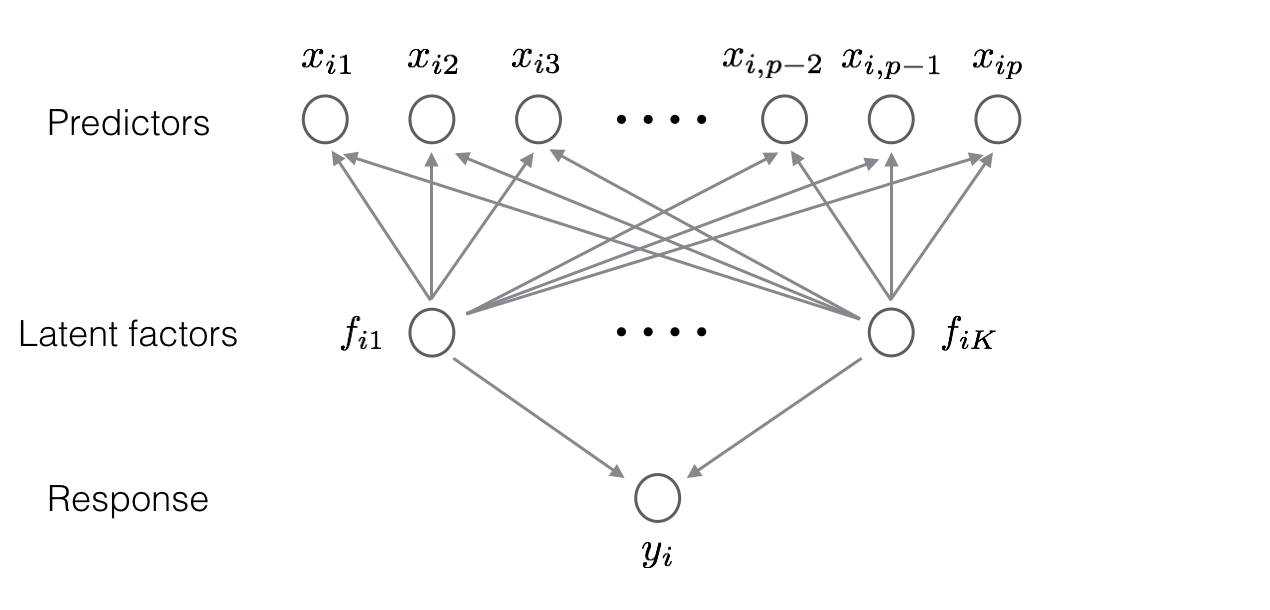}
	\caption{Illustration of the data generation mechanism in PCR. Both predictors $\bx_i$ and responses $y_i$ are driven by the latent factors $\bff_i$. PCR extracts latent factors via the principal components from $\bX$, and uses the resulting estimate $\hat \bF$ as the new predictor. Regressing $\by$ against $\hat \bF$ leads to the PCR estimator $\hat{\bfm \theta} \in \R^K$, which typically enjoys a smaller variance due to its reduced dimension, though it introduces bias. \label{Fig3}}
\end{figure}

By \eqref{eq2.2} (assume $\bmu = 0$ for simplicity),  $y_i \approx  (\bbeta^{\dag})^\top \bx_i + \varepsilon_i $, where $\bbeta^{\dag} := \bB(\bB^\top \bB)^{-1} \btheta^* \in \R^{p}$. This suggests that if we directly regress $y_i$ over $\bx_i$, then the regression coefficient $\bbeta^{\dag}$ should lie in the column space spanned by $\bB$. This inspires the core idea of PCR, i.e., instead of seeking the least square estimator in the entire $\RR^p$ space, we restrict our search scope to be the left leading singular space of $\bX$, which is approximately the column space of $\bB$ under the Pervasiveness Assumption.

Let us discuss PCR more rigorously.  To be consistent with the rest of this paper, we let $\bX \in \R^{p \times n}$, which is different from conventions, and
\beq
	y_i = \bx_i^\top \bbeta^*+ \varepsilon_i, \quad i = 1, \ldots, n, \quad \textnormal{or in matrix form}, \quad \by = \bX^\top \bbeta^* + \bvarepsilon.
\eeq
%where $y_i$ is the response, $\bbeta^* \in \R^p$ is the true regression vector and $\varepsilon_i$ is the noise that is uncorrelated with $\bx_i$, with mean $0$ and variance $\sigma ^ 2$.
Let $\bX = (\bx_1, \ldots, \bx_n) = \bP \bSigma \bQ^\top$ be the SVD of $\bX$, where $\bSigma = \diag(\sigma_1, \ldots, \sigma_{\min(n, p)})$ with non-increasing singular values. For some integer $ K$ satisfying $1 \le K \le \min(n, p)$, write $\bP = (\bP_K, \bP_{K+})$ and $\bQ = (\bQ_K, \bQ_{K+})$. The PCR estimator $\widehat \bbeta_K$ solves the following optimization problem:
\beq
	\label{eq:3.5}
	\widehat \bbeta_K := \argmin_{\bP_{K+}^\top \bbeta = \bzero}\ \ltwonorm{\by - \bX^\top \bbeta}.
\eeq
%The constraint $\bP_{K+}^\top \bbeta = \bzero$ forces $\widehat \bbeta_K$ to lie in the column space of $\bP_K$.
It is easy to verify that
\beq
	\label{eq:3.3}
	\widehat \bbeta_K = \bP_K \bSigma_K^{-1} \bQ_K^\top \by = \bP_K\bP_K^\top \bbeta^* + \bP_K \bSigma_K^{-1} \bQ_K^\top \bvarepsilon,
\eeq
where $\bSigma_K \in \R^{K \times K}$ is the top left submatrix of $\bSigma$. The following lemma calculates the excess risk of $\widehat\bbeta_K$, i.e., $\cE (\widehat \bbeta_K) := \E_{\bvarepsilon}[\ltwonorm{\bX^ \top \widehat \bbeta_K - \bX^\top \bbeta^*}^2 / n]$, treating $\bfm X$ as fixed. The proof is relegated to the appendix.

\begin{lem}
	\label{lem:3.1} Let $\bp_1,\ldots, \bp_{\min\{n,p\}} \in \R^p$ be the column vectors of $\bP$. For $j = 1, \ldots, p$, denote $\alpha_j = (\bbeta^*)^\top \bp_j$. We have
	\[
		\cE (\widehat \bbeta_K) = \frac{K \sigma^2}{n} + \sum\limits_{j = K + 1}^p \lambda_j^2 \alpha_j^2.
	\]
\end{lem}

Define the ordinary least squares (OLS) estimator $\widehat\bbeta := (\bX \bX^ \top)^{-1}\bX \by$. Note that $\cE(\widehat\bbeta) = \E_{\bvarepsilon} [\ltwonorm{\bX^\top \widehat\bbeta - \bX^ \top \bbeta^*}^2 / n]  = \min(n, p) \sigma^2 / n$. Comparing $\cE(\widehat\bbeta_K)$ and $\cE(\widehat\bbeta)$, one can clearly see a variance-bias tradeoff:
%(1) the first term $K\sigma^2/n$ is usually much smaller than $\min(n,p)\sigma^2/n$, since typically $K \ll \min(n,p)$---this is the advantage of reduced variance due to the constraint in \eqref{eq:3.5}; (2)
PCR reduces the variance by introducing a bias term $\sum_j \lambda_j^2 \alpha_j^2$, which is typically small and vanishes in the ideal case $\bP^\top_{K+}\bbeta^* = \bzero$ ---this is the bias incurred by imposing the constraint in \eqref{eq:3.5}.

In the high-dimensional setting where $p$ is large, calculating $\bP_K$ using SVD is computationally expensive. Recently, sketching has gained growing attention in statistics community and is used for downscaling and accelerating inference tasks with massive data. See recent surveys by \cite{Woo14} and \cite{YMM16}. The essential idea is to multiply the data matrix by a sketch matrix to reduce its dimension while still preserving the statistical performance of the procedure, since random projection reduces the strength of the idiosyncratic noise. To apply sketching to PCR, we first multiply the design matrix $\bX$ by an appropriately chosen matrix $\bR \in \R^{p \times m}$ with $K \le m < p$:
\beq
	\widetilde \bX := \bR^\top \bX,
\eeq
where $\bR$ is called the ``sketching matrix".  This creates $m$ indices based on $\bX$.  From the factor model perspective (assuming $\bmu = 0$), with a proper choice of $\bR$, we have $ \widetilde \bX \approx \bR^\top \bB \bF^{\top}$, since the idiosyncratic components in \eqref{fm} is averaged out due to weak dependence of $\bu$.  Hence, the indices in  $\widetilde \bX $ are approximately linear combinations of the factors $\{ \bff_i \}_{i=1}^n $.  At the same time, since $m \geq K$ and $\bR$ is nondegenerate, the row space of $\widetilde \bX $ is approximately the same as that spanned by $\bF^{\top}$.  This shows running linear regression on $\widetilde \bX $ is approximately the same as running it on $\bF^{\top}$, without using the computationally expensive PCA.

%where $\bR$ is a $p$-by-$m$ random matrix and $K < m < p$.
%Note that the sketched matrix $\widetilde \bX$ is an $m \times n$ matrix, and thus each column (observation) has a smaller dimension $m$ compared with $p$.
We now examine the property of sketching approach beyond the factor models.
Let $\widetilde \bX = \widetilde \bP \widetilde\bSigma \widetilde \bQ^\top$ be the SVD of $\widetilde\bX$, and write $\widetilde \bP = (\widetilde \bP_K, \widetilde\bP_{K+})$ and $\widetilde \bQ = (\widetilde\bQ_K, \widetilde\bQ_{K+})$. Imitating the form of  \eqref{eq:3.3}, we consider the following sketched PCR estimator:
\beq
	\label{eq:3.6}
	\widetilde \bbeta_K := \bR \widetilde \bP_K \widetilde \bSigma_K^{-1} \widetilde \bQ_K^\top \by,
\eeq
where $\widetilde \bSigma_K \in \R^{K \times K}$ is the top left submatrix of $\widetilde \bSigma$.

We now explain the above construction for $\widetilde \bbeta_K$. It is easy to derive from \eqref{eq:3.3} that given $\bR^\top \bX$ and $\by$ as the design matrix and response vector, the PCR estimator should be $\widetilde \bbeta^0_K := \widetilde \bP_K \widetilde \bSigma_K^{-1} \widetilde \bQ_K^\top \by$. Then the corresponding PCR projection of $\by$ onto $\bR^\top \bX$ should be $\bX^\top \bR \widetilde\bbeta^0_K = \bX^\top\bR \widetilde \bP_K \widetilde \bSigma_K^{-1} \widetilde \bQ_K^\top \by = \bX^ \top \widetilde\bbeta_K$. This leads to the construction of $\widetilde \bbeta_K$ in \eqref{eq:3.6}. Theorem 4 in \cite{MAv18} gives the excess risk of $\widetilde \bbeta_K$, which holds for any $\bR$ satisfying the conditions of the theorem.

\begin{thm}
	\label{thm:3.1}
	Assume $m \ge K $ and $\rank(\bR^\top \bX) \ge K$. If $\|\sin\Theta(\widetilde \bP_{K}, \bP_K)\|_2 \allowbreak \le \nu < 1$, then
    \beq
    	\cE(\widetilde \bbeta_K) \le \cE(\widehat \bbeta_K) + \frac{(2\nu + \nu ^ 2)\ltwonorm{\bX^ \top \bbeta ^ *} ^ 2}{n}.
    \eeq
\end{thm}

	This theorem shows that the extra bias induced by sketching is $(2\nu + \nu^2)\ltwonorm{\bX^\top \bbeta^*}^2 / n$. Given the bound of $\cE(\widehat \bbeta_K)$ in Lemma \ref{lem:3.1}, we can deduce that
	\[
		\cE(\widetilde \bbeta_K) \le \frac{K \sigma^2}{n} + \sum\limits_{j = K + 1}^p \alpha_j^2 \sigma_j^2 + \frac{(2\nu + \nu ^ 2)\ltwonorm{\bX^ \top \bbeta ^ *} ^ 2}{n}.
	\]
%	Under a typical setting where $\ltwonorm{\bX^\top \bbeta^*}^2 / n = O_{\mathbb{P}}(1)$, i.e., the average signal strength is $O_{\mathbb{P}}(1)$. Hence to ensure that the sketching does not hurt the statistical efficiency of $\widehat\bbeta_K$, Theorem \ref{thm:3.1} requires $\nu$ to be dominated by $\max(K\sigma^2 / n, \sum_{j = K + 1}^p \alpha_j^2 \sigma_j^2)$ .
As we will see below, a smaller $\nu$ requires a larger $m$, and thus more computation. Therefore, we observe a tradeoff between statistical accuracy and computational resources: if we have more computational resources, we can allow a large dimension of sketched matrix $\tilde \bX$, and the sketched PCR is more accurate, and vice versa.

One natural question thus arises: which $\bR$ should we choose to guarantee a small $\nu$ to retain the statistical rate of $\widehat\bbeta_K$? Recent results \citep{CNW15} on approximate matrix multiplication (AMM) suggest several candidate sketching matrices for $\bR$. Define the stable rank $\sr(\bX) := \fnorm{\bX} ^ 2 / \opnorm{\bX} ^ 2$, which can be interpreted as a soft version of the usual rank---indeed, $\sr(\bX) \le \rank(\bX)$ always holds, and $\sr(\bX)$ can be small if $\bX$ is approximately low-rank. An example of candidate sketching matrices for $\bR$ is a random matrix with independent and suitably scaled sub-Gaussian entries. As long as the sketch size $m = \Omega(\sr(\bX) + \log(1 / \delta) / \varepsilon ^ 2)$, it will hold for any $\varepsilon, \delta \in (0, 1 / 2)$ that
\beq
	\PP(\opnorm{\bX^\top \bR\bR^\top \bX - \bX^\top \bX} ^ 2 \ge \varepsilon \opnorm{\bX}^ 2) \le \delta.
\eeq
Combining this with the Davis-Kahan Theorem (Corollary~\ref{cor:dk}), we can deduce that $\opnorm{\sin \Theta(\widetilde \bP_K, \bP_K)}$ is small with certain eigen-gap condition. We summarize our argument by presenting a corollary of Theorem~9 in \cite{MAv18} below. Readers can find more candidate sketching matrices in the examples after Theorem~1 in \cite{CNW15}.

\begin{cor}
	For any  $\nu, \delta \in (0, 1 / 2)$, let $$\varepsilon = \nu (1 + \nu) ^ {-1} (\sigma_K^2 - \sigma_{K + 1}^2 ) / \sigma^ 2_1.$$ Let $\bR \in \RR^{p \times m}$ a random matrix with i.i.d. $\cN(0, 1 / m)$ entries. Then there exists a universal constant $C>0$ such that for any $\delta > 0$, if $m \ge C(\sr(\bX) + \log(1 / \delta) / \varepsilon ^ 2)$, it holds with probability at least $1 - \delta$ that
    \beq
    	\label{eq:3.9}
    	\cE(\widetilde \bbeta_K) \le \cE(\widehat \bbeta_K) + \frac{(2\nu + \nu ^ 2)\ltwonorm{\bX^\top \bbeta ^ *} ^ 2}{n}.
    \eeq
\end{cor}

\begin{rem}
	Note that $\varepsilon \le \nu(\sigma_K^2 - \sigma^2_{K + 1}) / \sigma_1^2$, and this bound is tight with a small $\nu$. Some algebra yields that \eqref{eq:3.9} holds when $$m = \Omega\Bigl(\sr(\bX) + \frac{ \sigma_1^2 \log(1 / \delta)}{\nu^2 (\sigma_K^2 - \sigma_{K + 1}^2)^2} \Bigr).$$ One can see that reducing $\nu$ requires a larger sketch size $m$. Besides, a large eigengap of the design matrix $\bX$ helps reduce the required sketch size.
\end{rem}

\subsection{Factor-Adjust Robust Multiple (FARM) tests}
Large-scale multiple testing is a fundamental problem in high-dimensional inference. In genome-wide association studies and many other applications, tens of thousands of hypotheses are tested simultaneously.
Standard approaches such as \cite{BHo95} and \cite{Sto02} can not control well both false and missed discovery rates in the presence of strong correlations among test statistics. Important efforts on dependence adjustment include \cite{efron2007correlation}, \cite{FKC09}, \cite{efron2010correlated}, and \cite{DSt12}.
%To incorporate the correlation effect in the testing procedure, both \cite{FKC09} and \cite{DSt12} assumed that the data are generated from a strict factor model with independent idiosyncratic errors and constructed an estimator for the false discovery proportion (FDP).
\cite{FHG12} and \cite{FHa17} considered FDP estimation under the approximate factor model. \cite{WZH17} studied a more complicated model with both observed variables and latent factors. All these existing papers heavily rely on the joint normality assumption of the data, which is easily violated in real applications. A recent paper \citep{FKS17} developed a factor-adjusted robust procedure that can handle heavy-tailed data while controlling FDP. We are going to introduce this method in this subsection.

Suppose our i.i.d.\ observations $\{ \bx_i \}_{i=1}^n$ satisfy the approximate factor model (\ref{fm}) where $\bmu\in\R^{p}$ is an unknown mean vector. To make the model identifiable, we use the Identifiability Assumption \ref{ass:1}. We are interested in simultaneously testing
\begin{equation*}
H_{0j}:\mu_j=0 \quad \text{versus}\quad H_{1j}:\mu_j\neq 0,\quad \text{for }j\in[p].
\end{equation*}
Let $T_j$ be a generic test statistic for $H_{0j}$. For a pre-specified level $z>0$, we reject $H_{0j}$ whenever $|T_j|\geq z$.
The numbers of total discoveries $R(z)$ and false discoveries $V(z)$ are defined as
$$
R(z)=\#\{ j: |T_j|\geq z \} \quad\text{and}\quad V(z)=\#\{ j: |T_j|\geq z,~\mu_j =0 \}.
$$
Note that $R(z)$ is observable while $V(z)$ needs to be estimated.
Our goal is to control the false discovery proportion $\mathrm{FDP}(z)=V(z)/R(z)$ with the convention $0/0=0$.

Na\"ive tests based on sample averages $\frac{1}{n}\sum_{i=1}^{n}\bx_i$ suffer from size distortion of FDP control due to dependence of common factors in \eqref{fm}.  On the other hand, the factor-adjusted test based on the sample averages of $\bx_i - \bB \bff_i$ ($\bB$ and $\bff_i$ need to be estimated) has two advantages: the noise $\bu_i$ is now weakly dependent so that FDP can be controlled with high accuracy, and the variance of $\bu_i$ is smaller than that of $\bB \bff_i + \bu_i$ in model \eqref{fm}, so that it is more powerful.
This will be convincingly demonstrated in Figure~\ref{Fig_FARM-Test} below. The factor-adjusted robust multiple test (FarmTest) is a robust implementation of the above idea \citep{FKS17}, which replaces the sample mean by its adaptive Huber estimation and extracts latent factors from a robust covariance input.

To begin with, we consider the Huber loss \citep{Hub64} with the robustification parameter $\tau \ge 0$:
\begin{equation*}
\ell_{\tau} (u) = \begin{cases}
u^2/2,&\mbox{ if }|u|\leq \tau\\
\tau |u|-\tau^2/2,&\mbox{ if }|u|>\tau
\end{cases},
\end{equation*}
and use $\hat{\mu}_j = \argmax_{\theta \in \R} \sum_{i=1}^{n} \ell_{\tau} ( x_{ij} - \theta )$ as a robust $M$-estimator of $\mu_j$. \cite{FKS17} suggested choosing $\tau \asymp \sqrt{ n / \log(np) }$ to deal with possible asymmetric distribution and called it adaptive Huber estimator.  They showed, assuming bounded fourth moments only, that
\begin{equation}
\sqrt{n} ( \hat{\mu}_j - \mu_j - \bb_j^{\top} \bar{\bff} ) = \cN(0,\sigma_{u,jj} ) + o_{\P} (1)
\text{ uniformly over }j\in[p],
\label{eqn-asymp-normality}
\end{equation}
where $\bar{\bff} = \frac{1}{n} \sum_{i=1}^{n} \bff_i$, and $\sigma_{u,jj}$ is the $(j,j)$th entry of $\bSigma_u$ as is defined in (\ref{cov.fm}). Assuming for now that $\{ \bb_j \}_{j=1}^p$, $\bar{\bff}$ and $\{ \sigma_{u,jj} \}_{j=1}^p$ are all observable, then the factor-adjusted test statistic $T_j=\sqrt{ n / \sigma_{u,jj} } ( \hat{\mu}_j - \bb_j^{\top} \bar{\bff} )$ is asymptotically $\cN(0,1)$. The law of large numbers implies that $V(z)$ should be close to $2p_0 \Phi(-z)$ for $z\geq 0$, where $\Phi ( \cdot )$ is the cumulative distribution function of $\cN(0,1)$, and $p_0=\# \{ j:\mu_j=0 \}$ is the number of true nulls. Hence
%$\mathrm{FDP} = V(z)/R(z) \approx 2p_0 \Phi(-z) / R(z)$.
%In the high dimensional and sparse regime, we have $p_0=p-o(p)$ and thus $\mathrm{FDP}^{\mathrm{A}} (z) = 2p \Phi(-z) / R(z)$ is a slightly conservative surrogate. If the proportion $\pi_0 = p_0/p$ is bounded away from 1, then we use $\mathrm{FDP}^{\mathrm{A}} (z) = 2p \hat{\pi}_0 \Phi(-z) / R(z)$ to estimate the FDP, where $\hat{\pi}_0$ is a estimate of $\pi_0$, see \cite{Sto02} and follow-up works.
%For simplicity, below we will just use $\mathrm{FDP}^{\mathrm{A}} (z) = 2p \Phi(-z) / R(z)$.
$$
\mathrm{FDP}(z) = \frac{V(z)}{R(z)} \approx \frac{2p_0 \Phi(-z)}{R(z)} \le \frac{2p \Phi(-z)}{R(z)} =: \mathrm{FDP}^{\mathrm{A}} (z).
$$
Note that in the high-dimensional and sparse regime, we have $p_0=p-o(p)$ and thus $\mathrm{FDP}^{\mathrm{A}} (z)$ is only a slightly conservative surrogate. However, we can also estimate the proportion $\pi_0 = p_0/p$ and use less conservative estimate $\mathrm{FDP}^{\mathrm{A}} (z) = 2p \hat{\pi}_0 \Phi(-z) / R(z)$ instead, where $\hat{\pi}_0$ is an estimate of $\pi_0$ whose idea is depicted in Figure~\ref{Fig4}; see \cite{Sto02}. Finally, we define the critical value $z_{\alpha} = \inf\{z\geq 0: \mathrm{FDP}^{\mathrm{A}} (z) \leq \alpha \}$ and reject $H_{0j}$ whenever $|T_j|\geq z_{\alpha}$.

\begin{figure}[h]
	\centering
	\includegraphics[clip, trim=0 0 0cm 0, width=.75\textwidth]{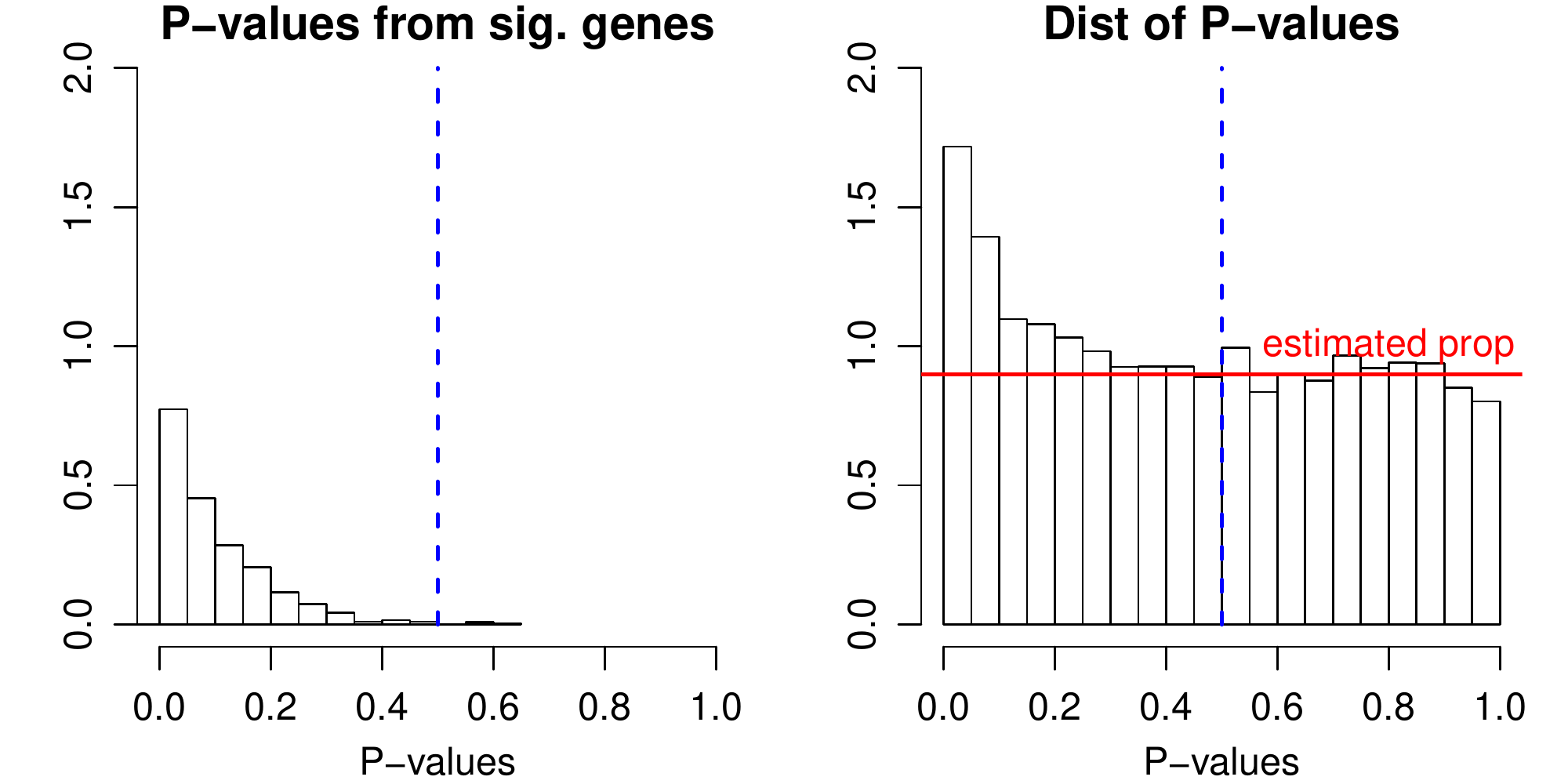}
	\caption{Estimation of proportion of true nulls.  The observed P-values (right panel) consist of those from significant variables (genes), which are usually small, and those from insignificant variables, which are uniformly distributed. Assuming the P-values for significant variables are mostly less than $\lambda$ (taken to be 0.5 in this illustration, left panel), the contributions of observed P-values $> \lambda$ are mostly from true nulls and this yields a natural estimator $\hat{\pi}_0(\lambda) = \frac{1}{(1-\lambda) p } \sum_{j=1}^p 1( \hat{P}_j > \lambda )$, which is the average height of the histogram with P-values $> \lambda$ (red line).  Note that the histograms above the red line estimates the distributions of P-values from the significant variavles (genes) in the left panel.\label{Fig4}}
\end{figure}

In practice, we have no access to $\{ \bb_j \}_{j=1}^p$, $\bar{\bff}$ or $\{ \sigma_{u,jj} \}_{j=1}^p$ in (\ref{eqn-asymp-normality}) and need to use their estimates. This results in the Factor-Adjusted Robust Multiple test (FarmTest) in \cite{FKS17}. The inputs include $\{ \bx_i \}_{i=1}^n$, a generic robust covariance matrix estimator $\hat{\bSigma} \in \R^{p\times p}$ from the data, a pre-specified level $\alpha\in(0,1)$ for FDP control , the number of factors $K$, and the robustification parameters $\gamma$ and $\{ \tau_j \}_{j=1}^p$. Note that $K$ can be estimated by the methods in Section \ref{sec-est-K}, and overestimating $K$ has little impact on final outputs.

\vspace{0.1in}
{\it Step 1}. Denote by $\hat{\bSigma} \in \R^{p\times p}$ a generic robust covariance matrix estimator. Compute the eigen-decomposition of $\hat \bSigma$, set $\{ \hat{\lambda}_j \}_{j=1}^K$ to be its top $K$ eigenvalues in descending order, and $\{ \hat{\bv}_j \}_{j=1}^K$  to be their corresponding eigenvectors. Let $\hat{\bB} = ( \tilde{\lambda}_1^{1/2} \hat{\bv}_1,\ldots,\tilde{\lambda}_K^{1/2} \hat{\bv}_K ) \in \R^{p\times K}$ where $\tilde{\lambda}_j = \max\{\hat{\lambda}_j,0\}$, and denote its rows by $\{ \hat{\bb}_j \}_{j=1}^p$.

\vspace{0.1in}
{\it Step 2}.
Let $\bar{x}_j=\frac{1}{n} \sum_{i=1}^{n} x_{ij}$ for $j\in[p]$ and $\hat{\bff} = \argmax_{\bff \in \R^K} \sum_{j=1}^{p} \ell_{\gamma} ( \bar{x}_j - \hat{\bb}_j^{\top}  \bff )$.
Construct factor-adjusted test statistics
\begin{equation}
T_j = \sqrt{ n / \hat{ \sigma } _{u,jj} } ( \hat{\mu}_j - \hat{\bb}_j^{\top} \hat{\bff} ) \quad\text{for } j\in[p],
\label{eqn-farm-test}
\end{equation}
where $\hat{\sigma}_{u,jj} = \hat{\theta}_j - \hat{\mu}_j^2 - \| \hat{\bb}_j \|_2^2$, $\hat{\theta}_j = \argmin_{\theta \geq \hat{\mu}_j^2 +  \| \hat{\bb}_j \|_2^2 } \ell_{\tau_j} ( x_{ij}^2 - \theta )$.

\vspace{0.1in}
{\it Step 3}. Calculate the critical value $z_{\alpha} = \inf\{z\geq 0: \mathrm{FDP}^{\mathrm{A}} (z) \leq \alpha \}$, where $\mathrm{FDP}^{\mathrm{A}} (z) = 2 \hat{\pi}_0 p \Phi(-z) / R(z)$, and reject $H_{0j}$ whenever $|T_j|\geq z_{\alpha}$.

\vspace{0.1in}
%The $\hat{\bff}^*$ in Step 3 is a least squares estimate from with $\bmu^*$ ignored.

In Step 2, we estimate $\bar{\bff}$ based on $\bar{x}_j = \mu_j + \hat{\bb}_j^{\top}  \bar{\bff} + \bar{u}_j$, which is implied by the factor model \eqref{fm}, and regard non-vanishing $\mu_j$ as an outlier.  In the estimation of  $\sigma_{u, jj}$, we used the identity $\theta_j := \E x_{ij}^2 = \mu_j^2 + \|\bb_j\|^2 + \sigma_{u, jj}$ and robustly estimated the second moment $\theta_j$.

Figure \ref{Fig_FARM-Test} is borrowed from Figure~1 in \cite{FKS17} that illustrates the effectiveness of this procedure. Here $n=100$, $p=500$, $K=3$, $\bff_i \sim \cN(\mathbf{0},\mathbf{I}_3)$ and the entries of $\bu_i$ are generated independently from the $t$-distribution with 3 degrees of freedom. It is known that $t-$distributions are not sub-Guassian variables and are often used to model heavy-tailed data. The unknown means $\mu \in \R^p$ are fixed as $\mu_j=0.6$ for $j \leq 125$ and $\mu_j=0$ otherwise. We plot the histograms of sample means, robust mean estimators, and their counterparts with factor-adjustment. The latent factors and heavy-tailed errors make it difficult to distinguish $\mu_j=0.6$ from $\mu_j=0$, and that explains why the sample means behave poorly. As is shown in Figure \ref{Fig_FARM-Test}, better separation can be obtained by  factor adjustment and robustification.
\begin{figure}[t]
	\centering
	\includegraphics[width=4.5 in]{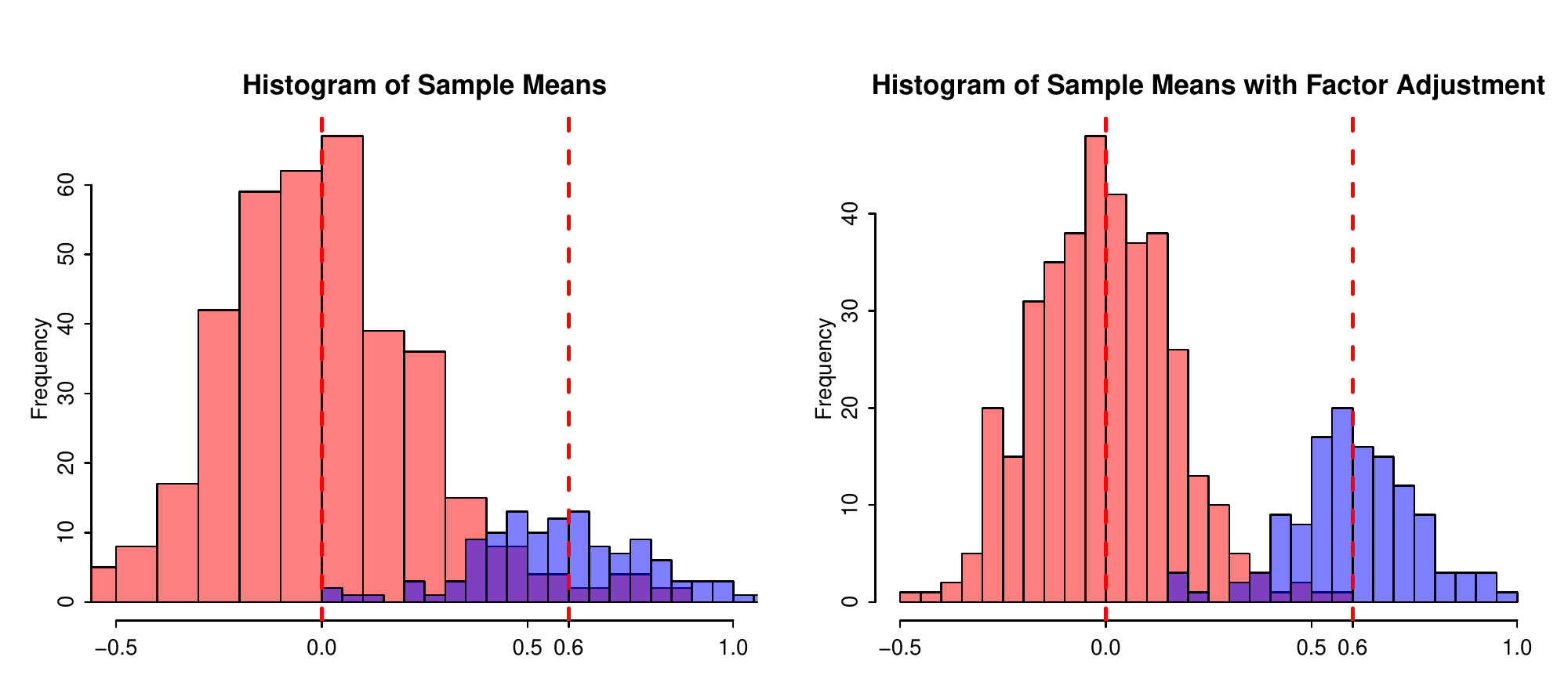}
	\includegraphics[width=4.5 in]{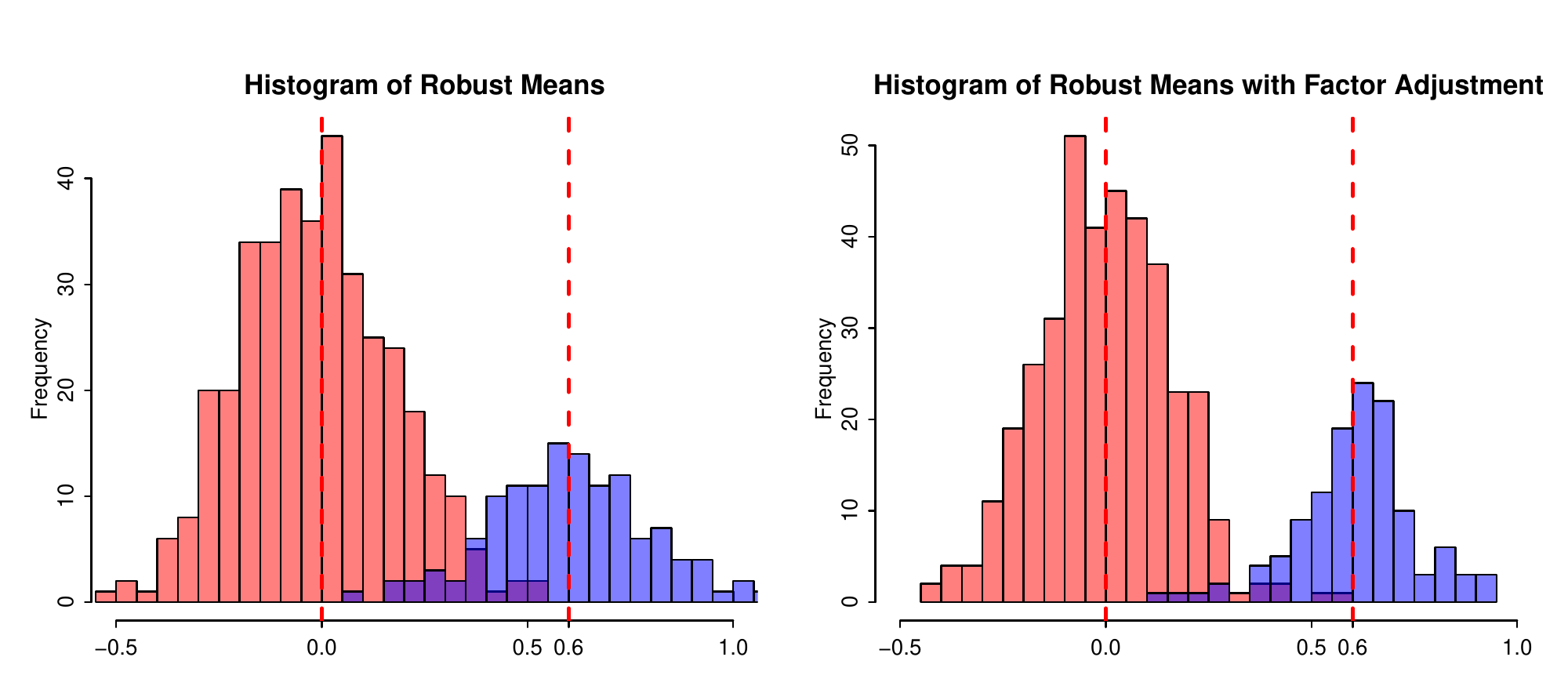}
	\caption{Histograms of four different mean estimators for simultaneous inference. Fix $n=100$, $p=500$ and $K=3$, and data are generated i.i.d.\ from $t_3$, which is heavy-tailed. Dashes lines correspond to $\mu_j=0$ and $\mu_j=0.6$, which is unknown. Robustification and factor adjustment help distinguish nulls and alternatives.}
	\label{Fig_FARM-Test}
\end{figure}

%In the above procedure, the estimation of factor model and construction of test statistics are based on the same set of samples. To avoid convoluted theoretical analysis and focus on the key idea, \cite{FKS17} proposed an alternative approach based on sample splitting. The samples $\{ \bx_i \}_{i=1}^n$ are first split into two equally-sized subsets. The Step 1 above is applied to the first half with the adaptive Huber method (see the Section 3.2.2 therein) and yields $\{ \hat{\bb}_j \}_{j=1}^{p}$. Then we proceed with the remaining steps using the second half of samples.

While existing literature usually imposes the joint normal assumption on $\{ \bff_i,\bu_i \}_{i=1}^n$, the FarmTest only requires the coordinates of $\{ \bu_{i} \}_{i=1}^n$ to have bounded fourth-order moments, and $\{ \bff_{i} \}_{i=1}^n$ to be sub-Gaussian.
Under standard regularity conditions for the approximate factor model, it is proved by \cite{FKS17} that
\begin{equation*}
\mathrm{FDP}^{\mathrm{A}}(z) - \mathrm{FDP}(z)=o_{\P}(1).
\end{equation*}
We see that $\mathrm{FDP}^{\mathrm{A}}$ is a valid approximation of FDP, which is therefore faithfully controlled by the FARM-Test.

\subsection{Factor-Adjusted Robust Model (FARM) selection}

Model selection is one of the central tasks in high dimensional data analysis.
%, and it is natural to asmany complex datasets have underlying models with sparse representations.
Parsimonious models enjoy interpretability, stability and oftentimes, better prediction accuracy. Numerous methods for model selection have been proposed in the past two decades, including, Lasso \citep{Tib96}, SCAD \citep{FLi01}, the elastic net \cite{ZHa05}, the Dantzig selector \citep{CTa07}, among others.
%since they are computationally attractive, and theoretical guarantees for their success have been established.
%, see \citep{MBu06,ZYu06,BRT09} and the references therein.
However, these methods work only when the covariates are weakly dependent or statisfy certain regularity conditions \citep{ZYu06, BRT09}.
When covariates are strongly correlated, \cite{PBH08,KSa11,Wan12,FKW16} used factor model to eliminate the dependencies caused by pervasive factors, and to conduct model selection using the resulting weakly correlated variables.

Assume that $\{ \bx_i \}_{i=1}^n$ follow the approximate factor model (\ref{fm}). %, with $\bmu=\mathbf{0}$ for simplicity.
As a standard assumption, the coordinates of $\bw_i=(\bff_i^{\top},\bu_i^{\top})^{\top} \in \R^{K + p}$ are weakly dependent. Thanks to this condition and the decomposition
\begin{align}
\bx_i^{\top} \bbeta= (\bmu+ \bB \bff_i + \bu_i )^{\top} \bbeta
= \alpha + \bu_i^{\top} \bbeta + \bff_i^\top \bgamma
%=
%\begin{pmatrix}
%\bff_i^{\top} & \bu_i^{\top}
%\end{pmatrix}
%\begin{pmatrix}
%\bB^{\top} \bbeta \\ \bbeta
%\end{pmatrix}
%=
%\bw_i^{\top}
%\begin{pmatrix}
%\bgamma \\ \bbeta
%\end{pmatrix},
    \label{eq3.12}
\end{align}
where $\alpha = \bmu^T \bbeta$ and $\bgamma = \bB^\top \bbeta$.
we may treat $\bw_i$ as the new predictors.  In other words, by lifting the number of variables from $p$ to $p+K$, the covariates of $\bw_i$ are now weakly dependent. The usual regularized estimation can now be applied to this new set of variables.  Note that we regard the coefficients $\bB^\top \bbeta$ as free parameters to facilitate the implementation (ignoring the relation $\gamma = \bB^\top \beta$) and this requires an additional assumption to make this valid \citep{FKW16}.

Suppose we wish to fit a  model $y_i = g(\bx_i^\top \bbeta, \varepsilon_i)$ via a loss function $L(y_i, \bx_i^\top \bbeta)$.
The above idea suggests the following two-step approach, which is called Factor-Adjusted Regularized (or Robust when so implemented) Model selection (FarmSelect) \citep{FKW16}.

\vspace{0.1in}
%\quad
{\it Step 1: Factor estimation}. Fit the approximate factor model (\ref{fm}) to get $\hat{\bB}$, $\hat{\bff}_i$ and $\hat{\bu}_i=\bx_i-\hat{\bB} \hat \bff_i $.

\vspace{0.1in}
%\quad
{\it Step 2: Augmented regularization}. Find $\alpha$, $\bbeta$ and $\bgamma$ to minimize
$$
    \sum_{i=1}^n L(y_i, \alpha + \hat{\bu}_i^{\top}\bbeta + \hat{\bff}_i^\top \bgamma) + \sum_{j=1}^p p_\lambda(|\beta_j|),
$$
where $p_\lambda(\cdot)$ is a folded concave penalty \citep{FLi01} with parameter $\lambda$.

%Set $\hat{\bW}=(\hat{\bF},\hat{\bU}^\top)^\top \in \R^{(K+p) \times n}$ and find
%\begin{align*}%\label{L1-ppmle}
%\hat{\btheta} \in \argmin_{\btheta\in \R^{K+p}}
%\left\{ L (\by, \hat{\bW}^\top \btheta)+\lambda \| \btheta_{-(1:K)} \|_1 \right\},
%\end{align*}
%where $\btheta_{-(1:K)} \in \R^p$ refers to the last $p$ entries of $\btheta \in \R^{K+p}$. Let $\hat{\bbeta} = \hat{\btheta}_{-(1:K)}$ as our estimator for $\bbeta^*$.

\vspace{0.1in}

In Step 1, standard estimation procedures such as POET \citep{FLM13} and S-POET \citep{WFa17} can be applied, as long as they produce consistent estimators of $\bB$, $\{ \bff_i \}_{i=1}^n$ and $\{ \bu_i \}_{i=1}^n$. Step 2 is carried out using usual regularization methods with new covariates.

Figure \ref{Fig_FARM-Select}, borrowed from Figure~3~(a) in \cite{FKW16}, shows that the proposed method outperforms other popular ones for model selection including Lasso \citep{Tib96}, SCAD \citep{FLi01} and elastic net \citep{ZHa05}, in the presence of correlated covariates. The basic setting is sparse linear regression $y = \bx^{\top} \bbeta^*+\varepsilon$ with $p=500$ and $n$ growing from $50$ to $160$. The true coefficients are $\bbeta^*=(\beta_1,\cdots,\beta_{10},\mathbf{0}_{p-10})^{\top}$, where $\{ \bbeta_j \}_{j=1}^{10}$ are drawn uniformly at random from $[2,5]$, and $\varepsilon\sim  \mathcal{N}(0,0.3)$. The correlation structure of covariates $\bx$ is calibrated from S\&P 500 monthly excess returns between 1980 and 2012.
\begin{figure}[htbp]
	\centering
	\includegraphics[width=5 in]{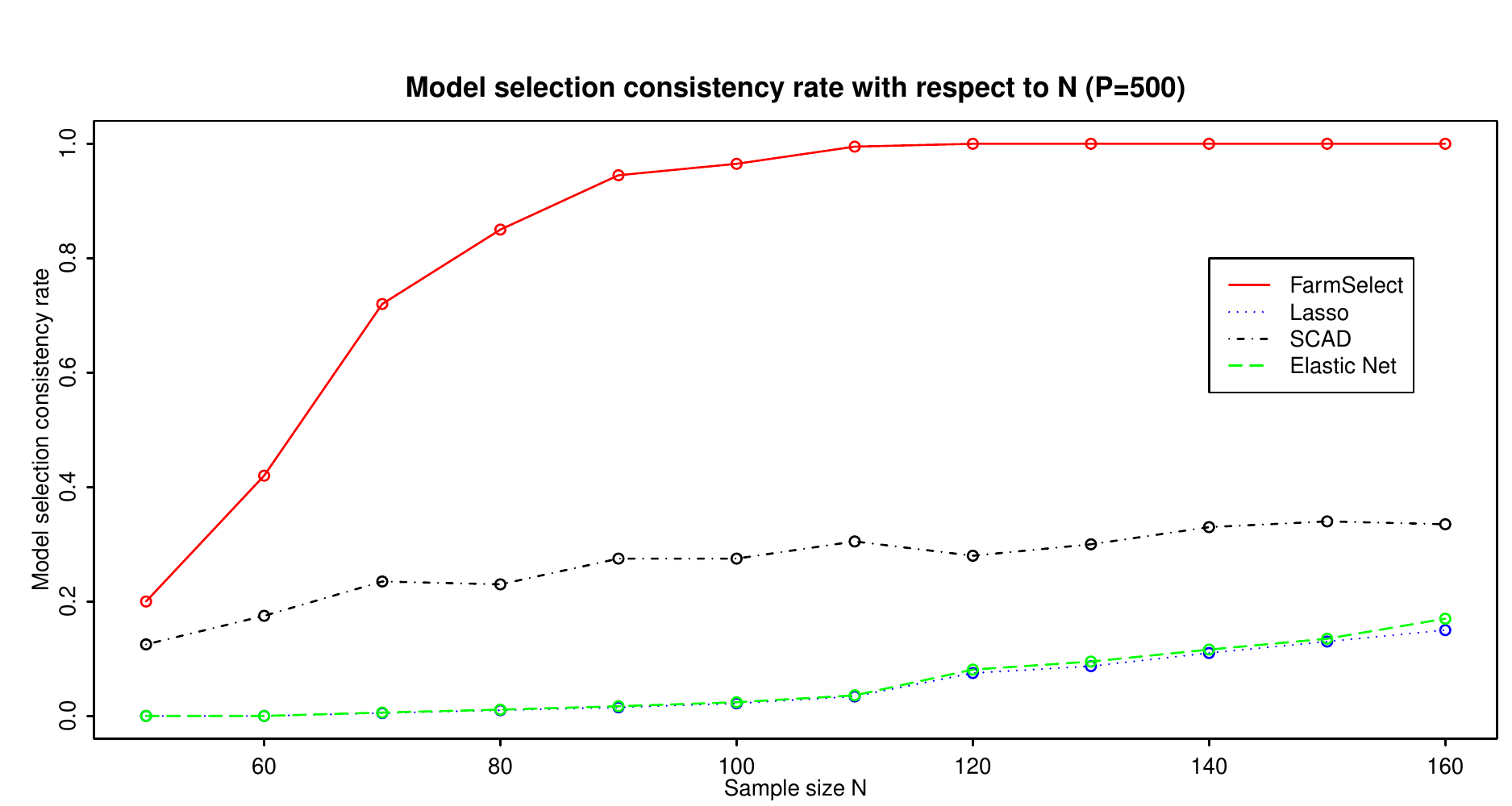}
	\caption{Model selection consistency rate, i.e.,\ the proportion of simulations that the selected model is identical to the true one, with $p = 500$ and $n$ varying from $50$ to $160$. With moderate sample size, the proposed method faithfully identifies the correct model while other methods cannot.
		}
	\label{Fig_FARM-Select}
\end{figure}

Under the generalized linear model, $L(y,z)= -yz + b( z ) $ and $b(\cdot)$ is a convex function. \cite{FKW16} analyzed theoretical properties of the above procedure. As long as the coordinates of $\bw_i$ (rather than $\bx_i$) are not too strongly dependent and the factor model is estimated to enough precision, $\hat{\bbeta}$ enjoys optimal rates of convergence $\|\hat{\bbeta} - \bbeta^* \|_q =O_{\P}( |S|^{1/q} \sqrt{\log p / n} )$, where $q=1$, $2$ or $\infty$. When the minimum entry of $|\bbeta^*|$ is at least $\Omega(\sqrt{\log p / n})$, the model selection consistency is achieved.

When we use the square loss, this method reduces to the one in \cite{KSa11}. By using the square loss and replacing the penalized multiple regression in Step 2 with marginal regression, we recover the factor-profiled variable screening method in \cite{Wan12}. While these papers aim at modeling and then eliminating the dependencies in $\bx_i$ via (\ref{fm}), \cite{PBH08} used a factor model to characterize the joint distribution of $(y_i,\bx_i^{\top})^{\top}$ and develops a related but different approach.

\section{Related learning problems}
\label{sec:4}

\subsection{Gaussian mixture model}

PCA, or more generally, spectral decomposition, can be also applied to learn mixture models for heterogeneous data. A thread of recent papers \citep{HKa13,AGH14, YiCarSan16, SedJanAna16} apply spectral decomposition to lower-order moments of the data to recover the parameters of interest in a wide class of latent variable models. Here we use the Gaussian mixture model to illustrate their idea. Consider a mixture of $K$ Gaussian distributions with spherical covariances. Let $w_k\in (0,1)$ be the probability of choosing component $k\in \{1,\ldots, K\}$, and $\{\bmu_1, \ldots, \bmu_k\}\subseteq \RR^p$ be the component mean vectors, and $\{\sigma_k^2\bI_p\}_{k=1}^K$ be the component covariance matrices, which is required by \cite{HKa13} and \cite{AGH14}. Each data vector
%$\bx$ is generated through
%\beq\label{eq:mix.mean}
%\bx = \bmu_h + \bz_h,
%\eeq
%where $h$ is a discrete random variable such that $\P(h=j)=w_j$ for $j=1,\ldots, K$, and $\bz_h \sim \cN(\mathbf{0}, \sigma_h^2 \bI_p)$.  In other words,
$\bx \sim w_1 \cN(\bmu_1, \sigma_1^2 \bI_p) + \cdots + w_K \cN(\bmu_K, \sigma_K^2 \bI_p)$ follows the mixture of the Gaussian distribution. The parameters of interest are $\{w_k, \bmu_k, \sigma^2_k\}_{k = 1} ^K$.

\cite{HKa13} and \cite{AGH14} shed lights on the close connection between the lower-order moments of the data and the parameters of interest, which motivates the use of Method of Moments (MoM). Denote the population covariance $\E[(\bx - \E \bx)(\bx - \E \bx)^\top]$ by $\bSigma$. Below we present Theorem 1 in \cite{HKa13} to elucidate the moment structure of the problem.

\begin{thm}
	\label{thm:4.1}
	Suppose that $\{\bmu_k\}_{k=1}^K$ are linearly independent. Then the average variance $\sigma_{\textnormal{ave}} ^ 2 := K ^ {-1} \sum\limits_{k = 1}^K \sigma_k^2$ is the smallest eigenvalue of $\bSigma$. Let $\bv$ be any eigenvector of $\bSigma$ that is associated with the eigenvalue $\sigma_{\textnormal{ave}} ^ 2$. Define the following quantities:
    	\begin{align*}
        	\bM_1 & = \E[(\bv^\top (\bx - \E \bx)) ^ 2 \bx] \in \RR ^ p, \\
            \bM_2 & = \E [\bx \otimes \bx] - \sigma_{\textnormal{ave}}^2 \cdot \bI_p \in \RR ^ {p \times p}, \\
            \bM_3 & = \E [\bx \otimes \bx \otimes \bx] - \sum\limits_{j = 1}^p (\bM_1 \otimes \be_j \otimes \be_j + \be_j \otimes \bM_1 \otimes \be_j + \be_j \otimes \be_j \otimes \bM_1) \in \RR ^ {p \times p \times p}.
        \end{align*}
    Then we have
    \beq
    	\label{eq:moment}
    	\bM_1 = \sum\limits_{k = 1}^K w_k \sigma_k^2 \bmu_k,\quad \bM_2 = \sum\limits_{k = 1}^K w_k \bmu_k \otimes \bmu_k,\quad \bM_3 = \sum\limits_{k = 1} ^ K w_k \bmu_k \otimes \bmu_k \otimes \bmu_k,
    \eeq
    where the notation $\otimes$ represents the tensor product.
\end{thm}

%, which can be understood as a way to generalize matrices to multi-dimensional arrays called tensors. Recall that the general routine of the MoM goes as follows.
%\begin{enumerate}
%	\item Construct moment equations to relate parameters of interest and data moments. \\
%	\item Use empirical moments to replace population moments. \\
%	\item Estimate unknown parameters by solving moment equations (ideally with fast computation).
%\end{enumerate}
Theorem \ref{thm:4.1} gives the relationship between the moments of the first three orders of $\bx$ and the parameters of interest. With  $\{\bM_i\}_{i=1}^3$ replaced by their empirical versions, the remaining task is to solve for all the parameters of interest via \eqref{eq:moment}. \cite{HKa13} and \cite{AGH14} proposed a fast method called \textit{robust tensor power method} to compute the estimators. The crux therein is to construct an estimable third-order tensor $\tilde \bM_3$ that can be decomposed as the sum of orthogonal tensors based on $\bmu_k$. This orthogonal tensor decomposition can be regarded as an extension of spectral decomposition to third-order tensors (simply speaking, three-dimensional arrays). Then the power iteration method is applied to the estimate of $\tilde \bM_3$ to recover each $\bmu_k$, as well as other parameters.

Specifically, consider first the following linear transformation of $\bmu_k$:
\beq
	\label{eq:4.3}
	\tilde \bmu_k:= \sqrt{\omega_k}\, \bW^\top\bmu_k
\eeq
for $k \in [K]$, where $\bW \in \RR^{p \times K}$. The key is to use the \textit{whitening} transformation by setting $\bW$ to be a square root of $\bfm M_2$.
This ensures that $\{\tilde \bmu_k\}_{k=1}^K$ are orthogonal to each other. Denoting $\ba^{\otimes 3} = \ba \otimes \ba \otimes \ba$,
\beq
	\label{eq:4.4}
	\tilde \bM_3 := \sum\limits_{k = 1}^K \omega_k(\bW^\top \bmu_k)^{\otimes 3} = \sum\limits_{k=1}^K \frac{1}{\sqrt{\omega_k}} \tilde \bmu_k^{\otimes 3} \in \R^{K \times K \times K}
\eeq
is an orthogonal tensor decomposition; that is, it satisfies orthogonality of $\{\tilde \bmu_k\}_{k=1}^K$. The following theorem from \cite{AGH14} summarizes the above argument, and more importantly, it shows how to obtain $\bmu_k$ back from $\tilde \bmu_k$.
\begin{thm}
	\label{thm:4.2}
	Suppose the vectors $\{\bmu_k\}_{k=1}^K$ are linearly independent, and the scalars $\{\omega_k\}_{k=1}^K$ are strictly positive. Let $\bM_2 = \bU\bD\bU^\top$ be the spectral decomposition of $\bM_2$ and let $\bW = \bU\bD^{-1/2}$. Then $\{\tilde \bmu_k\}_{k=1}^K$ in \eqref{eq:4.3} are orthogonal to each other. Furthermore, the Moore-Penrose pseudo-inverse of $\bW$ is $\bW^{\dag} := \bD^{1/2} \bU^\top \in \R^{K \times p}$, and we have $\bmu_k = (\bW^{\dag})^\top \tilde \bmu_k / \sqrt{\omega_k}$ for $k \in [K]$.
\end{thm}

%Now that the population structure is clear, we can construct the sample moments by replacing expectation with sample averages.
As promised, the orthogonal tensor $\tilde \bM_3$ can be estimated from empirical moments. We will make use of the following identity, which is similar to Theorem~\ref{thm:4.1}.
%through empirical moments of $\bW^\top\bx$. Indeed, similar to $\bM_3$,

\beq
	\label{eq:4.5}
		\tilde \bM_3 = \E [(\bW^\top \bx)^{\otimes 3} ] - \sum\limits_{j = 1}^p \sum\limits_{\textnormal{cyc}} (\bW^\top \bM_1) \otimes (\bW^\top \be_j) \otimes (\bW^\top \be_j),
\eeq
%\[
%	\begin{aligned}
%		\widetilde \bM_3 = \E [(\bW^\top \bx)^{\otimes 3} ] & - \sum_{j = 1}^p \bigl( (\bW^\top \bM_1) \otimes \bW_j \otimes \bW_j + \bW_j \otimes (\bW^\top \bM_1) \otimes \bW_j) \\
%		& - \bW_j \otimes \bW_j \otimes (\bW^\top \bM_1) \bigr),
%	\end{aligned}
%\]
where
%$\bW_j$ is the $j$th row of $\bW$ and
we used the cyclic sum notation $$\sum\limits_{\textnormal{cyc}} \ba \otimes \bb \otimes \bc := \ba \otimes \bb \otimes \bc + \bb \otimes \bc \otimes \ba + \bc \otimes \ba \otimes \bb.$$
Note that $\bW^\top \be_j \in \R^K$ is simply the $j$th row of $\bW$. To obtain an estimate of $\tilde \bM_3$, we replace the expectation $\E$ by the empirical average, and substitute $\bW$ and $\bM_1$ by their plug-in estimates. It is worth mentioning that, because $\tilde \bM_3$ has a smaller size than $\bM_3$, computations involving $\tilde \bM_3$ can be implemented more efficiently.

Once we obtain an estimate of $\tilde \bM_3$, which we denote by $\overline \bM_3$,
to recover $\{\bmu_k\}_{k=1}^K, \{\omega_k\}_{k=1}^K$ and $\{\sigma_k^2\}_{k=1}^K$, the only task left is computing the orthogonal tensor decomposition \eqref{eq:4.4} for $\overline \bM_3$. The tensor power method in \cite{AGH14} is shown to solve this problem with provable computational guarantees. We omit the details of the algorithm here. Interested readers are referred to Section~5 of \cite{AGH14} for the introduction and analysis of this algorithm.

To conclude this subsection, we summarize the entire procedure of estimating $\{\bmu_k, \sigma_k, \omega_k\}_{k = 1}^K$ as below.

\vspace{0.1in}

{\it Step 1.} Calculate the sample covariance matrix $\widehat\bSigma := n^{-1} \sum\limits_{i=1}^n (\bx_i - \overline\bx) (\bx_i - \overline \bx)^\top$, its minimum eigenvalue ${\widehat{\sigma}_{\textnormal{ave}}}^2$ and its associated eigenvector $\widehat\bv$.

\vspace{0.1in}

{\it Step 2.} Derive the estimators $\widehat\bM_1, \widehat\bM_2, \widehat\bM_3$ based on Theorem \ref{thm:4.1} by plug-in of empirical moments of $\bx$, $\widehat\bv$ and $\widehat\sigma^2_{\textnormal{ave}}$.

\vspace{0.1in}

{\it Step 3.} Calculate the spectral decomposition $\widehat\bM_2 = \widehat\bU \widehat\bD \widehat \bU^\top$. Let $\widehat \bW = \widehat \bU \widehat \bD^{-1/2}$. Construct an estimator of $\tilde \bM_3$, denoted by $\overline \bM_3$, based on \eqref{eq:4.5} by plug-in of empirical moments of $\widehat\bW^\top \bx$, $\widehat\bW$ and $\widehat\bM_1$. Apply the robust tensor power method in \cite{AGH14} to $\overline \bM_3$ and obtain $\{ \overline{\bmu}_k\}_{k=1}^K$ and $\{\widehat\omega_k\}_{k=1}^K$.

\vspace{0.1in}

{\it Step 4.} Set $\hat{\bW}^{\dag} = \hat \bD^{1/2} \hat \bU^\top$ and $\widehat \bmu_k = (\hat{\bW}^{\dag})^\top \overline{\bmu}_k / \sqrt{\widehat\omega_k}$. Solve the linear equation $\widehat \bM_1 = \sum\limits_{k= 1}^K \widehat\omega_k\widehat\sigma^2_k\widehat \bmu_k$ for $\{\widehat\sigma^2_k\}_{k=1}^K$.

\subsection{Community detection}\label{sec:comdet}
In statistical modeling of networks, the stochastic block model (SBM), first proposed by \cite{Hol83}, has gained much attention in recent years (see \citealp{Abb17} for a recent survey). Suppose our observation is a graph of $n$ vertices, each of which belongs to one of $K$ communities (or blocks). Let the vertices be indexed by $[n]$, and the community that vertex $i$ belongs to is indicated by an unknown $\bfm \theta_i \in [K]$. In SBM, the probability of an edge between two vertices depends entirely on the membership of the communities. To be specific, let $\bfm W \in \R^{K \times K}$ be a symmetric matrix where each entry takes value in $[0,1]$, and let $\bfm A \in \R^{n \times n}$ be the adjacency matrix, i.e., $ A_{ij} = 1$ if there is an edge between vertex $i$ and $j$, and $ A_{ij} = 0$ otherwise. Then, the SBM assumes
\begin{equation*}
\mathbb{P} (A_{ij} = 1) = W_{k \ell} \quad \text{with} \quad \theta_i = k, \theta_j = \ell
\end{equation*}
and $\{A_{ij}\}_{i > j}$ are independent.
Here, for ease of presentation, we allow self-connecting edges.  Figure~\ref{fig:heatmap} gives one realization of the network with two communities.

\begin{figure}[t]
\centering
\includegraphics[scale=0.4]{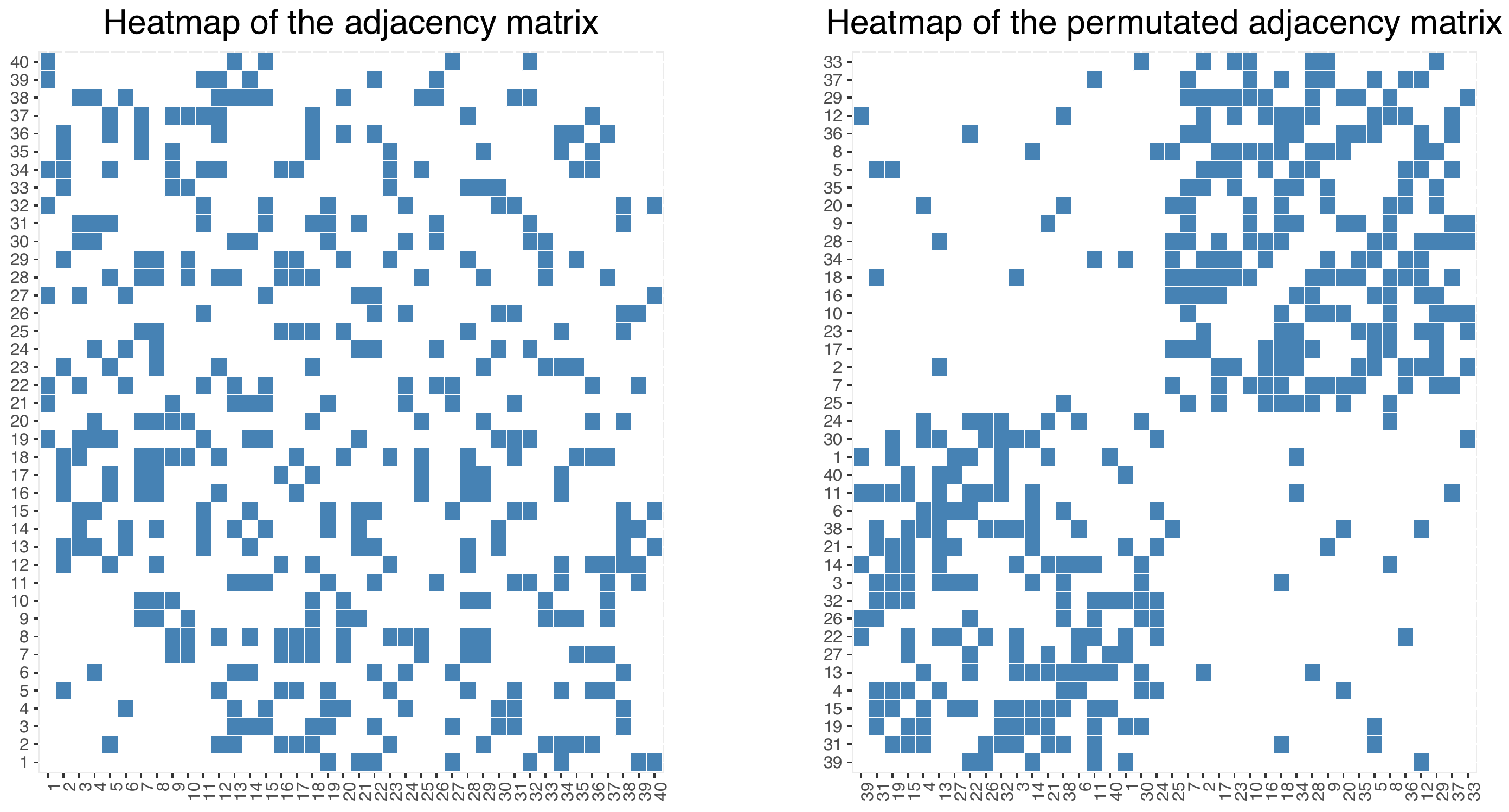}
\caption{In both heatmaps, a dark pixel represents an entry with value $1$ in a matrix, and a white pixel represents an entry with value $0$. The \textbf{left heatmap} shows the (observed) adjacency matrix $\bfm A$ of size $n = 40$ generated from the SBM with two equal-sized blocks ($K=2$), with edge probabilities $5\log n/n$ (within blocks) and $\log n/(4n)$ (between blocks). The \textbf{right heatmap} shows the same matrix with its row indices and column indices suitably permutated based on unobserved $z_i$. Clearly, we observe an approximate rank-$2$ structure in the right heatmap. This motivates estimating $z_i$ via the second eigenvector.}\label{fig:heatmap}
\end{figure}

Though seemingly different, this problem shares a close connection with PCA and spectral methods. Let $\bfm z_i = \bfm e_k$ (namely, the $k$th canonical basis in $\R^K$) if $\theta_i = k$, indicating the membership of $i$th node, and define $\bfm Z = [\bz_1,\ldots,\bz_n]^\top \in \R^{n \times K}$. The expectation of $\bfm A$ has a low-rank decomposition $\E \bfm A = \bfm Z \bfm W \bfm Z^\top$ and
\begin{equation}\label{eq:comdet}
\bfm A = \bfm Z \bfm W \bfm Z^\top + (\bfm A - \E \bfm A) .
\end{equation}
Loosely speaking, the matrix $\bfm Z$ plays a similar role as factors or  loading matrices (unnormalized), and $\bfm A - \E \bfm A$ is similar to the noise (idiosyncratic component). In the ideal situation, the adjacency matrix $\bfm A$ and its expectation are close, and naturally we expect the eigenvectors of $\bfm A$ to be useful for estimating $\theta_i$. Indeed, this observation is the underpinning of many methods \citep{RohChaYu11, Gao15, ASa15}. The vanilla spectral method for network/graph data is as follows:

\vspace{0.1in}

{\it Step 1.} Construct the adjacency matrix $\bfm A$ or other similarity-based matrices;

\vspace{0.1in}

{\it Step 2.} Compute eigenvectors $\bv_1,\ldots,\bv_L$  corresponding to the largest  eigenvalues, and form a matrix $\bfm V = [\bv_1,\ldots,\bv_\ell] \in \R^{n \times L}$;

\vspace{0.1in}

{\it Step 3.} Run a clustering algorithm on the row vectors of $\bfm V$.

\vspace{0.1in}

There are many variants and improvements of this vanilla spectral method. For example, in Step~1, very often the graph Laplacian $\bfm D - \bfm A$ or normalized Laplacian $\bfm D^{-1/2}(\bfm D - \bfm A) \bfm D^{-1/2}$ is used in place of the adjacency matrix, where $\bfm D = \diag(d_1,\ldots,d_n)$, and $d_i = \sum_{j} A_{ij}$ is the degree of vertex $i$. If real-valued similarities or distances between vertices are available, weighted graphs are usually constructed. In Step~2, there are many other refinements over raw eigenvectors in the construction of $\bV$, for example, projecting row vectors of $\bV$ onto the unit sphere \citep{Ng02}, and calculating scores based on eigenvector ratios \citep{Jin15}, etc. In Step~3, a very popular algorithm for clustering is the $K$-means algorithm.

We will look at the vanilla spectral algorithm in its simplest form. Our goal is \textit{exact recovery}, which means finding an estimator $\widehat{\btheta}$ of $\btheta = (\theta_1,\ldots,\theta_n)^\top$ such that as $n \to \infty$,
\[
\mathbb{P}(
\text{there exists a permutation }\pi \text{ of } [K] \text{ s.t. } \widehat{\theta}_i = \pi(\theta_i), \forall\, i \in[n]) = 1-o(1).
\]
Note that we can only determine $\btheta$ up to a permutation since the distribution of our observation is invariant to permutations of $[K]$. There are nice theoretical results, including information limits for exact recovery  in \cite{ABH16}.

Despite its simplicity, spectral methods can be quite sharp for exact recovery in SBM, which succeed in a regime that matches the information limit. The next theorem from \cite{AFWZ17} will make this point clear. Consider the SBM with two balanced blocks, i.e., $K=2$ and $\{i: \theta_i = 1\} = \{i: \theta_i = 2\} = n/2$, and suppose $W_{11} = W_{22} = a\log n/n$, $W_{12} = b\log n/n$ where $a >  b >0$. In this case, one can easily see that the second eigenvector of $\E \bfm A$ is given by $\bv^*_2 $ whose $i$th entry is given by $1/\sqrt{n}$ if
$\theta_i = 1$ and $-1$ otherwise. In other words, $\sgn(\bv^*_2)$ classifies the two communities, where $\sgn(\cdot)$ is the sign function applied to each entry of a vector.
This is shown in Figure~\ref{fig:SBMeig} for the case that $\# \{i: \theta_i=1\} = 2500$ (red curve, left panel), where the second eigenvector $\bfm v_2$ of $\bfm A$ is also depicted (blue curve).
The entrywise closeness between these two quantities is guaranteed by the perturbation theory under $\ell_\infty$-norm \citep{AFWZ17}.

\begin{thm}\label{thm:SBM}
Let $\bfm v_2$ be the normalized second eigenvector of $\bfm A$. If $\sqrt{a} - \sqrt{b} < \sqrt{2}$, then no estimator achieves exact recovery; if $\sqrt{a} - \sqrt{b} > \sqrt{2}$, then both the maximum likelihood estimator and the eigenvector estimator $\sgn(\bfm v_2)$ achieves exact recovery.
\end{thm}

The proof of this result is based on entry-wise analysis of eigenvectors in a spirit similar to Theorem~\ref{thm:pert-sym}, together with a probability tail bound for differences of binomial variables.

\subsection{Matrix completion}\label{sec:mc}

In recommendation systems, an important problem is to estimate users' preferences based on history data. Usually, the available data per user is very small compared with the total number of items (each user sees only a small number of movies and buys only a small fraction of books, comparing to the total). Matrix completion is one formulation of such problem.

The goal of (noisy) matrix completion is to estimate a low-rank matrix $\bfm M^* \in \R^{n_1 \times n_2}$ from noisy observations of some entries ($n_1$ users and $n_2$ items). Suppose we know $\rank(\bfm M^*) = K$. For each $i \in [n_1]$ and $j \in [n_2]$, let $I_{ij}$ be i.i.d.\ Bernoulli variable with $\mathbb{P}(I_{ij} = 1) = p$ that indicates if we have observed information about the entry $M^*_{ij}$, i.e., $I_{ij} = 1$ if and only if it is observed. Also suppose that our observation is $M_{ij} = M_{ij}^* + \varepsilon_{ij}$ if $I_{ij} = 1$, where $\varepsilon_{ij}$ is i.i.d.\ $\cN(0,\sigma^2)$ jointly independent of $I_{ij}$.

One natural way to estimate $\bfm M^*$ is to solve
\begin{equation*}
\min_{\bfm X \in \R^{n_1 \times n_2}}\frac{1}{2} \left\| \mathcal{P}_\Omega(\bfm M) - \mathcal{P}_\Omega(\bfm X) \right\|^2 \quad \text{subject to} \quad \rank(\bfm X) = K,
\end{equation*}
where $\mathcal{P}_\Omega : \R^{n_1 \times n_2} \to \R^{n_1 \times n_2}$ is the sampling operator defined by $[\mathcal{P}_\Omega(\bfm X)]_{ij} = I_{ij} X_{ij},\; \forall i,j$. The minimizer of this problem is essentially the MLE for $M^*$. Due to the nonconvex constraint $\rank(\bfm X) = K$, it is desirable to relax this optimization into a convex program. A popular way to achieve that is to transform the rank constraint into a penalty term $\lambda \| \bfm X \|_*$ that is added to the quadratic objective function, where $\lambda$ is a tuning parameter and $\| \cdot \|_*$ is the nuclear norm (that is, the $\ell_1$ norm of the vector of all its singular values), which encourages a solution with low rank (number of nonzero components in that vector). A rather surprising conclusion from \cite{CRe09} is that in the noiseless setting, solving the relaxed problem yields the same solution as the nonconvex problem with high probability.

We can view this problem from the perspective of factor models. The assumption that $\bfm M^*$ has low rank can be justified by interpreting each $M_{ij}$ as the linear combination of a few latent factors. Indeed, if $M_{ij}$ is the preference score of user $i$ for item $j$, then it is reasonable to posit $M_{ij} = \bfm b_i^\top \bfm f_j$, where $\bfm f_j \in \R^K$ is the features item $j$ possesses and $\bfm b_i \in \R^K$ is the tendency of user $i$ towards the features. In this regard, $\bfm M^* = \bfm B \bfm F^\top$ can be viewed as the part explained by the factors in the factor models.

This discussion motivates us to write our observation as
\begin{equation*}
\mathcal{P}_\Omega(\bfm M) = p \bfm M^* + (\mathcal{P}_\Omega(\bfm M^*) - \E\mathcal{P}_\Omega(\bfm M^*)  + \mathcal{P}_\Omega(\bfm E)), \quad \text{where} ~ \bfm E := (\varepsilon_{ij})_{i,j} \in \R^{n_1 \times n_2},
\end{equation*}
since $\E\mathcal{P}_\Omega(\bfm M^*) = p \bfm M^*$.  This decomposition gives the familiar ``low-rank plus noise" structure. It is natural to conduct PCA on $\mathcal{P}_\Omega(\bfm M)$ to extract the low-rank part.

Let the best rank-$K$ approximation of $\mathcal{P}_\Omega(\bfm M)$ be given by $\bfm U \diag( \sigma_1,\ldots, \sigma_K) \bfm V^\top$, where $\{\sigma_k\}_{k=1}^K$ are the largest $K$ singular values in descending order, and columns of $\bfm U\in \R^{n_1 \times K}, \bfm V \in \R^{n_2 \times K}$ correspond to their normalized left and right singular vectors, respectively. Similarly, we have singular value decomposition $\bfm M^* = \bfm U^* \diag( \sigma_1^*,\ldots, \sigma_K^*) (\bfm V^*)^\top$. The following result from \cite{AFWZ17} provides entry-wise bounds for our estimates. For a matrix, denote by $\| \cdot \|_{\max}$ the largest absolute value of all entries, and $\| \cdot \|_{2\to \infty}$ the largest $\ell_2$ norm of all row vectors.

\begin{thm}\label{thm:MC}
Let $n = n_1 + n_2$, $\eta = \max \{ \| \bfm U \|_{2\to \infty}, \| \bfm V \|_{2\to \infty} \}$ and $\kappa = \sigma_1^*  / \sigma_K^*$. There exist constants $C,C'>0$ and an orthogonal matrix $\bfm R \in \R^{K \times K}$ such that the following holds. If $p \ge 6\log n/n$ and $\kappa \frac{n (\| \bfm M^* \|_{\max} + \sigma )}{\sigma_r^*} \sqrt{\frac{\log n}{np}} \le 1/C$, then with at least probability $1 - C/n$,
\begin{align*}
\max \{ \| \bfm U \bfm R - \bfm U^* \|_{\max} , \| \bfm V \bfm R - \bfm V^* \|_{\max} \} &\le C'\eta \kappa \frac{n (\| \bfm M^* \|_{\max} + \sigma )}{\sigma_r^*} \sqrt{\frac{\log n}{np}}, \\
\| \bfm U \diag\{ \sigma_1,\ldots, \sigma_K\} \bfm V^\top - \bfm M^* \|_{\max} & \le C' \eta^2 \kappa^4( \| \bfm M^* \|_{\max} + \sigma)  \sqrt{\frac{n\log n}{p}}.
\end{align*}
\end{thm}
We can simplify the bounds with a few additional assumptions. If $n_1 \asymp n_2$, then $\eta$ is of order $O(\sqrt{K/n})$ assuming a bounded coherence number. In addition, if $\kappa$ is also bounded, then
\begin{equation*}
\| \bfm U \diag\{ \sigma_1,\ldots, \sigma_K\} \bfm V^\top - \bfm M^* \|_{\max}  \lesssim ( \| \bfm M^* \|_{\max} + \sigma)  \sqrt{\frac{\log n}{np}}.
\end{equation*}
We remark that the requirement on the sample ratio $p \gtrsim \log n/n$ is the weakest condition necessary for matrix completion, which ensures each row and column and sampled with high probability. Also, the entry-wise bound above can recover the Frobenius bound \citep{KMO101} up to a log factor. It is more precise than the Frobenius bound, because the latter only provides control on average error.

\subsection{Synchronization problems}

Synchronization problems are a class of problems in which one estimates signals from their pairwise comparisons. Consider the phase synchronization problem as an example, that is, estimating $n$ angles $\theta_1,\ldots,\theta_n$ from noisy measurements of their differences. We can express an angle $\theta_\ell$ in the equivalent form of a unit-modulus complex number $z_\ell = \exp(i \theta_\ell)$, and thus, the task is to estimate a complex vector $\bfm z = (z_1,\ldots, z_n)^\top \in \mathbb{C}^n$. Suppose our measurements have the form $C_{\ell k} = \bar z_\ell z_k + \sigma w_{\ell k}$, where $\bar z_\ell$ denotes the conjugate of $z_\ell$, and for all $\ell > k$, $w_{\ell k} \in \mathbb{C}$ is i.i.d.\ complex Gaussian variable (namely, the real part and imaginary part of $w_{\ell k}$ are $\cN(0,1/2)$ and independent). Then, the phase of $C_{\ell k}$ (namely $\mathrm{arg}(C_{\ell k})$) encodes the noisy difference $\theta_k - \theta_{\ell}$.

More generally, the goal of a synchronization problem is to estimate $n$ signals from their pairwise measurements, where each signal is an element from a group, e.g., the group of rotations in three dimensions. Synchronization problems are motivated from imaging problems such as cryo-EM \citep{shkolnisky2012viewing}, camera calibration \citep{tron2009distributed}, etc.

Synchronization problems also admit the ``low-rank plus noise" structure. Consider our phase synchronization problem again. If we let $w_{k \ell} = w_{\ell k}$ ($\ell > k$) and $w_{\ell \ell} = 0$, and write $\bfm W = (w_{\ell k})_{\ell,k=1}^n $, then our measurement matrix $\bfm C = (C_{\ell k})_{\ell, k=1}^n$ has the structure
\begin{equation*}
\bfm C  = \bfm z \bfm z^* + \sigma \bfm W,
\end{equation*}
where $^*$ denotes the conjugate transpose. This decomposition has a similar form to \eqref{eq:comdet} in community detection. Note that $\bfm z \bfm z^*$ is a complex matrix with a single nonzero eigenvalue $n$, and $\| \sigma \bfm W \|_2 $ is of order $\sigma \sqrt{n}$ with high probability (which is a basic result in random matrix theory). Therefore, we expect that no estimators can do well if $\sigma \gtrsim \sqrt{n}$. Indeed, the information-theoretic limit is established in \cite{lelarge2016fundamental}. Our next result from \cite{ZhoBou18} gives estimation guarantees if the reverse inequality is true (up to a log factor).

\begin{thm}\label{thm:synch}
Let $\bfm v \in \mathbb{C}^n$ be the leading eigenvector of $\bfm C$ such that $\| \bfm v \|_2 = \sqrt{n}$ and $\bfm v^* \bfm z = |\bfm v^* \bfm z|$. Then, if $\sigma \lesssim \sqrt{n / \log n}$, then with probability $1 - O(n^{-2})$, the relative errors satisfy
\begin{equation*}
n^{-1/2} \| \bfm v - \bfm z \|_2 \lesssim \sigma/\sqrt{n}, \quad \text{and} \quad \| \bfm v - \bfm z \|_\infty \lesssim \sigma\sqrt{\log n/n}.
\end{equation*}
Moreover, the above two inequalities also hold for the maximum likelihood estimator.
\end{thm}
Note that the eigenvector of a complex matrix is not unique: for any $\alpha \in \R$, the vector $e^{i \alpha} \bfm v$ is also an eigenvector, so we fix the global phase $e^{i \alpha}$ by restricting $\bfm v^* \bfm z = |\bfm v^* \bfm z|$. Note also that the maximum likelihood estimator is different from $\bfm v$, because the MLE must satisfy the entry-wise constraint $|z_\ell| = 1$ for any $\ell \in [n]$.
This result implies consistency of $\bfm v$ in terms of both the $\ell_2$ norm and the $\ell_\infty$ norm if $\sigma \ll \sqrt{n / \log n}$, and thus, provides good evidence that spectral methods (or PCA) are simple, generic, yet powerful.

\appendix

\section{Proofs}

\begin{proof}[Proof of Corollary~\ref{cor:dk}]
Notice that the result is trivial if $\delta_0 \le 2\| \tilde{\bA} - \bA \|_2$, since $\| (\tilde \bA - \bA)\bV \|_2 \le \| \tilde \bA - \bA \|_2$ and $\| \sin \bTheta(\tilde \bA, \bA) \|_2 \le 1$ always hold. If $\delta_0 > 2\| \tilde{\bA} - \bA \|_2$, then by Weyl's inequality,
$$
\mathcal{L}(\tilde{\bV}^\bot) \subset (-\infty, \alpha - \delta_0 + \| \tilde \bA - \bA \|_2] \cup [\beta+ \delta_0 -\| \tilde \bA - \bA \|_2, +\infty   ).
$$
Thus, we can set $\delta = \delta_0 - \| \tilde \bA - \bA \|_2$ in Theorem~\ref{thm:dk} and derive
$$
\| \sin \bTheta(\tilde \bV, \bV) \|_2 \le \frac{\| \tilde \bA - \bA \|_2}{\delta_0 - \| \tilde \bA - \bA \|_2} \le \frac{\| \tilde \bA - \bA \|_2}{\delta_0 - \delta_0/2} = 2\delta_0^{-1} \| \tilde \bA - \bA \|_2.
$$
This proves the spectral norm case.
\end{proof}

\begin{proof}[Proof of Theorem~\ref{thm:pert-sym}]
\textbf{Step 1:} First, we derive a few elementary inequalities: for any $m \in [n]$,
\begin{equation}\label{ineq:elementary}
\| \bfm W^{(m)} \|_2 \le \| \bfm W \|_2, \quad \| \bfm w_m \|_2 \le \| \bfm W \|_2, \quad \| \bfm W - \bfm W^{(m)} \|_2 \le 2\| \bfm W \|_2.
\end{equation}
To prove these inequalities, recall the (equivalent) definition of spectral norm for symmetric matrices:
\begin{equation*}
\| \bfm W \|_2 = \max_{\xx, \yy \in S^{n-1}}\xx^\top \bfm W \yy = \max_{\xx \in S^{n-1}} \| \bfm W \xx \|_2,
\end{equation*}
where $S^{n-1}$ is the unit sphere in $\R^{n}$, and $\xx = (x_1,\ldots, x_n)^\top, \yy = (y_1,\ldots, y_n)^\top$. The first and second inequalities follow from
\begin{align*}
\| \bfm W \|_2 &\ge \max\{ \xx^\top \bfm W \yy : \xx, \yy \in S^{n-1}, x_m = y_m = 0 \} = \| \bfm W^{(m)} \|_2, \quad \text{and}\\
 \| \bfm W \|_2 &\ge \max_{\xx \in S^{n-1}} | \langle \bfm w_m, \xx \rangle | = \| \bfm w_m \|_2.
\end{align*}
The third inequality follows from the first one and the triangle inequality.

\textbf{Step 2:} Next, by the definition of eigenvectors,
\begin{equation}\label{eq:key1}
\widetilde{\bfm v}_\ell - \bfm v_\ell = \frac{\widetilde{\bfm A} \widetilde{\bfm v}_\ell  }{\widetilde \lambda_\ell} - \bfm v_\ell =  \left(\frac{\bfm A \widetilde{\bfm v}_\ell  }{\widetilde \lambda_\ell} - \bfm v_\ell \right) + \frac{\bfm W \widetilde{\bfm v}_\ell  }{\widetilde \lambda_\ell}.
\end{equation}
We first control the entries of the first term on the right-hand side. Using the decomposition \eqref{eq:decomp}, we have
\begin{equation}\label{ineq:Avbnd1}
\left[\frac{\bfm A \widetilde{\bfm v}_\ell  }{\widetilde \lambda_\ell} - \bfm v_\ell \right]_m = \left( \frac{\lambda_\ell}{\widetilde \lambda_\ell} \langle \bfm v_\ell, \widetilde{\bfm v}_\ell \rangle - 1\right) [\bfm v_\ell]_m + \sum_{k \neq \ell, k \le K} \frac{\lambda_k}{\widetilde \lambda_\ell}  \langle \bfm v_k, \widetilde{\bfm v}_\ell \rangle [\bfm v_k]_m, \quad \forall \,m \in [n].
\end{equation}
Using the triangle inequality, we have
\begin{equation*}
\left| \frac{\lambda_\ell}{\widetilde \lambda_\ell}\langle \bfm v_\ell, \widetilde{\bfm v}_\ell \rangle - 1 \right| \le \left| \frac{\lambda_\ell}{\widetilde \lambda_\ell}\langle \bfm v_\ell, \widetilde{\bfm v}_\ell \rangle - \langle \bfm v_\ell, \widetilde{\bfm v}_\ell \rangle \right| + \left| \langle \bfm v_\ell, \widetilde{\bfm v}_\ell \rangle - 1 \right| \le \frac{ | \lambda_\ell - \widetilde \lambda_\ell|}{| \widetilde \lambda_\ell |} + \frac{1}{2} \left\|  \widetilde{\bfm v}_\ell -  \bfm v_\ell  \right\|^2.
\end{equation*}
By Weyl's inequality, $| \widetilde \lambda_\ell - \lambda_\ell | \le \| \bfm W \|_2$, and thus $| \widetilde \lambda_\ell |  \ge | \lambda_\ell | - \| \bfm W \|_2 \ge \delta_\ell - \| \bfm W \|_2$. Also, by Corollary~\ref{cor:dk} (simplified Davis-Kahan's theorem) and its following remark, $\| \widetilde{\bfm v}_\ell -  \bfm v_\ell \|_2 \le 2\sqrt{2}\, \| \bfm W \|_2 / \delta_\ell$. Therefore, under the condition $\delta_\ell \ge 2\| \bfm W \|_2$,
\begin{align*}
\left| \frac{\lambda_\ell}{\widetilde \lambda_\ell}\langle \bfm v_\ell, \widetilde{\bfm v}_\ell \rangle - 1 \right| \le \frac{ \| \bfm W \|_2 }{\delta_\ell - \| \bfm W \|_2} + \frac{ 4 \| \bfm W \|_2^2}{\delta_\ell^2} \le \frac{2 \| \bfm W \|_2 }{\delta_\ell} + \frac{2 \| \bfm W \|_2 }{\delta_\ell}= \frac{4 \| \bfm W \|_2 }{\delta_\ell}.
\end{align*}
Using Corollary~\ref{cor:dk} again, we obtain
\begin{equation*}
\sum_{k \neq \ell, k \le K} \frac{\lambda_k^2}{\widetilde \lambda_\ell^2}  \langle \bfm v_k, \widetilde{\bfm v}_\ell \rangle^2 \lesssim \sum_{k \neq \ell, k \le K} \langle \bfm v_k, \widetilde{\bfm v}_\ell \rangle^2 \le 1 - \langle \bfm v_\ell, \widetilde{\bfm v}_\ell \rangle^2 = \sin^2\theta(\bfm v_\ell, \widetilde{\bfm v}_\ell) \le \frac{ 4\| \bfm W \|_2^2 }{\delta_{\ell}^2},
\end{equation*}
where the first inequality is due to $|\widetilde \lambda_\ell| \ge |\lambda_\ell| - \| \bfm W \|_2 \ge 4|\lambda_\ell|/5$ and the condition $|\lambda_\ell| \asymp \max_{k \in [K]} | \lambda_k|$, and the second inequality is due to the fact that $\{\bfm v_k\}_{k=1}^K$ is a subset of orthonormal basis. Now we use the Cauchy-Schwarz inequality to bound the second term on the right-hand side of \eqref{ineq:Avbnd1} and get
\begin{equation}\label{ineq:bnd1}
\left| \Big[\frac{\bfm A \widetilde{\bfm v}_\ell  }{\widetilde \lambda_\ell} - \bfm v_\ell \Big]_m \right|  \lesssim \frac{ \| \bfm W \|_2}{\delta_\ell} \left( \sum_{k=1}^K [\bfm v_k]_m^2 \right)^{1/2}.
\end{equation}

\textbf{Step 3:} To bound the entries of the second term in \eqref{eq:key1}, we use the leave-one-out idea as follows.
\begin{equation}\label{eq:Wv}
[\bfm W \widetilde{\bfm v}_\ell]_m = [\bfm W \widetilde{\bfm v}_\ell^{(m)}]_m +  [\bfm W (  \widetilde{\bfm v}_\ell - \widetilde{\bfm v}_\ell^{(m)}) ]_m = \langle \bfm w_m,  \widetilde{\bfm v}_\ell^{(m)} \rangle + \langle \bfm w_m, \widetilde{\bfm v}_\ell - \widetilde{\bfm v}_\ell^{(m)} \rangle, \quad \forall\, m \in [n].
\end{equation}
%where $\widetilde{\bfm v}_\ell^{(m)}$ is defined in \eqref{dec:vm}, and $\bfm w_m$ is the $m$th column of $\bfm W$.
%Note that $\langle \bfm w_m,  \widetilde{\bfm v}_\ell^{(m)} \rangle $ is the inner products of two independent vectors, so it can be well controlled if $\bfm w_m$ is a subgaussian vector.
We can bound the second term using the Cauchy-Schwarz inequality: $|\langle \bfm w_m, \widetilde{\bfm v}_\ell - \widetilde{\bfm v}_\ell^{(m)} \rangle| \le  \| \bfm w_m \|_2 \| \widetilde{\bfm v}_\ell - \widetilde{\bfm v}_\ell^{(m)}  \|_2$. The crucial observation is that, if we view $\widetilde{\bfm v}_\ell$ as the perturbed version of $\widetilde{\bfm v}_\ell^{(m)}$, then by Theorem~\ref{thm:dk} (Davis-Kahan's theorem) and Weyl's inequality, for any $\ell \in [K]$,
\begin{equation*}
\| \widetilde{\bfm v}_\ell - \widetilde{\bfm v}_\ell^{(m)}  \|_2 \le \frac{\sqrt{2} \| \bfm \Delta^{(m)} \widetilde{\bfm v}_\ell^{(m)} \|_2}{\widetilde \delta_\ell^{(m)} - \| \bfm \Delta^{(m)} \|_2}, \quad \text{where} ~ \bfm \Delta^{(m)} := \bfm W - \bfm W^{(m)}.
\end{equation*}
%Here, $\widetilde \delta_\ell^{(m)} := \min\{ \widetilde \lambda_{\ell-1}^{(m)} - \widetilde \lambda_\ell^{(m)}, \widetilde \lambda_\ell^{(m)} - \widetilde \lambda_{\ell+1}^{(m)} \}$ is the eigen-gap of $\bfm A + \bfm W^{(m)}$, where $ \widetilde \lambda_\ell^{(m)}$ is the eigenvalue corresponding to $ \widetilde{\bfm v}_\ell^{(m)}$ (assuming $\widetilde \lambda_0^{(m)} = +\infty$ and $\widetilde \lambda_{K+1}^{(m)} = -\infty$ as before).
Here, $\widetilde \delta_\ell^{(m)} $ is the eigen-gap of $\bfm A + \bfm W^{(m)}$, and it satisfies $\widetilde \delta_\ell^{(m)} \ge \delta_\ell - 2 \| \bfm W^{(m)} \|_2$ since $| \widetilde \lambda_i^{(m)} - \lambda_i | \le \| \bfm W^{(m)} \|_2 $ for all $i \in [n]$, by Weyl's inequality. By \eqref{ineq:elementary}, we have $\widetilde \delta_\ell^{(m)} - \| \bfm \Delta^{(m)} \|_2 \ge \delta_\ell - 4 \| \bfm W \|_2$. Thus, under the condition $\delta_\ell \ge 5\| \bfm W \|_2$, we have
\begin{equation*}
\| \widetilde{\bfm v}_\ell - \widetilde{\bfm v}_\ell^{(m)}  \|_2 \lesssim \frac{\| \bfm \Delta^{(m)} \widetilde{\bfm v}_\ell^{(m)} \|_2}{\delta_\ell}.
%\le \frac{|\langle \bfm w_m,  \widetilde{\bfm v}_\ell^{(m)} \rangle |  }{\delta_\ell}
\end{equation*}
%where, the first inequality is due to $\delta_\ell - 4 \| \bfm W \|_2 \ge \delta_\ell -  4 \delta_\ell / 5 = \delta_\ell / 5 $, and the second inequality follows from the observation that the $m$th entry of the vector $\bfm \Delta^{(m)} \widetilde{\bfm v}_\ell^{(m)}$ is exactly $\langle \bfm w_m,  \widetilde{\bfm v}_\ell^{(m)} \rangle$, and other
Note that the $m$th entry of the vector $\bfm \Delta^{(m)} \widetilde{\bfm v}_\ell^{(m)}$ is exactly $\langle \bfm w_m,  \widetilde{\bfm v}_\ell^{(m)} \rangle$, and other entries are $W_{im} [\widetilde{\bfm v}_\ell^{(m)}]_m$ where $i \neq m$. Thus,
\begin{equation*}
\| \widetilde{\bfm v}_\ell - \widetilde{\bfm v}_\ell^{(m)}  \|_2 \lesssim \frac{1}{\delta_\ell} \left( \langle \bfm w_m,  \widetilde{\bfm v}_\ell^{(m)} \rangle^2 + \sum_{i\neq m} W_{im}^2 [\widetilde{\bfm v}_\ell^{(m)}]_m^2  \right)^{1/2} \le \frac{1}{\delta_\ell}  \left( | \langle \bfm w_m,  \widetilde{\bfm v}_\ell^{(m)}  \rangle | + \| \bfm w_m \|_2 | [\widetilde{\bfm v}_\ell^{(m)}]_m | \right),
\end{equation*}
where we used $\sqrt{a+b} \le \sqrt{a} + \sqrt{b}$ ($a,b \ge 0$). The above inequality, together with $|\langle \bfm w_m, \widetilde{\bfm v}_\ell - \widetilde{\bfm v}_\ell^{(m)} \rangle| \le  \| \bfm w_m \|_2 \| \widetilde{\bfm v}_\ell - \widetilde{\bfm v}_\ell^{(m)}  \|_2$, leads to a bound on $[\bfm W \widetilde{\bfm v}_\ell]_m$ in \eqref{eq:Wv}.
\begin{align}
\left| [\bfm W \widetilde{\bfm v}_\ell]_m \right| &\lesssim |\langle \bfm w_m,  \widetilde{\bfm v}_\ell^{(m)} \rangle| + \frac{ \| \bfm w_m \|_2 }{\delta_\ell} \left( | \langle \bfm w_m,  \widetilde{\bfm v}_\ell^{(m)}  \rangle | + \| \bfm w_m \|_2 | [\widetilde{\bfm v}_\ell^{(m)}]_m | \right) \notag \\
& \lesssim |\langle \bfm w_m,  \widetilde{\bfm v}_\ell^{(m)} \rangle| + \| \bfm w_m \|_2 | [\widetilde{\bfm v}_\ell^{(m)}]_m |  \label{ineq:Wv}
\end{align}
where we used $\delta_\ell^{-1} \| \bfm w_m \|_2 \le \delta_\ell^{-1} \| \bfm W \|_2  < 1$. We claim that $| [\widetilde{\bfm v}_\ell^{(m)}]_m |  \lesssim  ( \sum_{k=1}^K [ \bfm v_k]_m^2 )^{1/2}$. Once this is proved, combining it with \eqref{ineq:bnd1} and \eqref{ineq:Wv} yields the desired bound on the entries of $\widetilde{\bfm v}_\ell - \bfm v_\ell$ in \eqref{eq:key1}:
\begin{align*}
\left| [ \widetilde{\bfm v}_\ell - \bfm v_\ell ]_m \right| &\lesssim \frac{ \| \bfm W \|_2}{\delta_\ell} \left( \sum_{k=1}^K [\bfm v_k]_m^2 \right)^{1/2} + \frac{1}{\delta_\ell} \left(  |\langle \bfm w_m,  \widetilde{\bfm v}_\ell^{(m)} \rangle| + \| \bfm w_m \|_2 | [\widetilde{\bfm v}_\ell^{(m)}]_m |  \right) \\
&\lesssim \frac{ \| \bfm W \|_2}{\delta_\ell} \left( \sum_{k=1}^K [\bfm v_k]_m^2 \right)^{1/2} + \frac{|\langle \bfm w_m,  \widetilde{\bfm v}_\ell^{(m)} \rangle|}{\delta_\ell},
\end{align*}
where, in the first inequality, we used $|\widetilde \lambda_\ell| \ge |\lambda_\ell| - \| \bfm W \|_2 \ge  \delta_\ell - \delta_\ell / 5 = 4 \delta_\ell / 5$, and in the second inequality, we used $\| \bfm w_m \|_2 \le \| \bfm W \|_2$ and the claim.

\textbf{Step 4:} Finally, we prove our claim that $| [\widetilde{\bfm v}_\ell^{(m)}]_m |  \lesssim  ( \sum_{k=1}^K [ \bfm v_k]_m^2 )^{1/2}$. By definition, $\widetilde \lambda_\ell^{(m)} \widetilde{\bfm v}_\ell^{(m)} = (\bfm A + \bfm W^{(m)}) \widetilde{\bfm v}_\ell^{(m)} $. Note that the $m$th row of $ \bfm W^{(m)} \widetilde{\bfm v}_\ell^{(m)}$ is $0$, since $\bfm W^{(m)}$ has only zeros in its $m$th row. Thus,
\begin{equation*}
 [\widetilde{\bfm v}_\ell^{(m)}]_m = \left( [\widetilde{\bfm v}_\ell^{(m)}]_m- [\bfm v_\ell ]_m \right) + [\bfm v_\ell]_m = \Big[\frac{\bfm A \widetilde{\bfm v}_\ell^{(m)}  }{\widetilde \lambda_\ell^{(m)}} - \bfm v_\ell \Big]_m + [\bfm v_\ell]_m.
\end{equation*}
With an argument similar to the one that leads to \eqref{ineq:bnd1}, we can bound the first term on the right-hand side.
\begin{equation*}
\left| \Big[\frac{\bfm A \widetilde{\bfm v}_\ell^{(m)}  }{\widetilde \lambda_\ell^{(m)}} - \bfm v_\ell \Big]_m \right| \lesssim \frac{ \| \bfm W^{(m)} \|_2}{\delta_\ell} \left( \sum_{k=1}^K [\bfm v_k]_m^2 \right)^{1/2} \le \left( \sum_{k=1}^K [\bfm v_k]_m^2 \right)^{1/2}.
\end{equation*}
Clearly, $| [\bfm v_\ell]_m |$ is also upper bounded by the right-hand side above. This proves our claim and concludes the proof.
\end{proof}

\begin{proof}[Proof of Corollary~\ref{cor:pert-assym}]
Let us construct symmetric matrices $\bfm A, \bfm W, \widetilde{\bfm W}$ of size $n + p$ via a standard \textit{dilation} technique \citep{paulsen2002completely}. Define
\begin{equation*}
\bfm A = \left( \begin{array}{cc}
\mathbf{0} &  \bfm L \\
\bfm L^\top & \mathbf{0}
\end{array}\right), \quad
\bfm W = \left( \begin{array}{cc}
\mathbf{0} &  \bfm E \\
\bfm E^\top & \mathbf{0}
\end{array}\right), ~~ \text{and} ~~ \widetilde{\bfm A} = \bfm A + \bfm W.
\end{equation*}
It can be checked that $\rank(\bfm A) = 2K$, and importantly,
\begin{equation}\label{eq:dilation-decomp}
\bfm A  = \frac{1}{2} \sum_{k=1}^K \sigma_k \left( \begin{array}{c} \bfm u_k \\ \bfm v_k \end{array} \right) \left( \begin{array}{cc} \bfm u_k^\top & \bfm v_k^\top \end{array} \right) - \frac{1}{2} \sum_{k=1}^K \sigma_k \left( \begin{array}{c} \bfm u_k \\ -\bfm v_k \end{array} \right) \left( \begin{array}{cc} \bfm u_k^\top & -\bfm v_k^\top \end{array} \right).
\end{equation}

\textbf{Step 1:} Check the conditions of Theorem~\ref{thm:pert-sym}. The nonzero eigenvalues of $\bfm A$ are $\pm \sigma_k$, ($k \in [K]$), and the corresponding eigenvectors are $(\bfm u_k^\top, \pm \bfm v_k^\top)^\top / \sqrt{2} \in \R^{n + p}$. It is clear that the eigenvalue condition $|\lambda_\ell| \asymp \max_{k \in [K]} | \lambda_k|$ in Theorem~\ref{thm:pert-sym} is satisfied, and the eigen-gap $\delta_\ell$ of $\bfm A$ is exactly $\gamma_\ell$. Since the identity \eqref{eq:dilation-decomp} holds for any matrix constructed from dilation, by applying it to $\bfm W$ we get $\| \bfm W \|_2 = \| \bfm E \|_2$.

\textbf{Step 2:} Apply the conclusion of Theorem~\ref{thm:pert-sym}. Similarly as before, we write $\bfm W^{(m)}$ as the matrix obtained by setting $m$th row and $m$th column of $\bfm W$ to zero, where $m \in [n+p]$. We also denote $\widetilde{\bfm A}^{(m)} = \bfm A + \bfm W^{(m)}$.
Using a similar argument as Step 1, we find
\begin{enumerate}
\item[(1)] the eigenvectors of $\widetilde{\bfm A}$ are $ \displaystyle\binom{\widetilde{\bfm u}_k}{\pm \widetilde{\bfm v}_k}/ \sqrt{2}$,
\item[(2)] the eigenvectors of $\widetilde{\bfm A}^{(i)}$ are $ \displaystyle\binom{*}{\pm \widetilde{\bfm v}^{(i)}_k}/ \sqrt{2}$, $\forall \, i \in[n]$, and
\item[(3)] the eigenvectors of $\widetilde{\bfm A}^{(n+j)}$ are $  \displaystyle\binom{\widetilde{\bfm u}^{(j)}_k}{*}/ \sqrt{2}$, $\forall \, j \in[p]$,
\end{enumerate}
where $*$ means some appropriate vectors we do not need in the proof (we do not bother introducing notations for them). We also observe that
\begin{equation*}
\bfm w_m = \begin{cases}
(\begin{array}{cc} \mathbf{0} & \bfm e_i^\row \end{array})^\top, & m = i \in [n] \\
(\begin{array}{cc} (\bfm e_j^\coln)^\top & \mathbf{0} \end{array})^\top, & m = n+ j, ~ j \in [p]
\end{cases}
\end{equation*}
Note that the inner product between $\bfm w_m$ and the eigenvector of $\widetilde{\bfm A}^{(m)}$ is $\langle (\bfm e_i^\row)^\top,  \pm \widetilde{\bfm v}_k^{(i)} \rangle$ if $m = i \in [n]$, or $\langle \bfm e_j^\coln,  \widetilde{\bfm u}_k^{(j)} \rangle$ if $m = n+j, ~j \in [p]$. Therefore, applying Theorem~\ref{thm:pert-sym} to the first $n$ entries of
\begin{equation*}
\frac{1}{\sqrt{2}} \left( \begin{array}{c} \widetilde{\bfm u}_\ell - \bfm u_\ell \\ \widetilde{\bfm v}_\ell - \bfm v_\ell
\end{array} \right),
\end{equation*}
we obtain the first inequality of Corollary~\ref{cor:pert-assym}, and applying Theorem~\ref{thm:pert-sym} to the last $p$ entries leads to the second inequality.
\end{proof}

\begin{proof}[Proof of Lemma \ref{lem:3.1}]
\begin{align*}
	\E_{\bvarepsilon} [\ltwonorm{\bX^\top \widehat \bbeta_K - \bX^\top \bbeta^*}^2 / n] & = \E_{\bvarepsilon} [\ltwonorm{\bQ_K \bSigma_K \bP_K^\top \bbeta^* + \bQ_K\bQ_K^\top \bvarepsilon - \bX^\top \bbeta^*}^2 / n] \\
    & = \E_{\bvarepsilon} [\ltwonorm{\bQ_K\bQ_K^\top \bvarepsilon - \bQ_{K+}\bSigma_{K+}\bP_{K+}^\top\bbeta^*}^2/n] \\
    & = \frac{K \sigma^2}{n} + \underbrace{{\bbeta^*}^\top \bP_{K+}}_{\balpha ^\top} \bSigma^2_{K+} \underbrace{\bP^\top_{K+} \bbeta^*}_{\balpha}. \\
    & = \frac{K \sigma^2}{n} + \sum\limits_{j=K+1}^d \lambda^2_{j} \alpha^2_j.
\end{align*}
\end{proof}

\bibliography{factor}

\begin{thebibliography}{105}
\expandafter\ifx\csname natexlab\endcsname\relax\def\natexlab#1{#1}\fi
\expandafter\ifx\csname url\endcsname\relax
  \def\url#1{\texttt{#1}}\fi
\expandafter\ifx\csname urlprefix\endcsname\relax\def\urlprefix{URL }\fi

\bibitem[{Abbe(2017)}]{Abb17}
\textsc{Abbe, E.} (2017).
\newblock Community detection and stochastic block models: recent developments.
\newblock \textit{arXiv preprint arXiv:1703.10146} .

\bibitem[{Abbe et~al.(2016)Abbe, Bandeira and Hall}]{ABH16}
\textsc{Abbe, E.}, \textsc{Bandeira, A.~S.} and \textsc{Hall, G.} (2016).
\newblock Exact recovery in the stochastic block model.
\newblock \textit{IEEE Transactions on Information Theory} \textbf{62}
  471--487.

\bibitem[{Abbe et~al.(2017)Abbe, Fan, Wang and Zhong}]{AFWZ17}
\textsc{Abbe, E.}, \textsc{Fan, J.}, \textsc{Wang, K.} and \textsc{Zhong, Y.}
  (2017).
\newblock Entrywise eigenvector analysis of random matrices with low expected
  rank.
\newblock \textit{arXiv preprint arXiv:1709.09565} .

\bibitem[{Abbe and Sandon(2015)}]{ASa15}
\textsc{Abbe, E.} and \textsc{Sandon, C.} (2015).
\newblock Community detection in general stochastic block models: Fundamental
  limits and efficient algorithms for recovery.
\newblock In \textit{Foundations of Computer Science (FOCS), 2015 IEEE 56th
  Annual Symposium on}. IEEE.

\bibitem[{Ahn and Horenstein(2013)}]{ahn2013eigenvalue}
\textsc{Ahn, S.~C.} and \textsc{Horenstein, A.~R.} (2013).
\newblock Eigenvalue ratio test for the number of factors.
\newblock \textit{Econometrica} \textbf{81} 1203--1227.

\bibitem[{Anandkumar et~al.(2014)Anandkumar, Ge, Hsu, Kakade and
  Telgarsky}]{AGH14}
\textsc{Anandkumar, A.}, \textsc{Ge, R.}, \textsc{Hsu, D.}, \textsc{Kakade,
  S.~M.} and \textsc{Telgarsky, M.} (2014).
\newblock Tensor decompositions for learning latent variable models.
\newblock \textit{The Journal of Machine Learning Research} \textbf{15}
  2773--2832.

\bibitem[{Anderson and Amemiya(1988)}]{AAm88}
\textsc{Anderson, T.~W.} and \textsc{Amemiya, Y.} (1988).
\newblock The asymptotic normal distribution of estimators in factor analysis
  under general conditions.
\newblock \textit{The Annals of Statistics} \textbf{16} 759--771.

\bibitem[{Bai and Li(2012)}]{BLi12}
\textsc{Bai, J.} and \textsc{Li, K.} (2012).
\newblock Statistical analysis of factor models of high dimension.
\newblock \textit{The Annals of Statistics} \textbf{40} 436--465.

\bibitem[{Bai and Ng(2002)}]{BNg02}
\textsc{Bai, J.} and \textsc{Ng, S.} (2002).
\newblock Determining the number of factors in approximate factor models.
\newblock \textit{Econometrica} \textbf{70} 191--221.

\bibitem[{Baik et~al.(2005)Baik, Ben~Arous and P{\'e}ch{\'e}}]{BBP05}
\textsc{Baik, J.}, \textsc{Ben~Arous, G.} and \textsc{P{\'e}ch{\'e}, S.}
  (2005).
\newblock Phase transition of the largest eigenvalue for nonnull complex sample
  covariance matrices.
\newblock \textit{Annals of Probability}  1643--1697.

\bibitem[{Bartlett(1938)}]{Bar38}
\textsc{Bartlett, M.~S.} (1938).
\newblock Methods of estimating mental factors.
\newblock \textit{Nature} \textbf{141} 609--610.

\bibitem[{Bartlett(1950)}]{bartlett1950}
\textsc{Bartlett, M.~S.} (1950).
\newblock Tests of significance in factor analysis.
\newblock \textit{British Journal of Mathematical and Statistical Psychology}
  \textbf{3} 77--85.

\bibitem[{Bean et~al.(2013)Bean, Bickel, El~Karoui and Yu}]{bean2013optimal}
\textsc{Bean, D.}, \textsc{Bickel, P.~J.}, \textsc{El~Karoui, N.} and
  \textsc{Yu, B.} (2013).
\newblock Optimal {M}-estimation in high-dimensional regression.
\newblock \textit{Proceedings of the National Academy of Sciences} \textbf{110}
  14563--14568.

\bibitem[{Benaych-Georges and Nadakuditi(2011)}]{BenFloRaj11}
\textsc{Benaych-Georges, F.} and \textsc{Nadakuditi, R.~R.} (2011).
\newblock The eigenvalues and eigenvectors of finite, low rank perturbations of
  large random matrices.
\newblock \textit{Advances in Mathematics} \textbf{227} 494--521.

\bibitem[{Benjamini and Hochberg(1995)}]{BHo95}
\textsc{Benjamini, Y.} and \textsc{Hochberg, Y.} (1995).
\newblock Controlling the false discovery rate: a practical and powerful
  approach to multiple testing.
\newblock \textit{Journal of the royal statistical society. Series B
  (Methodological)}  289--300.

\bibitem[{Bickel and Levina(2008)}]{BLe08}
\textsc{Bickel, P.~J.} and \textsc{Levina, E.} (2008).
\newblock Covariance regularization by thresholding.
\newblock \textit{The Annals of Statistics} \textbf{36} 2577--2604.

\bibitem[{Bickel et~al.(2009)Bickel, Ritov and Tsybakov}]{BRT09}
\textsc{Bickel, P.~J.}, \textsc{Ritov, Y.} and \textsc{Tsybakov, A.~B.} (2009).
\newblock Simultaneous analysis of lasso and dantzig selector.
\newblock \textit{The Annals of Statistics} \textbf{37} 1705--1732.

\bibitem[{Cai and Liu(2011)}]{CLi11}
\textsc{Cai, T.} and \textsc{Liu, W.} (2011).
\newblock Adaptive thresholding for sparse covariance matrix estimation.
\newblock \textit{Journal of the American Statistical Association} \textbf{106}
  672--684.

\bibitem[{Candes and Tao(2007)}]{CTa07}
\textsc{Candes, E.} and \textsc{Tao, T.} (2007).
\newblock The dantzig selector: Statistical estimation when p is much larger
  than n.
\newblock \textit{The Annals of Statistics} \textbf{35} 2313--2351.

\bibitem[{Cand{\`e}s et~al.(2011)Cand{\`e}s, Li, Ma and Wright}]{Can11}
\textsc{Cand{\`e}s, E.~J.}, \textsc{Li, X.}, \textsc{Ma, Y.} and
  \textsc{Wright, J.} (2011).
\newblock Robust principal component analysis?
\newblock \textit{Journal of the ACM (JACM)} \textbf{58} 11.

\bibitem[{Cand{\`e}s and Recht(2009)}]{CRe09}
\textsc{Cand{\`e}s, E.~J.} and \textsc{Recht, B.} (2009).
\newblock Exact matrix completion via convex optimization.
\newblock \textit{Foundations of Computational Mathematics} \textbf{9} 717.

\bibitem[{Cape et~al.(2017)Cape, Tang and Priebe}]{CTP17}
\textsc{Cape, J.}, \textsc{Tang, M.} and \textsc{Priebe, C.~E.} (2017).
\newblock The two-to-infinity norm and singular subspace geometry with
  applications to high-dimensional statistics.
\newblock \textit{arXiv preprint arXiv:1705.10735} .

\bibitem[{Catoni(2012)}]{Cat12}
\textsc{Catoni, O.} (2012).
\newblock Challenging the empirical mean and empirical variance: a deviation
  study.
\newblock In \textit{Annales de l'Institut Henri Poincar{\'e}, Probabilit{\'e}s
  et Statistiques}, vol.~48. Institut Henri Poincar{\'e}.

\bibitem[{Cattell(1966)}]{cattell1966scree}
\textsc{Cattell, R.~B.} (1966).
\newblock The scree test for the number of factors.
\newblock \textit{Multivariate behavioral research} \textbf{1} 245--276.

\bibitem[{Chamberlain and Rothschild(1982)}]{CRo82}
\textsc{Chamberlain, G.} and \textsc{Rothschild, M.} (1982).
\newblock Arbitrage, factor structure, and mean-variance analysis on large
  asset markets.

\bibitem[{Cohen et~al.(2015)Cohen, Nelson and Woodruff}]{CNW15}
\textsc{Cohen, M.~B.}, \textsc{Nelson, J.} and \textsc{Woodruff, D.~P.} (2015).
\newblock Optimal approximate matrix product in terms of stable rank.
\newblock \textit{arXiv preprint arXiv:1507.02268} .

\bibitem[{Davis and Kahan(1970)}]{DKa70}
\textsc{Davis, C.} and \textsc{Kahan, W.~M.} (1970).
\newblock The rotation of eigenvectors by a perturbation. iii.
\newblock \textit{SIAM Journal on Numerical Analysis} \textbf{7} 1--46.

\bibitem[{Desai and Storey(2012)}]{DSt12}
\textsc{Desai, K.~H.} and \textsc{Storey, J.~D.} (2012).
\newblock Cross-dimensional inference of dependent high-dimensional data.
\newblock \textit{Journal of the American Statistical Association} \textbf{107}
  135--151.

\bibitem[{Dobriban(2017)}]{dobriban2017factor}
\textsc{Dobriban, E.} (2017).
\newblock Factor selection by permutation.
\newblock \textit{arXiv preprint arXiv:1710.00479} .

\bibitem[{Donoho et~al.(2013)Donoho, Gavish and Johnstone}]{DGJ13}
\textsc{Donoho, D.~L.}, \textsc{Gavish, M.} and \textsc{Johnstone, I.~M.}
  (2013).
\newblock Optimal shrinkage of eigenvalues in the spiked covariance model.
\newblock \textit{arXiv preprint arXiv:1311.0851} .

\bibitem[{Efron(2007)}]{efron2007correlation}
\textsc{Efron, B.} (2007).
\newblock Correlation and large-scale simultaneous significance testing.
\newblock \textit{Journal of the American Statistical Association} \textbf{102}
  93--103.

\bibitem[{Efron(2010)}]{efron2010correlated}
\textsc{Efron, B.} (2010).
\newblock Correlated z-values and the accuracy of large-scale statistical
  estimates.
\newblock \textit{Journal of the American Statistical Association} \textbf{105}
  1042--1055.

\bibitem[{Eldridge et~al.(2017)Eldridge, Belkin and Wang}]{EldBelWan17}
\textsc{Eldridge, J.}, \textsc{Belkin, M.} and \textsc{Wang, Y.} (2017).
\newblock Unperturbed: spectral analysis beyond {D}avis-{K}ahan.
\newblock \textit{arXiv preprint arXiv:1706.06516} .

\bibitem[{Fama and French(1993)}]{FFr93}
\textsc{Fama, E.~F.} and \textsc{French, K.~R.} (1993).
\newblock Common risk factors in the returns on stocks and bonds.
\newblock \textit{Journal of financial economics} \textbf{33} 3--56.

\bibitem[{Fan et~al.(2008)Fan, Fan and Lv}]{FFL08}
\textsc{Fan, J.}, \textsc{Fan, Y.} and \textsc{Lv, J.} (2008).
\newblock High dimensional covariance matrix estimation using a factor model.
\newblock \textit{Journal of Econometrics} \textbf{147} 186--197.

\bibitem[{Fan and Han(2017)}]{FHa17}
\textsc{Fan, J.} and \textsc{Han, X.} (2017).
\newblock Estimation of the false discovery proportion with unknown dependence.
\newblock \textit{Journal of the Royal Statistical Society: Series B
  (Statistical Methodology)} \textbf{79} 1143--1164.

\bibitem[{Fan et~al.(2012)Fan, Han and Gu}]{FHG12}
\textsc{Fan, J.}, \textsc{Han, X.} and \textsc{Gu, W.} (2012).
\newblock Estimating false discovery proportion under arbitrary covariance
  dependence.
\newblock \textit{Journal of the American Statistical Association} \textbf{107}
  1019--1035.

\bibitem[{Fan et~al.(2017{\natexlab{a}})Fan, Ke, Sun and Zhou}]{FKS17}
\textsc{Fan, J.}, \textsc{Ke, Y.}, \textsc{Sun, Q.} and \textsc{Zhou, W.-X.}
  (2017{\natexlab{a}}).
\newblock Farm-test: Factor-adjusted robust multiple testing with false
  discovery control.
\newblock \textit{arXiv preprint arXiv:1711.05386} .

\bibitem[{Fan et~al.(2016{\natexlab{a}})Fan, Ke and Wang}]{FKW16}
\textsc{Fan, J.}, \textsc{Ke, Y.} and \textsc{Wang, K.} (2016{\natexlab{a}}).
\newblock Decorrelation of covariates for high dimensional sparse regression.
\newblock \textit{arXiv preprint arXiv:1612.08490} .

\bibitem[{Fan et~al.(2017{\natexlab{b}})Fan, Li and Wang}]{FLW17}
\textsc{Fan, J.}, \textsc{Li, Q.} and \textsc{Wang, Y.} (2017{\natexlab{b}}).
\newblock Estimation of high dimensional mean regression in the absence of
  symmetry and light tail assumptions.
\newblock \textit{Journal of the Royal Statistical Society: Series B
  (Statistical Methodology)} \textbf{79} 247--265.

\bibitem[{Fan and Li(2001)}]{FLi01}
\textsc{Fan, J.} and \textsc{Li, R.} (2001).
\newblock Variable selection via nonconcave penalized likelihood and its oracle
  properties.
\newblock \textit{Journal of the American statistical Association} \textbf{96}
  1348--1360.

\bibitem[{Fan et~al.(2011)Fan, Liao and Mincheva}]{FLM11}
\textsc{Fan, J.}, \textsc{Liao, Y.} and \textsc{Mincheva, M.} (2011).
\newblock High-dimensional covariance matrix estimation in approximate factor
  models.
\newblock \textit{The Annals of Statistics} \textbf{39} 3320--3356.

\bibitem[{Fan et~al.(2013)Fan, Liao and Mincheva}]{FLM13}
\textsc{Fan, J.}, \textsc{Liao, Y.} and \textsc{Mincheva, M.} (2013).
\newblock Large covariance estimation by thresholding principal orthogonal
  complements.
\newblock \textit{Journal of the Royal Statistical Society: Series B
  (Statistical Methodology)} \textbf{75} 603--680.

\bibitem[{Fan et~al.(2018{\natexlab{a}})Fan, Liu and Wang}]{FLW18}
\textsc{Fan, J.}, \textsc{Liu, H.} and \textsc{Wang, W.} (2018{\natexlab{a}}).
\newblock Large covariance estimation through elliptical factor models.
\newblock \textit{Annals of Statistics} \textbf{46} 1383--1414.

\bibitem[{Fan et~al.(2018{\natexlab{b}})Fan, Wang and Zhong}]{FWZ16}
\textsc{Fan, J.}, \textsc{Wang, W.} and \textsc{Zhong, Y.}
  (2018{\natexlab{b}}).
\newblock An $\ell_{\infty}$ eigenvector perturbation bound and its
  application.
\newblock \textit{Journal of Machine Learning Research} \textbf{18} 1--42.

\bibitem[{Fan et~al.(2016{\natexlab{b}})Fan, Wang and Zhu}]{FWZhu16}
\textsc{Fan, J.}, \textsc{Wang, W.} and \textsc{Zhu, Z.} (2016{\natexlab{b}}).
\newblock A shrinkage principle for heavy-tailed data: High-dimensional robust
  low-rank matrix recovery.
\newblock \textit{arXiv preprint arXiv:1603.08315} .

\bibitem[{Friguet et~al.(2009)Friguet, Kloareg and Causeur}]{FKC09}
\textsc{Friguet, C.}, \textsc{Kloareg, M.} and \textsc{Causeur, D.} (2009).
\newblock A factor model approach to multiple testing under dependence.
\newblock \textit{Journal of the American Statistical Association} \textbf{104}
  1406--1415.

\bibitem[{Gao et~al.(2015)Gao, Ma, Zhang and Zhou}]{Gao15}
\textsc{Gao, C.}, \textsc{Ma, Z.}, \textsc{Zhang, A.~Y.} and \textsc{Zhou,
  H.~H.} (2015).
\newblock Achieving optimal misclassification proportion in stochastic block
  model.
\newblock \textit{arXiv preprint arXiv:1505.03772} .

\bibitem[{Hirzel et~al.(2002)Hirzel, Hausser, Chessel and Perrin}]{HHC02}
\textsc{Hirzel, A.~H.}, \textsc{Hausser, J.}, \textsc{Chessel, D.} and
  \textsc{Perrin, N.} (2002).
\newblock Ecological-niche factor analysis: how to compute habitat-suitability
  maps without absence data?
\newblock \textit{Ecology} \textbf{83} 2027--2036.

\bibitem[{Hochreiter et~al.(2006)Hochreiter, Clevert and Obermayer}]{HCO06}
\textsc{Hochreiter, S.}, \textsc{Clevert, D.-A.} and \textsc{Obermayer, K.}
  (2006).
\newblock A new summarization method for affymetrix probe level data.
\newblock \textit{Bioinformatics} \textbf{22} 943--949.

\bibitem[{Holland et~al.(1983)Holland, Laskey and Leinhardt}]{Hol83}
\textsc{Holland, P.~W.}, \textsc{Laskey, K.~B.} and \textsc{Leinhardt, S.}
  (1983).
\newblock Stochastic blockmodels: First steps.
\newblock \textit{Social networks} \textbf{5} 109--137.

\bibitem[{Horn(1965)}]{horn1965rationale}
\textsc{Horn, J.~L.} (1965).
\newblock A rationale and test for the number of factors in factor analysis.
\newblock \textit{Psychometrika} \textbf{30} 179--185.

\bibitem[{Hotelling(1933)}]{Hot33}
\textsc{Hotelling, H.} (1933).
\newblock Analysis of a complex of statistical variables into principal
  components.
\newblock \textit{Journal of educational psychology} \textbf{24} 417.

\bibitem[{Hsu and Kakade(2013)}]{HKa13}
\textsc{Hsu, D.} and \textsc{Kakade, S.~M.} (2013).
\newblock Learning mixtures of spherical gaussians: moment methods and spectral
  decompositions.
\newblock In \textit{Proceedings of the 4th conference on Innovations in
  Theoretical Computer Science}. ACM.

\bibitem[{Huber(1964)}]{Hub64}
\textsc{Huber, P.~J.} (1964).
\newblock Robust estimation of a location parameter.
\newblock \textit{The annals of mathematical statistics}  73--101.

\bibitem[{Jin(2015)}]{Jin15}
\textsc{Jin, J.} (2015).
\newblock Fast community detection by score.
\newblock \textit{The Annals of Statistics} \textbf{43} 57--89.

\bibitem[{Johnstone and Lu(2009)}]{johnstone2009consistency}
\textsc{Johnstone, I.~M.} and \textsc{Lu, A.~Y.} (2009).
\newblock On consistency and sparsity for principal components analysis in high
  dimensions.
\newblock \textit{Journal of the American Statistical Association} \textbf{104}
  682--693.

\bibitem[{Jolliffe(1986)}]{Jol86}
\textsc{Jolliffe, I.~T.} (1986).
\newblock Principal component analysis and factor analysis.
\newblock In \textit{Principal component analysis}. Springer, 115--128.

\bibitem[{Kendall(1965)}]{Ken65}
\textsc{Kendall, M.~G.} (1965).
\newblock A course in multivariate analysis .

\bibitem[{Keshavan et~al.(2010)Keshavan, Montanari and Oh}]{KMO101}
\textsc{Keshavan, R.~H.}, \textsc{Montanari, A.} and \textsc{Oh, S.} (2010).
\newblock Matrix completion from noisy entries.
\newblock \textit{Journal of Machine Learning Research} \textbf{11} 2057--2078.

\bibitem[{Kneip and Sarda(2011)}]{KSa11}
\textsc{Kneip, A.} and \textsc{Sarda, P.} (2011).
\newblock Factor models and variable selection in high-dimensional regression
  analysis.
\newblock \textit{The Annals of Statistics} \textbf{39} 2410--2447.

\bibitem[{Koltchinskii and Lounici(2017)}]{KLo17}
\textsc{Koltchinskii, V.} and \textsc{Lounici, K.} (2017).
\newblock Concentration inequalities and moment bounds for sample covariance
  operators.
\newblock \textit{Bernoulli} \textbf{23} 110--133.

\bibitem[{Koltchinskii and Xia(2016)}]{KolXia16}
\textsc{Koltchinskii, V.} and \textsc{Xia, D.} (2016).
\newblock Perturbation of linear forms of singular vectors under gaussian
  noise.
\newblock In \textit{High Dimensional Probability VII}. Springer, 397--423.

\bibitem[{Lam and Yao(2012)}]{LamYao12}
\textsc{Lam, C.} and \textsc{Yao, Q.} (2012).
\newblock Factor modeling for high-dimensional time series: inference for the
  number of factors.
\newblock \textit{The Annals of Statistics} \textbf{40} 694--726.

\bibitem[{Lawley and Maxwell(1962)}]{LMa62}
\textsc{Lawley, D.} and \textsc{Maxwell, A.} (1962).
\newblock Factor analysis as a statistical method.
\newblock \textit{Journal of the Royal Statistical Society. Series D (The
  Statistician)} \textbf{12} 209--229.

\bibitem[{Leek and Storey(2008)}]{LSt08}
\textsc{Leek, J.~T.} and \textsc{Storey, J.~D.} (2008).
\newblock A general framework for multiple testing dependence.
\newblock \textit{Proceedings of the National Academy of Sciences} \textbf{105}
  18718--18723.

\bibitem[{Lelarge and Miolane(2016)}]{lelarge2016fundamental}
\textsc{Lelarge, M.} and \textsc{Miolane, L.} (2016).
\newblock Fundamental limits of symmetric low-rank matrix estimation.
\newblock \textit{arXiv preprint arXiv:1611.03888} .

\bibitem[{Li et~al.(2017)Li, Cheng, Fan and Wang}]{LCF17}
\textsc{Li, Q.}, \textsc{Cheng, G.}, \textsc{Fan, J.} and \textsc{Wang, Y.}
  (2017).
\newblock Embracing the blessing of dimensionality in factor models.
\newblock \textit{Journal of the American Statistical Association}  1--10.

\bibitem[{McCrae and John(1992)}]{MJo92}
\textsc{McCrae, R.~R.} and \textsc{John, O.~P.} (1992).
\newblock An introduction to the five-factor model and its applications.
\newblock \textit{Journal of personality} \textbf{60} 175--215.

\bibitem[{Minsker(2016)}]{Min16}
\textsc{Minsker, S.} (2016).
\newblock Sub-gaussian estimators of the mean of a random matrix with
  heavy-tailed entries.
\newblock \textit{arXiv preprint arXiv:1605.07129} .

\bibitem[{Mor-Yosef and Avron(2018)}]{MAv18}
\textsc{Mor-Yosef, L.} and \textsc{Avron, H.} (2018).
\newblock Sketching for principal component regression.
\newblock \textit{arXiv preprint arXiv:1803.02661} .

\bibitem[{Ng et~al.(2002)Ng, Jordan and Weiss}]{Ng02}
\textsc{Ng, A.~Y.}, \textsc{Jordan, M.~I.} and \textsc{Weiss, Y.} (2002).
\newblock On spectral clustering: Analysis and an algorithm.
\newblock In \textit{Advances in neural information processing systems}.

\bibitem[{Onatski(2010)}]{onatski2010determining}
\textsc{Onatski, A.} (2010).
\newblock Determining the number of factors from empirical distribution of
  eigenvalues.
\newblock \textit{The Review of Economics and Statistics} \textbf{92}
  1004--1016.

\bibitem[{Onatski(2012)}]{Onatski12}
\textsc{Onatski, A.} (2012).
\newblock Asymptotics of the principal components estimator of large factor
  models with weakly influential factors.
\newblock \textit{Journal of Econometrics} \textbf{168} 244--258.

\bibitem[{O'Rourke et~al.(2016)O'Rourke, Vu and Wang}]{OVuWan16}
\textsc{O'Rourke, S.}, \textsc{Vu, V.} and \textsc{Wang, K.} (2016).
\newblock Eigenvectors of random matrices: a survey.
\newblock \textit{Journal of Combinatorial Theory, Series A} \textbf{144}
  361--442.

\bibitem[{O'Rourke et~al.(2017)O'Rourke, Vu and Wang}]{o2017random}
\textsc{O'Rourke, S.}, \textsc{Vu, V.} and \textsc{Wang, K.} (2017).
\newblock Random perturbation of low rank matrices: Improving classical bounds.
\newblock \textit{Linear Algebra and its Applications} .

\bibitem[{Paul(2007)}]{Paul07}
\textsc{Paul, D.} (2007).
\newblock Asymptotics of sample eigenstructure for a large dimensional spiked
  covariance model.
\newblock \textit{Statistica Sinica}  1617--1642.

\bibitem[{Paul et~al.(2008)Paul, Bair, Hastie and Tibshirani}]{PBH08}
\textsc{Paul, D.}, \textsc{Bair, E.}, \textsc{Hastie, T.} and
  \textsc{Tibshirani, R.} (2008).
\newblock " preconditioning" for feature selection and regression in
  high-dimensional problems.
\newblock \textit{The Annals of Statistics}  1595--1618.

\bibitem[{Paulsen(2002)}]{paulsen2002completely}
\textsc{Paulsen, V.} (2002).
\newblock \textit{Completely bounded maps and operator algebras}, vol.~78.
\newblock Cambridge University Press.

\bibitem[{Pearson(1901)}]{Pea01}
\textsc{Pearson, K.} (1901).
\newblock Principal components analysis.
\newblock \textit{The London, Edinburgh and Dublin Philosophical Magazine and
  Journal} \textbf{6} 566.

\bibitem[{Rohe et~al.(2011)Rohe, Chatterjee and Yu}]{RohChaYu11}
\textsc{Rohe, K.}, \textsc{Chatterjee, S.} and \textsc{Yu, B.} (2011).
\newblock Spectral clustering and the high-dimensional stochastic blockmodel.
\newblock \textit{The Annals of Statistics} \textbf{39} 1878--1915.

\bibitem[{Sedghi et~al.(2016)Sedghi, Janzamin and Anandkumar}]{SedJanAna16}
\textsc{Sedghi, H.}, \textsc{Janzamin, M.} and \textsc{Anandkumar, A.} (2016).
\newblock Provable tensor methods for learning mixtures of generalized linear
  models.
\newblock In \textit{Artificial Intelligence and Statistics}.

\bibitem[{Shkolnisky and Singer(2012)}]{shkolnisky2012viewing}
\textsc{Shkolnisky, Y.} and \textsc{Singer, A.} (2012).
\newblock Viewing direction estimation in cryo-{EM} using synchronization.
\newblock \textit{SIAM journal on imaging sciences} \textbf{5} 1088--1110.

\bibitem[{Spearman(1927)}]{Spe27}
\textsc{Spearman, C.} (1927).
\newblock The abilities of man. .

\bibitem[{Srivastava and Vershynin(2013)}]{SriVer13}
\textsc{Srivastava, N.} and \textsc{Vershynin, R.} (2013).
\newblock Covariance estimation for distributions with $2+\varepsilon$ moments.
\newblock \textit{The Annals of Probability} \textbf{41} 3081--3111.

\bibitem[{Stewart and Sun(1990)}]{SSu90}
\textsc{Stewart, G.} and \textsc{Sun, J.} (1990).
\newblock Matrix perturbation theory .

\bibitem[{Stock and Watson(2002)}]{SWa02}
\textsc{Stock, J.~H.} and \textsc{Watson, M.~W.} (2002).
\newblock Forecasting using principal components from a large number of
  predictors.
\newblock \textit{Journal of the American statistical association} \textbf{97}
  1167--1179.

\bibitem[{Storey(2002)}]{Sto02}
\textsc{Storey, J.~D.} (2002).
\newblock A direct approach to false discovery rates.
\newblock \textit{Journal of the Royal Statistical Society: Series B
  (Statistical Methodology)} \textbf{64} 479--498.

\bibitem[{Tibshirani(1996)}]{Tib96}
\textsc{Tibshirani, R.} (1996).
\newblock Regression shrinkage and selection via the lasso.
\newblock \textit{Journal of the Royal Statistical Society. Series B
  (Methodological)}  267--288.

\bibitem[{Tron and Vidal(2009)}]{tron2009distributed}
\textsc{Tron, R.} and \textsc{Vidal, R.} (2009).
\newblock Distributed image-based 3-{D} localization of camera sensor networks.
\newblock In \textit{Decision and Control, 2009 held jointly with the 2009 28th
  Chinese Control Conference. CDC/CCC 2009. Proceedings of the 48th IEEE
  Conference on}. IEEE.

\bibitem[{Tropp(2012)}]{Tro12}
\textsc{Tropp, J.~A.} (2012).
\newblock User-friendly tail bounds for sums of random matrices.
\newblock \textit{Foundations of computational mathematics} \textbf{12}
  389--434.

\bibitem[{Vershynin(2010)}]{Ver10}
\textsc{Vershynin, R.} (2010).
\newblock Introduction to the non-asymptotic analysis of random matrices.
\newblock \textit{arXiv preprint arXiv:1011.3027} .

\bibitem[{Vershynin(2012)}]{Ver12}
\textsc{Vershynin, R.} (2012).
\newblock How close is the sample covariance matrix to the actual covariance
  matrix?
\newblock \textit{Journal of Theoretical Probability} \textbf{25} 655--686.

\bibitem[{Wang(2012)}]{Wan12}
\textsc{Wang, H.} (2012).
\newblock Factor profiled sure independence screening.
\newblock \textit{Biometrika} \textbf{99} 15--28.

\bibitem[{Wang et~al.(2017)Wang, Zhao, Hastie and Owen}]{WZH17}
\textsc{Wang, J.}, \textsc{Zhao, Q.}, \textsc{Hastie, T.} and \textsc{Owen,
  A.~B.} (2017).
\newblock Confounder adjustment in multiple hypothesis testing.
\newblock \textit{The Annals of Statistics} \textbf{45} 1863--1894.

\bibitem[{Wang and Fan(2017)}]{WFa17}
\textsc{Wang, W.} and \textsc{Fan, J.} (2017).
\newblock Asymptotics of empirical eigenstructure for high dimensional spiked
  covariance.
\newblock \textit{Ann. Statist.} \textbf{45} 1342--1374.

\bibitem[{Wedin(1972)}]{Wed72}
\textsc{Wedin, P.-A.} (1972).
\newblock Perturbation bounds in connection with singular value decomposition.
\newblock \textit{BIT Numerical Mathematics} \textbf{12} 99--111.

\bibitem[{Woodruff(2014)}]{Woo14}
\textsc{Woodruff, D.~P.} (2014).
\newblock Sketching as a tool for numerical linear algebra.
\newblock \textit{Foundations and Trends{\textregistered} in Theoretical
  Computer Science} \textbf{10} 1--157.

\bibitem[{Yang et~al.(2016)Yang, Meng and Mahoney}]{YMM16}
\textsc{Yang, J.}, \textsc{Meng, X.} and \textsc{Mahoney, M.~W.} (2016).
\newblock Implementing randomized matrix algorithms in parallel and distributed
  environments.
\newblock \textit{Proceedings of the IEEE} \textbf{104} 58--92.

\bibitem[{Yi et~al.(2016)Yi, Caramanis and Sanghavi}]{YiCarSan16}
\textsc{Yi, X.}, \textsc{Caramanis, C.} and \textsc{Sanghavi, S.} (2016).
\newblock Solving a mixture of many random linear equations by tensor
  decomposition and alternating minimization.
\newblock \textit{arXiv preprint arXiv:1608.05749} .

\bibitem[{Yu et~al.(2014)Yu, Wang and Samworth}]{YWS14}
\textsc{Yu, Y.}, \textsc{Wang, T.} and \textsc{Samworth, R.~J.} (2014).
\newblock A useful variant of the davis--kahan theorem for statisticians.
\newblock \textit{Biometrika} \textbf{102} 315--323.

\bibitem[{Zhao and Yu(2006)}]{ZYu06}
\textsc{Zhao, P.} and \textsc{Yu, B.} (2006).
\newblock On model selection consistency of lasso.
\newblock \textit{Journal of Machine learning research} \textbf{7} 2541--2563.

\bibitem[{Zhong(2017)}]{Zho17}
\textsc{Zhong, Y.} (2017).
\newblock Eigenvector under random perturbation: A nonasymptotic
  {R}ayleigh-{S}chr\"{o} dinger theory.
\newblock \textit{arXiv preprint arXiv:1702.00139} .

\bibitem[{Zhong and Boumal(2018)}]{ZhoBou18}
\textsc{Zhong, Y.} and \textsc{Boumal, N.} (2018).
\newblock Near-optimal bounds for phase synchronization.
\newblock \textit{SIAM Journal on Optimization} \textbf{28} 989--1016.

\bibitem[{Zou and Hastie(2005)}]{ZHa05}
\textsc{Zou, H.} and \textsc{Hastie, T.} (2005).
\newblock Regularization and variable selection via the elastic net.
\newblock \textit{Journal of the Royal Statistical Society: Series B
  (Statistical Methodology)} \textbf{67} 301--320.

\end{thebibliography}
\bibliographystyle{ims}

\end{document}